\documentclass[onecolumn,pra,nofootinbib]{revtex4-2}

\usepackage[dvips]{graphicx} 
\usepackage{amsfonts}
\usepackage{amssymb}
\usepackage{amscd}
\usepackage{amsmath}    
\usepackage{enumerate}
\usepackage{epsfig}
\usepackage{subfigure}
\usepackage{subfloat}
\usepackage{xcolor}
\usepackage{amsthm}
\usepackage{physics}
\usepackage{multirow}
\usepackage[most]{tcolorbox}
\usepackage{tikz}
\usetikzlibrary{quantikz}
\usetikzlibrary{backgrounds}
\usepackage{tabularx}
\usepackage{float}
\usepackage{appendix}
\usepackage{makecell}
\usepackage{algorithm}
\usepackage{algorithmicx}
\usepackage{algpseudocode}
\usepackage{dsfont}
\usepackage[colorlinks, linkcolor=red, anchorcolor=blue, citecolor=green]{hyperref}
\usepackage{MnSymbol}
\newtcolorbox[auto counter]{mybox}[2]{
enhanced,
breakable,
label=#1,
colback=blue!5!white,
colframe=blue!75!black,
fonttitle=\bfseries,
title=Box \thetcbcounter: #2
}

\newtheorem{theorem}{Theorem}
\newtheorem{lemma}{Lemma}
\newtheorem{corollary}{Corollary}

\newtheorem{definition}{Definition}
\newtheorem{proposition}{Proposition}

\newcommand{\id}{\mathbb{I}}
\newcommand{\wt}{\mathrm{wt}}
\newcommand{\wts}{\mathrm{wt_s}}

\newcommand{\red}[1]{{\color{red} #1}}
\newcommand{\blue}[1]{{\color{blue} #1}}

\newcommand{\comments}[1]{}
\definecolor{mintcyan}{RGB}{180, 240, 220}

\begin{document}
\title{State complexity and phase identification in adaptive quantum circuits}
\author{Guoding Liu}
\thanks{Both authors contributed equally to this work}
\author{Junjie Chen}
\thanks{Both authors contributed equally to this work}
\author{Xiongfeng Ma}
\email{xma@tsinghua.edu.cn}
\affiliation{Center for Quantum Information, Institute for Interdisciplinary Information Sciences, Tsinghua University, Beijing, 100084 China}

\begin{abstract}
Adaptive quantum circuits, leveraging measurements and classical feedback, significantly expand the landscape of realizable quantum states compared to their non-adaptive counterparts, enabling the preparation of long-range entangled states and topological phases at constant depths. However, the ancilla overhead for preparing arbitrary states can be prohibitive, raising a fundamental question: which states can be efficiently realized with limited ancilla and low depth? Addressing this question requires a rigorous quantitative characterization of state complexity, or the minimum depth and ancillas, to realize a state in adaptive circuits.
In this work, we tackle this problem by introducing two properties of quantum states: state weight and anti-shallowness, connected to the correlation range and correlation strength within a state, respectively. We prove that these quantities are bounded under limited circuit resources, thereby providing rigorous bounds on the approximate complexity of state preparation and gate implementation. Illustrative examples include the GHZ state, W state, QLDPC code states, and the Toffoli gate.
Besides complexity, we show that states within the same quantum phase, defined by a set of quantum states connected with constant-depth circuits, must share the same scaling of weight or anti-shallowness. This establishes these quantities as indicators of quantum phases and their essential roles in many-body physics.
\end{abstract}

\maketitle

\section{Introduction}
Adaptive quantum circuits, featuring measurements and classical feedback, enable capabilities far beyond purely unitary or non-adaptive operations. Using information extraction, dynamic control, and conditional evolution, adaptive strategies significantly enhance quantum tasks such as state preparation~\cite{Piroli2021adaptive,verresen2022efficientlypreparingschrodingerscat,Yan2025Variational,Malz2024MPS,piroli2024approximatingmanybodyquantumstates}, gate implementation~\cite{Sun2023Preparation,zi2025constantdepthquantumcircuitsarbitrary}, circuit sampling~\cite{cao2025measurementdrivenquantumadvantagesshallow}, error correction~\cite{Shor1995code,gottesman1997stabilizer,nielsen2010quantum}, and measurement-based quantum computing~\cite{Raussendorf2003MBQC,Briegel2009MBQC}. Remarkably, even shallow adaptive circuits can generate long-range entangled states~\cite{Lu2022LRE,Zhu2023LRE,Tantivasadakarn2023Finite,Iqbal2024Topological}, which are central to topological phases~\cite{Chen2010topological} and error-correcting codes~\cite{bravyi2024entanglement}. Adaptivity thus expands the landscape of states and gates accessible by quantum circuits.

It has been shown that shallow adaptive circuits can prepare any quantum state exactly~\cite{zi2025constantdepthquantumcircuitsarbitrary}, but this generally requires an exponentially large number of ancillary qubits. In practice, both circuit depth and ancilla count are constrained, which limits the set of states that can actually be realized. For example, in geometrically constrained systems, the information scrambling speed is limited when only a restricted number of measurements are allowed~\cite{friedman2023locality}. One would expect that low-depth, few-ancilla adaptive circuits have only limited state-preparation power. This naturally raises the question: what kinds of quantum states can be realized under limited resources? The situation is particularly unclear in all-to-all systems where any two qubits are directly connected.

Addressing this question requires a careful characterization of state complexity in adaptive circuits. Unlike the non-adaptive setting, where complexity is fully captured by the minimum depth needed to generate a given state~\cite{Brandao2021Complexity,Daniel2017complexity,Liu2018complexity,anshu2022Bounds}, state complexity in adaptive circuits involves two resources: circuit depth and the number of ancillas. Characterizing complexity in this context therefore necessitates identifying intrinsic properties of quantum states that faithfully capture the space and time costs of adaptive preparation, and establishing quantitative relations between these properties and the required resources. Such a framework is essential for uncovering and leveraging the true power of adaptive circuits.

The study of circuit complexity is closely connected to quantum phases~\cite{Chen2010topological,baiguera2025quantumcomplexitygravityquantum}, which are defined as sets of quantum states that can be transformed into one another via shallow circuits, i.e., constant-depth circuits. States within the same phase share similar properties, including entanglement structure~\cite{Chen2010topological,Wen2013LRE,Leone2021quantumchaos}, error robustness~\cite{Kitaev2003anyons,Yi2024AQEC}, and symmetry properties~\cite{Sachdev1999phase}. Thus, it is important to distinguish states from different phases. Notably, while states in the same phase exhibit the same scaling of circuit complexity, the converse does not necessarily hold. Distinguishing quantum states with the same circuit complexity but belonging to different phases thus provides a more refined characterization of quantum states. A quantitative study relies on developing phase indicators: if the scaling of these indicators differs between two states, they necessarily belong to different quantum phases.

In this work, we address these questions by studying state complexity and phase indicators in adaptive circuits, summarized in Fig.~\ref{fig:summary}. We introduce two key properties of quantum states: state weight and anti-shallowness, and show that they are both related to the correlation function. Importantly, we prove that they are phase indicators, and they provide lower bounds on approximate circuit complexity in adaptive circuits: given their values, one can determine the necessary circuit depth and ancilla count for adaptive state preparation within a small error.

\begin{figure}[htbp]
\centering
\includegraphics[width=15cm]{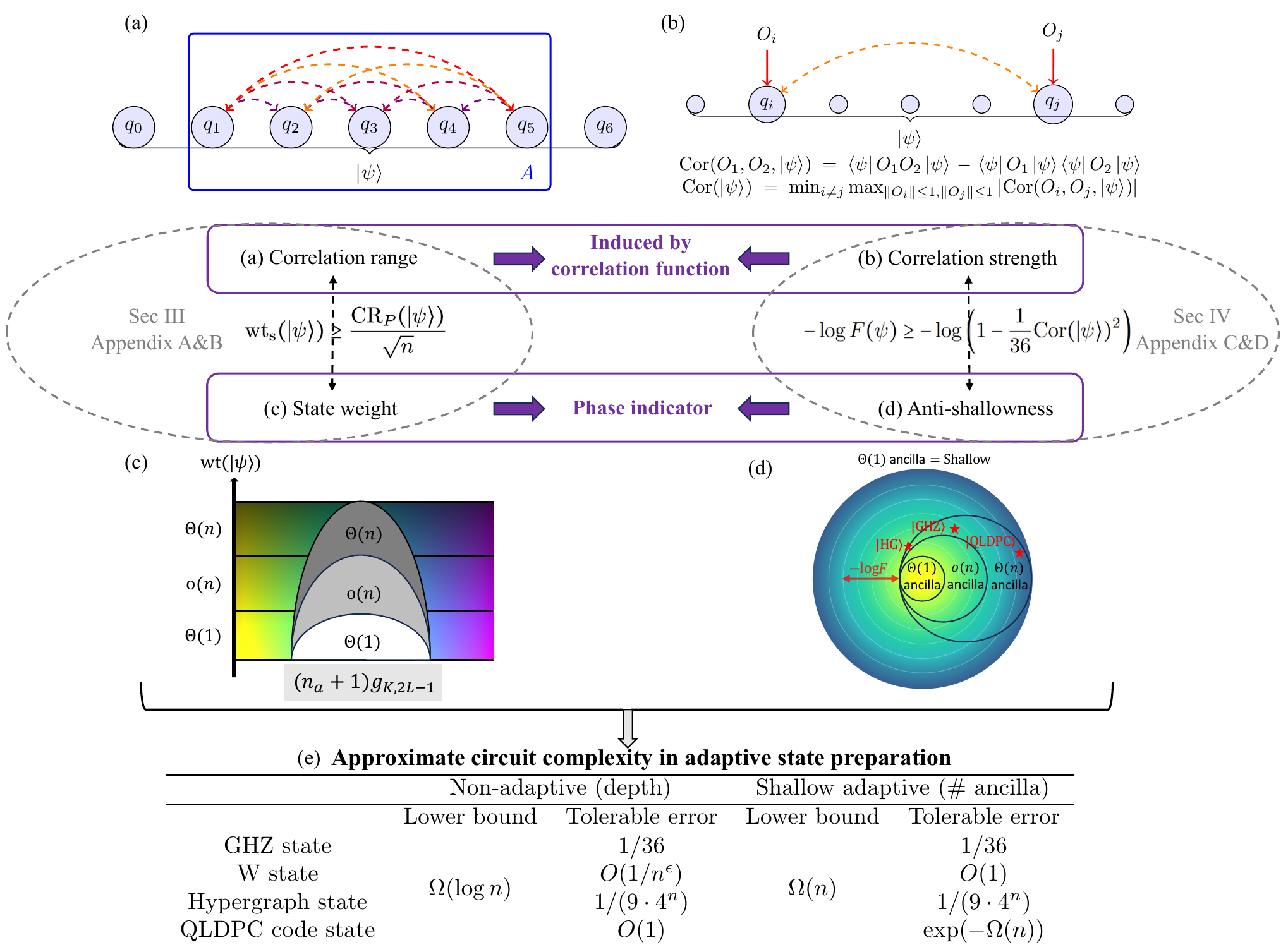}
\caption{
Diagram of main results. We introduce two properties of quantum states, state weight and anti-shallowness, serving as phase indicators, since their distinct scaling behaviors distinguish different quantum phases. We further relate state weight to correlation range and anti-shallowness to correlation strength by proving the two inequalities shown in the figure. The results of state weight are shown in Section~\ref{sc:tradeoff} and Appendices~\ref{appendsc:genestabilizer} and~\ref{appendsc:weight}. The results of anti-shallowness are shown in Section~\ref{sc:shallow} and Appendices~\ref{appendsc:antishallow} and~\ref{appendsc:qlpdc}.
(a) Diagram of correlation range. The correlation range is defined as the maximal region in which all qubits exhibit nonzero correlations.
(b) Diagram of correlation strength. The correlation strength is the minimum magnitude of correlation functions across all qubit pairs.
(c) State weight and associated time-space trade-off relations. For a state prepared by depth-$L$ adaptive circuits with $n_a$ ancillas and $K$-bounded fan-in gates, the state weight satisfies $\wt(\ket{\psi}) \leq (n_a+1) g_{K,2L-1}$.
(d) Anti-shallowness. The anti-shallowness $-\log F$ of any shallow-adaptive-circuit state cannot exceed the number of ancillas, with different colors denoting different magnitudes of anti-shallowness. Notably, although good QLDPC states, GHZ states, and hypergraph states all require $\Omega(n)$ ancillas for preparation, their anti-shallowness scales as $\Theta(n)$, constant, and exponentially small, respectively. We also prove that shallow adaptive circuits with only $\Theta(1)$ ancillas can generate only trivial states, leading to the label ``$\Theta(1)$ ancilla $=$ Shallow" in the diagram.
(e) Examples of quantum states with approximate circuit complexity lower bounds that follow from our indicators. The circuit complexity includes two types. One is the circuit depth in non-adaptive circuits, and the other is the number of ancillary qubits in shallow adaptive circuits. We also show the tolerable error to make the lower bounds hold. That is, preparing quantum states with infidelity smaller than the tolerable error still requires $\Omega(\log n)$ depth or $\Omega(n)$ ancillas. As a remark, both state weight and anti-shallowness are always positive and upper-bounded by $O(n)$.}
\label{fig:summary}
\end{figure}

The definition of state weight is inspired by operator weight and the stabilizer formalism~\cite{gottesman1997stabilizer}. For stabilizer states, one can list the stabilizer generators and minimize their weight. Analogously, for a generic pure state $\ket{\psi}$, one may define generalized stabilizer generators~\cite{Zhang2021generalized,Martin2007generalized}, and define the state weight $\wt(\ket{\psi})$ as the minimized largest support among them. This quantity captures the range of correlations in the system. To formalize this, we introduce the correlation range, the maximal size of a region in which all qubits are mutually correlated, and prove that a large correlation range implies a large state weight.

For both state weight and correlation range, we establish time–space trade-offs in adaptive circuits: if a state is prepared by depth-$L$ adaptive circuits with $n_a$ ancillas and $K$-bounded fan-in gates, their scaling cannot exceed $(n_a+1)g_{K,2L-1}$, as shown in Fig.~\ref{fig:summary}(c). Here, $g_{K,2L-1}$ is the lightcone size of depth-$(2L-1)$ circuits. As an example, it equals to $K^{2L-1}$ for all-to-all connected circuits. This bound is further strengthened for stabilizer states and Clifford adaptive circuits. Such trade-offs allow us to estimate the approximate circuit complexity of typical quantum states. For example, preparing an $n$-qubit GHZ or W state with constant infidelity necessarily requires $\Omega(n)$ ancillas. Moreover, the complexity of states translates into the complexity of gates. For instance, the complexity of the hypergraph state $CZ_n\ket{+}^{\otimes n}$~\cite{Rossi2013Hyper,Qu2013hypergraph,Chen2024magicofquantum} implies a lower bound on $n$-qubit Toffoli gate compiling, which requires $\Omega(\log n)$ circuit depth or $\Omega(n)$ ancillary qubits for realization.

The second indicator, anti-shallowness, is defined as $-\log F(\psi)$~\cite{bravyi2024entanglement,Wei2003Geometric,Orus2014Geometric}, where $F(\psi)$ is the maximal fidelity between $\ket{\psi}$ and shallow-circuit states. Anti-shallowness quantifies how far $\ket{\psi}$ deviates from the trivial phase and reflects its robustness against noises~\cite{bravyi2024entanglement}. Unlike state weight, which reflects the range of correlations, anti-shallowness is related to their magnitude. Particularly, we introduce the global correlation strength, the minimum correlation value among all qubit pairs, and show that it provides a lower bound on anti-shallowness.

Focusing on shallow adaptive circuits, attractive due to their low depth and resilience to decoherence~\cite{Tantivasadakarn2023Finite,Iqbal2024Topological}, we prove that the anti-shallowness of any state prepared in this regime cannot exceed the number of ancillas. This bound is tight: good quantum low-density parity-check (QLDPC) code states~\cite{Panteleev2022goodQLDPC} on $t$ qubits can be prepared with $\Theta(t)$ ancillas, achieving $\Theta(t)$ anti-shallowness scaling. Thus, shallow adaptive circuits with $\Theta(n)$ ancillas suffice to realize the maximum possible scaling $\Theta(n)$. This result yields an approximate circuit complexity lower bound for good QLDPC code states, which require $\Omega(n)$ ancillas for shallow preparation with exponentially small infidelity. As a phase indicator, anti-shallowness further distinguishes states by their noise robustness. For instance, although good QLDPC states, GHZ states, and hypergraph states all require $\Omega(n)$ ancillas for preparation, their anti-shallowness scales as $\Theta(n)$, constant, and exponentially small, respectively, as illustrated in Fig.~\ref{fig:summary}(d).

The paper is organized as follows. Section~\ref{sc:pre} covers the necessary preliminaries. In Section~\ref{sc:tradeoff}, we introduce the state weight and correlation range, establish time-space trade-off relations in adaptive circuits, and show the approximate circuit complexity for permutation-invariant states in an adaptive setting. Section~\ref{sc:shallow} introduces anti-shallowness, connects its range with the number of ancillas, and provides circuit complexity for good QLDPC code states. This section also presents protocols for preparing QLDPC code states. We conclude with a discussion in Section~\ref{sc:discussion}.

\section{Preliminary}\label{sc:pre}
\subsection{Notations and adaptive circuits}
In this part, we introduce basic concepts and notations in this work. We use $m$ and $n$ to denote the number of qubits, where $m$ is the total number of qubits, and $n$ is typically the number of qubits in the target state. We also denote $n_a = m-n$ as the number of ancillary qubits. The Hilbert space of $n$ qubits is $\mathcal{H} = \bigotimes_{i=1}^n \mathcal{H}_i$, where each $\mathcal{H}_i$ is a 2-dimensional Hilbert space with computational basis $\{\ket{0}, \ket{1}\}$. The basis of $\mathcal{H}$ is then $\{\bigotimes_{i=1}^n \ket{a_i}, a_i\in \{0, 1\}\}$, and any state in $\mathcal{H}$ can be written as a linear combination of these basis states. For convenience, we write $\ket{a}\ket{b} = \ket{a} \otimes \ket{b}$ and $\ket{ab}$ interchangeably, and denote $\ket{0}^{\otimes n}$ as $\ket{0^n}$.

A quantum gate $U$ is a unitary operation on $\mathcal{H}$. If $U$ acts nontrivially on only $K$ qubits, we call it a $K$-bounded fan-in gate. Formally, $U = U_A \otimes \id_{\bar{A}}$, where $A \subseteq [n] \coloneqq \{1,\ldots,n\}$ with $|A| = K$, and $U_A$ cannot be further decomposed. Typically, $K$ is a small constant, reflecting physical constraints. Similarly, a Hermitian operator or observable $O$ can also be written as $O = O_A \otimes \id_{\bar{A}}$, where $A$ is the \textit{support} of $O$, and we denote $\mathrm{supp}(O) = A$. We also
denote the weight of $O$ as $\wt(O) = |A|$.

For state preparation tasks, unless stated otherwise, we assume the initial state is $\ket{\mathbf{0}} \coloneqq \ket{0^m}$. A quantum circuit with only unitary operations prepares a state by applying a sequence of unitary gates ${U_1, U_2, \ldots, U_D}$ to $\ket{\mathbf{0}}$. Each $U_i$ is a layer of $K$-bounded fan-in gates acting on different qubits, and the full unitary is $U = U_D \cdots U_1$. The circuit depth of $U$ is $D$, and the circuit complexity of a state is the minimum depth required to prepare it using such a circuit. A shallow circuit is the case when $D$ is constant.

To enhance circuit power, one can introduce adaptive circuits~\cite{Piroli2021adaptive}. An adaptive quantum circuit includes unitary gates, mid-circuit measurements, classical processing of measurement outcomes, and classically controlled unitary gates based on previous outcomes. Without loss of generality, measurements are in the computational basis, and measured qubits are discarded. Thus, each measurement consumes an additional ancillary qubit. Notably, both quantum gates and measurements should be taken into account for circuit depth. For instance, the sequential measurements in measurement-based quantum computing would generally lead to a non-constant circuit depth.

The adaptive circuit is expressed as $U' = E_L \cdots E_1$, where each layer $E_i$ may include $K$-bounded fan-in unitary gates or measurements with classical control, possibly depending on all previous outcomes. Within each layer, unitary operations and measurements must act on distinct qubits. As long as operations, including gates and measurements, target different qubits, they can be executed in parallel, thereby reducing the circuit depth. The time cost is defined as the circuit depth needed to perform all operations in parallel. Like before, the depth is $L$, and $L$ is constant for shallow adaptive circuits. Note that $U'$ is generally not unitary, and the dimensions of input and output may differ. For clarification, we refer to circuits with only unitary gates as non-adaptive circuits to distinguish them from adaptive ones. A diagram of the adaptive circuit is shown in Fig.~\ref{fig:Genericadapt}.

\begin{figure}[!htbp]
 \flushleft
 \subfigure[]{
 \begin{minipage}[b]{0.48\textwidth}\label{fig:Genericadapt}
  \begin{quantikz}[row sep=0.2cm, column sep=0.2cm]
  \lstick{$\ket{0}$} & \gate{H} & \ctrl{2} & \qw  & \qw & \qw  & \qw & \qw & \qw & \gate{Y} & \qw \rstick{$\ket{\psi}$} \\
  \lstick{$\ket{0}$} & \gate{H} & \ctrl{1}  & \qw & \meter{} & \cw & \cwbend{1}  \\
  \lstick{$\ket{0}$} & \qw & \targ{} & \gate{T} & \qw & \qw & \gate{S} & \meter{} & \cw & \cwbend{-2} \\
  \end{quantikz}
 \end{minipage}
 }
 \subfigure[]{
 \begin{minipage}[b]{0.48\textwidth}\label{fig:Clifadapt}
  \begin{quantikz}[row sep=0.2cm, column sep=0.2cm]
  \lstick{$\ket{0}$} & \gate{H} & \ctrl{1} & \qw & \meter{} & \cw & \cw & \cwbend{1} \\
  \lstick{$\ket{0}$} & \qw & \targ{} & \ctrl{1} & \qw & \qw & \gate{S} & \gate{X} & \qw \rstick{$\ket{\psi}$}  \\
  \lstick{$\ket{0}$} & \gate{H} & \qw & \targ{} & \meter{} & \cw & \cwbend{-1}   \\
  \end{quantikz}
 \end{minipage}
 }
\label{fig:adapt}
\caption{Diagram of adaptive circuits. Figure (a) shows a generic adaptive circuit comprising any quantum gate with a bounded fan-in of 3. Figure (b) shows a Clifford adaptive circuit, where only classically controlled Clifford gates and Pauli measurements are allowed.}
\end{figure}
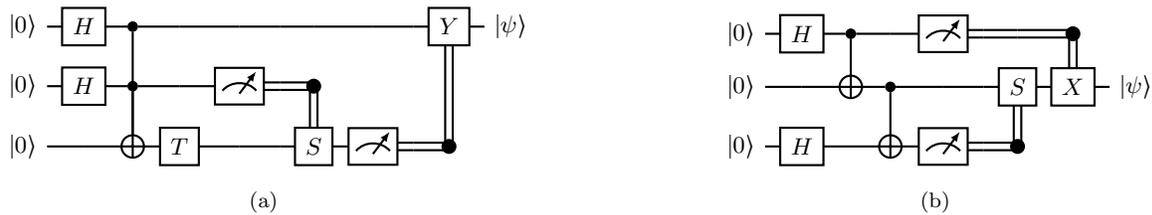

In this work, we normally consider the adaptive state preparation task of using an $m$-qubit system to prepare an $n$-qubit state. The process involves applying a constant number of quantum gate layers $U_1, U_2, \dots, U_{L'}$, interleaved with $n_a$ mid-circuit measurements. The goal is to obtain an $n$-qubit output state. Without loss of generality, in an all-to-all system, we assume the $i$-th measurement is performed on the $(n+i)$-th qubit, yielding a computational-basis outcome $s_i \in \{0,1\}$ and leaving behind the corresponding state $\ket{s_i}$. Let $\mathbf{s} = s_1s_2\cdots s_{n_a}$ denote the full string of measurement outcomes. The total circuit depth is $L$, accounting for both unitary layers and measurements, with $L' \leq L$ denoting the number of unitary layers. The final state $\ket{\psi}$ can be expressed as:
\begin{equation}\label{eq:adapt_state}
\ket{\psi} \propto (\id_n\otimes \bra{\mathbf{s}}) U_{L'}\cdots U_2U_1\ket{0^m},
\end{equation}
where $\id_n$ is the identity on the first $n$ qubits. Each layer $U_i$ consists of $K$-bounded fan-in gates and may depend on the outcomes of prior measurements. This expression is valid because the measured qubits are discarded immediately. For shallow circuits, both $L$ and $L'$ are constants. When $m = n$, this reduces to the standard non-adaptive circuit case.

\subsection{Stabilizer formalism}
Below, we introduce the concepts of the Pauli group, the Clifford group, and the stabilizer state.
The Pauli group on $n$ qubits is defined as
\begin{equation}
\mathsf{P}_n = \{\pm 1, \pm i\}\times \{\id, X, Y, Z\}^{\otimes n},
\end{equation}
where $\id$ is a 2-dimensional identity operator, and $X, Y, Z$ are three Pauli matrices,
\begin{equation}
X = \begin{pmatrix}
0 & 1\\
1 & 0
\end{pmatrix},
Y = \begin{pmatrix}
0 & -i\\
i & 0
\end{pmatrix},
Z = \begin{pmatrix}
1 & 0\\
0 & -1
\end{pmatrix}.
\end{equation}
The $n$-qubit Clifford group $\mathsf{C}_n$ is the normalizer group of $\mathsf{P}_n$,
\begin{equation}
\mathsf{C}_n = \{C |\forall P\in \mathsf{P}_n, CPC^{-1}\in \mathsf{P}_n\}.
\end{equation}

An $n$-qubit stabilizer state, $\ket{\psi}$, is defined as the common eigenstate of $n$ independent and commutative Pauli operators, $\{S_1, S_2, \cdots, S_n\}$ while all the eigenvalues are $1$, i.e., $S_i\ket{\psi} = \ket{\psi}$. Independence means no $S_i$ can be written as a product of the others. Since all $S_i$ must have eigenvalue $+1$, they must belong to the set ${\pm 1} \times \{\id, X, Y, Z\}^{\otimes n}$ but cannot be $-\id^{\otimes n}$. The group $\mathbb{S}_{\psi}\equiv\langle S_1, S_2, \cdots, S_n\rangle$ generated by these $n$ Pauli operators is called the stabilizer group of $\ket{\psi}$, and $\mathcal{S}_{\psi}\equiv\{S_1, S_2, \cdots, S_n\}$ are named stabilizer generators. A stabilizer state can always be represented as a summation of all elements in the stabilizer group,
\begin{equation}
\ketbra{\psi}=\prod_{S\in \mathcal{S}_{\psi}}\frac{\id+S}{2}=\frac{1}{2^n}\sum_{S\in \mathbb{S}_{\psi}}S.
\end{equation}

When the gates are restricted to the Clifford case, we call the (adaptive) quantum circuit the (adaptive) Clifford circuit. It is worth mentioning that each Clifford gate, including the classically controlled Clifford gates, must have a bounded fan-in of $K$. The Clifford circuit enjoys the property of classical simulability~\cite{gottesman1998heisenberg,Aaronson2004stabilizer} and hence is favored in quantum information studies. We depict the diagram of the adaptive Clifford circuit in Fig.~\ref{fig:Clifadapt}. Notably, any $n$-qubit stabilizer state can be prepared using $O(n^2)$ two-qubit Clifford gates, and the depth can be compressed to $O(\log n)$ with ancillary qubits~\cite{Aaronson2004stabilizer}.

Another important concept in this work is the QLDPC code. An $[[n, k, d]]$ QLDPC code space is defined by the common eigenspace of $n-k$ independent and mutually commutative projectors $\{\Pi_1, \Pi_2, \cdots, \Pi_{n-k}\}$,
\begin{equation}
\mathcal{C} = \{\ket{\psi}\in \bigotimes_{j=1}^n\mathcal{H}_j |\forall i, \Pi_i\ket{\psi} = \ket{\psi}\}.
\end{equation}
The LDPC condition requires that each projector $\Pi_i$ has weight $\wt(\Pi_i) \leq s = O(1)$, and each qubit appears in at most $s$ projectors. The maximum weight of the projectors, or $s$, is named the sparsity of the code. The code distance $d$ represents the tolerable error of code space $\mathcal{C}$. This is manifested by the Knill-Laflamme conditions~\cite{Knill1997KLcondition}. For any observable $O$ with $\wt(O)<d$,
\begin{equation}
\Pi_{\mathcal{C}} O \Pi_{\mathcal{C}} = c(O)\Pi_{\mathcal{C}},
\end{equation}
where $\Pi_{\mathcal{C}}$ is the projector to code space $\mathcal{C}$ and $c(O)$ is a quantity related to $O$. A good QLDPC code requires $k = \Theta(n)$ and $d = \Theta(n)$.

Similarly, an $[[n, k, d]]$ stabilizer code is defined as the common $+1$ eigenspace of $n-k$ independent and mutually commutative Pauli operators $\{S_1, S_2, \cdots, S_{n-k}\}$. The Pauli operators in $\{S_1, S_2, \cdots, S_{n-k}\}$ are called stabilizer generators. A code can be both a QLDPC code and a stabilizer code as long as $\wt(S_i) = O(1)$, and each qubit appears in support of $O(1)$ stabilizer generators.

\section{State weight and time-space trade-off in adaptive state preparation}\label{sc:tradeoff}
In this section, we formally define the concept of state weight and illustrate its role in phase classification in Section~\ref{ssc:stateweight}. Using this quantity, we then derive time-space trade-off relations for adaptive state preparation. As a preliminary step, we analyze the non-adaptive case first, which also provides a key tool for proving the adaptive results. The main findings for adaptive circuits are presented in Section~\ref{ssc:adt}. Beyond these general results, we establish tighter bounds for preparing stabilizer states using non-adaptive and adaptive Clifford circuits, introducing the concept of stabilizer weight, a variant of state weight tailored to stabilizer states. For adaptive circuits, we identify a trade-off between circuit depth and the number of ancillary qubits, showing that the depth advantage of adaptive circuits scales with the number of ancillas. This highlights a fundamental trade-off between ``time" and ``space" resources and shows that the state weight characterizes the essential resources required to prepare a given state.

In Section~\ref{ssc:corrrange}, we introduce the correlation function and define the correlation range of a quantum state. The correlation range is the maximal size of a region where any two distinct subsets of qubits within this region have correlations above a threshold. We demonstrate a close relationship between the weight and the correlation range of a state, providing a physical intuition of state weight. Similar to state weight, the correlation range can be used to analyze the necessary time and space resources to prepare a quantum state. Based on weight and correlation range, we are able to evaluate the approximate circuit complexity of quantum states in an adaptive setting. We summarize the circuit complexity of typical permutation-invariant states in Section~\ref{ssc:PIstate}.

The previous four subsections focus on circuits with all-to-all connectivity, extending the results in Ref.~\cite{friedman2023locality}, which requires locality constraints. This extension is essential for tasks with such connectivity, like the implementation of QLDPC codes~\cite{Panteleev2022goodQLDPC,leverrier2022quantumtannercodes,dinur2022goodquantumldpccodes}. Our results also apply to the case with geometric constraints, which are discussed in Section~\ref{ssc:geometric}. A summary of the trade-off relations derived in this section is provided in Table~\ref{table:bound}.

\begin{table}[!ht]
\centering
\caption{Summary of trade-off relations for generic and stabilizer state preparation. The circuit has depth $L$, with $m$ initial qubits, $n$ final qubits, and $n_a = m - n$ ancillary qubits. The gate fan-in is bounded by $K$, and the light cone size $g_{K, L}$ equals $K^L$ for all-to-all connected circuits. The quantities $\wt(\ket{\psi})$ and $\wts(\ket{\psi})$ denote the state weight and stabilizer weight, respectively. The term $\mathrm{CR}(\ket{\psi})$ denotes the correlation range of $\ket{\psi
}$. These trade-off relations demonstrate that ancillary qubits enable a linear expansion of the light cone, underscoring their role in enhancing state preparation efficiency.}
\resizebox{\textwidth}{!}{
\begin{tabular}{cccc}\hline
& Non-adaptive & Adaptive & Clifford adaptive \\
Generic state & $g_{K,L}\geq \wt(\ket{\psi}), g_{K,L}\geq \mathrm{CR}(\ket{\psi})$ & $(n_a+1)g_{K,2L-1}\geq \wt(\ket{\psi}), (n_a+1)g_{K,2L-1}\geq \mathrm{CR}(\ket{\psi})$ & N/A \\
Stabilizer state & $g_{K,L}\geq \wts(\ket{\psi})$ & $(n_a+1)g_{K,2L-1}\geq \wt(\ket{\psi})$ & $(n_a+1)g_{K,L}\geq \wts(\ket{\psi})$  \\\hline
\end{tabular}
}
\label{table:bound}
\end{table}

\subsection{State weight and phase classification}\label{ssc:stateweight}
We begin by defining the stabilizer weight for stabilizer states and then generalize this concept to the weight for arbitrary quantum states.

Given a set of stabilizer generators for a stabilizer state $\ket{\psi}$, we define the weight vector of $\mathcal{S}_\psi$ as
\begin{equation}
\vec{\wt}\left(\mathcal{S}_{\psi}\right)\equiv \left(\wt(S_1), \wt(S_2), \cdots, \wt(S_n)\right).
\end{equation}
Without loss of generality, we assume the weights are ordered non-increasingly: $\wt(S_1) \geq \wt(S_2) \geq \cdots \geq \wt(S_n)$. This convention is used throughout when referring to a weight vector. A total order between two weight vectors can thus be defined: $\vec{\wt}(\mathcal{S}_\psi) > \vec{\wt}(\mathcal{S}'_\psi)$ iff there exists $k \in [n]$ such that $\wt(S_i) = \wt(S'_i)$ for all $i < k$, and $\wt(S_k) > \wt(S'_k)$.

Note that the choice of stabilizer generators for a given stabilizer state $\ket{\psi}$ is not unique. Let $\mathcal{S}_\psi^*$ denote the set of stabilizer generators with the minimal weight vector, referred to as the minimal stabilizer generators:
\begin{equation} \label{eq:minwtvec}
\mathcal{S}^*_{\psi} = \mathop{\arg\min}_{\mathcal{S}_{\psi}} \vec{\wt}\left(\mathcal{S}_{\psi}\right),
\end{equation}
and define the stabilizer weight of $\ket{\psi}$ as $\wts(\ket{\psi}) = \wt(S_1^*)$, where $S_1^* \in \mathcal{S}^*_\psi$.

For a general quantum state $\ket{\psi}$, we first define a generalized stabilizer group~\cite{Zhang2021generalized,Martin2007generalized}. A generalized stabilizer $T$ of $\ket{\psi}$ is a Hermitian, involutory matrix with zero trace, i.e., $T = T^\dagger = T^{-1}$ and $\tr(T) = 0$, such that $T\ket{\psi} = \ket{\psi}$. A generalized stabilizer group $\mathbb{T}_\psi \equiv \langle T_1, T_2, \ldots, T_n \rangle$ is generated by $n$ independent, commuting generalized stabilizers $\mathcal{T}_\psi = \{ T_1, T_2, \ldots, T_n \}$, where all non-identity elements in $\mathbb{T}_\psi$ are generalized stabilizers of $\ket{\psi}$.

Note that not all such sets of generators yield a valid generalized stabilizer group. An additional requirement is that all products $T(\vec{t}) \equiv \prod_i T_i^{t_i}$, for $\vec{t} = (t_1, t_2, \ldots, t_n) \in \mathbb{Z}_2^n$, must also be generalized stabilizers; in particular, they must have zero trace.

This definition is equivalent to one based on state preparation via a unitary transformation. Specifically, $T_i$ is a generalized stabilizer generator of $\ket{\psi}$ if and only if there exists a unitary $U$ such that $T_i = U Z_i U^\dagger$, where $Z_i$ is the Pauli-$Z$ operator on qubit $i$. A proof of this equivalence is given in Appendix~\ref{appendsc:genestabilizer}. Consequently, $\ket{\psi} = U \ket{\mathbf{0}}$, and its density matrix can be expressed as
\begin{equation}
\ketbra{\psi} = \prod_{T \in \mathcal{T}_\psi} \frac{\id + T}{2} = \frac{1}{2^n} \sum_{T \in \mathbb{T}_\psi} T.
\end{equation}

Unlike stabilizer states, the generalized stabilizer group for a generic state is not unique, as it depends on the choice of $U$. This ambiguity must be considered when defining the set of generalized stabilizer generators with minimal weight vector:
\begin{equation} \label{eq:minwtvecgene}
\mathcal{T}^*_{\psi} = \mathop{\arg\min}_{\mathcal{T}_{\psi}} \vec{\wt}\left(\mathcal{T}_{\psi}\right).
\end{equation}
We then define the state weight as $\wt(\ket{\psi}) = \wt(T_1^*)$, where $T_1^* \in \mathcal{T}_\psi^*$.

It is straightforward to see that for any stabilizer state $\ket{\psi}$, the quantum state weight satisfies $\wt(\ket{\psi}) \leq \wts(\ket{\psi})$. However, whether equality always holds is not clear. Notably, local unitary (LU) equivalent states share the same state weight, while local Clifford (LC) equivalent stabilizer states share the same stabilizer weight. Since there exist stabilizer states that are LU-equivalent but not LC-equivalent~\cite{ji2007LULC}, this implies the possibility of stabilizer states $\ket{\psi}$ for which $\wt(\ket{\psi}) \ne \wts(\ket{\psi})$. Moreover, whether the weight of a stabilizer state should have the same scaling as its stabilizer weight also remains unknown.

Interestingly, the state weight serves as a useful indicator for distinguishing between different quantum phases. Specifically, if two quantum states have different scaling behaviors of their state weights, they must belong to different quantum phases. This follows from the fact that a shallow circuit cannot change the scaling of operator weight, which holds based on Lemma~\ref{lemma:operatorweightgrowth}.

To see this, suppose $\ket{\psi_2} = U\ket{\psi_1}$ for some shallow-circuit unitary $U$. Then we have $\wt(\ket{\psi_2}) / \wt(\ket{\psi_1}) = \Theta(1)$. The argument proceeds as follows: the set $U\mathcal{T}_{\psi_1}^* U^\dagger$ forms a valid set of generalized stabilizer generators for $\ket{\psi_2}$. Therefore, we have
$\wt(\ket{\psi_2}) \leq \wt(UT_1^*U^\dagger)$, where $\wt(\ket{\psi_1}) = \wt(T_1^*)$. Lemma~\ref{lemma:operatorweightgrowth} ensures that shallow circuits only increase operator weight by a constant factor, implying $\wt(\ket{\psi_2}) = O(\wt(\ket{\psi_1}))$. The reverse direction follows by symmetry.

Thus, shallow circuits preserve the scaling behavior of state weight, meaning that states with different weight scalings cannot be connected by such circuits. Consequently, distinct scalings of state weight imply distinct quantum phases. This criterion can also be extended to stabilizer states, where the stabilizer weight serves as a phase indicator in the same way.

Below, we present Lemma~\ref{lemma:operatorweightgrowth}, showing how operator weight grows after one layer of quantum gates.

\begin{lemma}\label{lemma:operatorweightgrowth}
Given a matrix, $O$, and a layer of $K$-bounded fan-in unitary gates, $U$, we have $\wt(UOU^{\dagger})\leq K\wt(O)$.
\end{lemma}
\begin{proof}
Any unitary $U$ can be decomposed into $U = \bigotimes_{\abs{A}\leq K} U_{A}$ where $U_{A}$ only acts on the subsystem $A$. Define set $\mathcal{A} = \{A| A\bigcap \mathrm{supp}(O)\neq \emptyset\}$. Obviously, $\abs{\mathcal{A}}\leq \wt(O)$. Note that $U O U^{\dagger} = (\bigotimes_{A\in \mathcal{A}} U_{A}) O (\bigotimes_{A\in \mathcal{A}} U_{A}^{\dagger})$. Then,
\begin{equation}
\begin{split}
\wt(U O U^{\dagger})
&\leq \abs{\mathrm{supp}(O) \bigcup \cup A\in \mathcal{A}}\\
&\leq \abs{\mathrm{supp}(O)} - \abs{\mathcal{A}} + \abs{\cup A\in \mathcal{A}}\\
&\leq \wt(O) - \abs{\mathcal{A}} + \sum_{A\in \mathcal{A}}\abs{A}\\
&\leq \wt(O) - \abs{\mathcal{A}} +  K\abs{\mathcal{A}}\\
&\leq K\wt(O).
\end{split}
\end{equation}
\end{proof}

Lemma~\ref{lemma:operatorweightgrowth} implies that the growth of operator weight along the quantum circuit is bounded by the ``forward lightcone" from the initial operator, as shown in Fig.~\ref{fig:forwardlightcone}, which can be applied to generalized stabilizer generators. For any state $\ket{\psi}$ prepared by a $D$-depth unitary $U$, the set $\mathcal{T}=\{UZ_iU^{\dagger}\}_{i=1}^n$ forms a valid set of generalized stabilizer generators for $\ket{\psi}$. Based on Lemma~\ref{lemma:operatorweightgrowth}, we have
\begin{equation}\label{eq:prooftheo1}
\wt(\ket{\psi})=\wt(T^*_1)\leq\wt(T_1)=\wt(UZ_{i_1}U^{\dagger})\leq K^D.
\end{equation}
where $Z_{i_1}$ corresponds to the element $T_1$ in $\mathcal{T}$. This establishes a lower bound on the circuit depth $D$ required to prepare $\ket{\psi}$ in the non-adaptive setting.

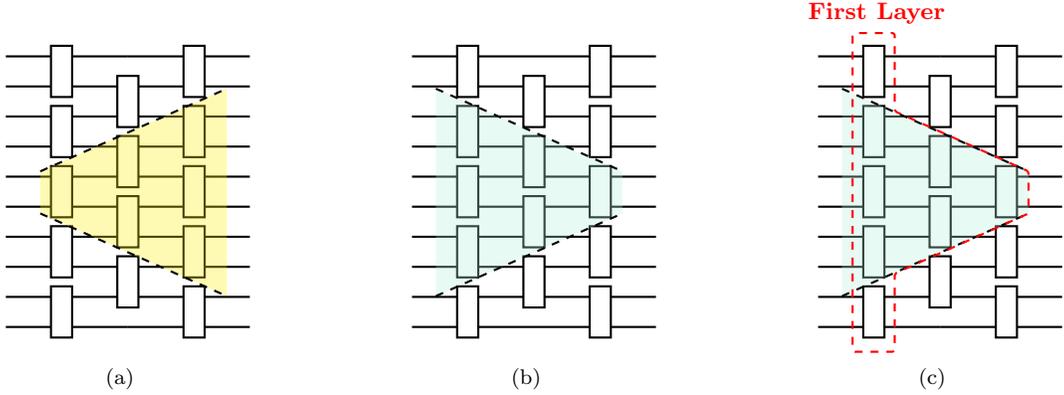
\begin{figure}[!htbp]
 \centering
  \subfigure[]{\label{fig:forwardlightcone}
  \begin{quantikz}[row sep={0.4cm,between origins}, column sep=0.6cm]
  & \gate[wires=2]{} & \qw & \gate[wires=2]{} & \qw \\
  &  & \gate[wires=2]{} &  & \qw  \\
  & \gate[wires=2]{} &  & \gate[wires=2]{} & \qw  \\
  &  & \gate[wires=2]{} &  & \qw  \\
  & \gate[wires=2]{} &  & \gate[wires=2]{} & \qw  \\
  &  & \gate[wires=2]{} &  & \qw  \\
  & \gate[wires=2]{} &  & \gate[wires=2]{} & \qw  \\
  &  & \gate[wires=2]{} &  & \qw  \\
  & \gate[wires=2]{} &  & \gate[wires=2]{} & \qw  \\
  &  & \qw &  & \qw  \\
  \end{quantikz}
  \begin{tikzpicture}[overlay,x=4mm,y=4mm]
    \draw[black, dashed, thick] (-7,1) -- ++(6.2,2.75);
    \draw[black, dashed, thick] (-7,-0.4) -- ++(6.2,-2.75);
    \begin{scope}[on background layer]
    \fill[yellow, fill opacity=0.3, draw=none] (-7,1) -- (-7,-0.4) -- (-0.8,-3.15) -- (-0.8,3.75) -- cycle;
    \end{scope}
  \end{tikzpicture}
  }
  \hspace{16mm}
  \subfigure[]{\label{fig:backwardlightcone}
  \begin{quantikz}[row sep={0.4cm,between origins}, column sep=0.6cm]
  & \gate[wires=2]{} & \qw & \gate[wires=2]{} & \qw\\
  &  & \gate[wires=2]{} &  & \qw  \\
  & \gate[wires=2]{} &  & \gate[wires=2]{} & \qw  \\
  &  & \gate[wires=2]{} &  & \qw  \\
  & \gate[wires=2]{} &  & \gate[wires=2]{} & \qw  \\
  &  & \gate[wires=2]{} &  & \qw  \\
  & \gate[wires=2]{} &  & \gate[wires=2]{} & \qw  \\
  &  & \gate[wires=2]{} &  & \qw  \\
  & \gate[wires=2]{} &  & \gate[wires=2]{} & \qw  \\
  &  & \qw &  & \qw  \\
  \end{quantikz}
  \begin{tikzpicture}[overlay,x=4mm,y=4mm]
    \draw[black, dashed, thick] (-7.35,3.75) -- ++(6.2,-2.75);
    \draw[black, dashed, thick] (-7.35,-3.15) -- ++(6.2,2.75);
    \begin{scope}[on background layer]
    \fill[mintcyan, fill opacity=0.3, draw=none] (-7.35,3.75) -- (-7.35,-3.15) -- (-1.15,-0.4) -- (-1.15,1) -- cycle;
    \end{scope}
  \end{tikzpicture}
  }
  \hspace{16mm}
  \subfigure[]{\label{fig:backwardfirst layer}
  \begin{quantikz}[row sep={0.4cm,between origins}, column sep=0.6cm]
  & \gate[wires=2]{} & \qw & \gate[wires=2]{} & \qw\\
  &  & \gate[wires=2]{} &  & \qw  \\
  & \gate[wires=2]{} &  & \gate[wires=2]{} & \qw  \\
  &  & \gate[wires=2]{} &  & \qw  \\
  & \gate[wires=2]{} &  & \gate[wires=2]{} & \qw  \\
  &  & \gate[wires=2]{} &  & \qw  \\
  & \gate[wires=2]{} &  & \gate[wires=2]{} & \qw  \\
  &  & \gate[wires=2]{} &  & \qw  \\
  & \gate[wires=2]{} &  & \gate[wires=2]{} & \qw  \\
  &  & \qw &  & \qw  \\
  \end{quantikz}
  \begin{tikzpicture}[overlay,x=4mm,y=4mm]
    \draw[black, dashed, thick] (-7.35,3.75) -- ++(6.2,-2.75);
    \draw[black, dashed, thick] (-7.35,-3.15) -- ++(6.2,2.75);
    \begin{scope}[on background layer]
    \fill[mintcyan, fill opacity=0.3, draw=none] (-7.35,3.75) -- (-7.35,-3.15) -- (-1.15,-0.4) -- (-1.15,1) -- cycle;
    \end{scope}

    \draw[red, dashed, thick, rounded corners=1pt] (-7,5.6) -- (-7,-5) -- (-5.6,-5) -- (-5.6,-2.4) -- (-1.15,-0.4) -- (-1.15,1.0) -- (-5.6, 3) -- (-5.6,5.6) -- cycle;

    \node[fit={(-9.2,8) (-3.2,4)}, text=red, font=\small\bfseries] {First Layer};
  \end{tikzpicture}
  }
\caption{(a) Forward lightcone diagram, illustrating the set of qubits that can be influenced by the initial two qubits through the circuit evolution. (b) Backward lightcone diagram, illustrating the set of qubits that can influence the final two qubits under the circuit evolution. (c) The operation on the new first layer consists of two parts: the effective operator resulting from the reverse evolution of the $m - n$ measurement operations, and the unitary gates from the original first layer that lie outside the backward lightcone.}
\label{fig:lightconediagram}
\end{figure}

For stabilizer states, this bound can be further tightened by replacing the state weight with the stabilizer weight. A full derivation is provided in Appendix~\ref{ssc:nonadt}.

\subsection{Trade-off relations in adaptive state preparation}\label{ssc:adt}
We now analyze the resource requirements for state preparation using adaptive circuits based on state weight. We begin by presenting results for the preparation of generic quantum states using general adaptive quantum circuits.

\begin{theorem}\label{theo:adalowerboundgeneric}
Suppose a depth-$L$ adaptive circuit with $K$-bounded fan-in gates and $m$-qubit initial state, $\ket{0^m}$, prepares $n$-qubit state $\ket{\psi}$, we have the following bound:
\begin{equation}\label{eq:boundadaptive}
(m-n+1)K^{2L-1}\geq\wt(\ket{\psi}).
\end{equation}
\end{theorem}

The proof of Theorem~\ref{theo:adalowerboundgeneric} is based on the method of ``backward lightcone" originating from the measurement operations, as illustrated in Fig.~\ref{fig:backwardlightcone}. We begin by analyzing the reverse evolution of the $m - n$ measurement operations to obtain an effective operator acting on the first layer of the circuit. This reverse evolution causes each measurement to spread backward through the circuit, forming a backward lightcone. The weight of the resulting effective operator is thus bounded by the size of this lightcone, which is at most $(m - n)K^L$. Consequently, the weight of the state after this new ``first layer" is also bounded by the backward lightcone size. Beyond this point, the remaining operations are bounded fan-in unitary gates. To estimate the final state's weight, we apply the forward lightcone argument or Lemma~\ref{lemma:operatorweightgrowth}, which contributes an additional expansion factor of $K^{L-1}$. The full proof is provided in Appendix~\ref{appendsc:thmpf:adalowerboundgeneric}.

Note that when $m = n$, or there is no measurement, Eq.~\eqref{eq:boundadaptive} reduces to $K^{2L-1} \geq \wt(\ket{\psi})$, which is weaker than Eq.~\eqref{eq:prooftheo1}. This looseness arises from a non-tight inequality used in the derivation. Accordingly, we apply Eq.~\eqref{eq:boundadaptive} only in the case where $m - n > 0$.

When $m - n > 0$, and we focus only on depth scaling - disregarding the difference between $K^L$ and $K^{2L-1}$ - Theorem~\ref{theo:adalowerboundgeneric} suggests that measurements on $m - n$ qubits, combined with appropriate post-selection or control, can effectively ``enlarge" the forward lightcone by a factor of $m - n$. However, it is important to emphasize that $K^L$ and $K^{2L-1}$ represent fundamentally different regimes. As no known state preparation protocol saturates the bound in Eq.~\eqref{eq:boundadaptive} exactly, we conjecture that the bound remains valid even if $K^{2L-1}$ is replaced by $K^L$. We further support this conjecture by verifying it in the case of preparing stabilizer states using adaptive Clifford circuits, as shown in the following proposition.

\begin{proposition}\label{prop:adalowerboundstab}
Suppose a depth-$L$ Clifford adaptive circuit with $K$-bounded fan-in gates and $m$-qubit initial state, $\ket{0^m}$, prepares $n$-qubit stabilizer state $\ket{\psi}$, we have the following bound:
\begin{equation}\label{eq:boundadaptivestab}
(m-n+1)K^L \geq \wts(\ket{\psi}).
\end{equation}
\end{proposition}

Proposition~\ref{prop:adalowerboundstab} is proved by combining the lightcone method with the stabilizer formalism. Since measurements can be deferred to the end of the circuit, the lightcone argument bounds the weight of stabilizer generators of the pre-measurement state by $K^L$. Because we consider Clifford adaptive circuits, the pre-measurement state is necessarily a stabilizer state. The stabilizers of the pre-measurement and post-measurement states are related through the measurement outcomes. Using this relationship, the stabilizer weight of the final state can also be bounded, specifically, by $(m - n + 1) \cdot K^L$. The detailed proof is provided in Appendix~\ref{appendssc:proofadalowerboundstab}.

As an example, adaptive state preparation for GHZ state $\ket{\mathrm{GHZ}}=\frac{1}{2^{n/2}}(\ket{0}^{\otimes n}+\ket{1}^{\otimes n})$ always saturates the bound in Proposition~\ref{prop:adalowerboundstab}, with the scaling of $m - n$ and $K^L$ varying depending on the chosen parameters. It can be implemented using an adaptive circuit with depth $L = \log_K a$ and $\lceil \frac{n}{a} \rceil - 1$ ancillary qubits for any choice of $a$, yielding $(m-n+1)K^L\approx n=\wts(\ket{\mathrm{GHZ}})$,
since the stabilizer weight of an $n$-qubit GHZ state is $n$.
The preparation scheme involves first generating a $\lceil \frac{n}{a} \rceil$-qubit GHZ state using a shallow adaptive circuit with all ancillary qubits~\cite{Piroli2021adaptive}, followed by applying CNOT gates between this state and the $\ket{0^{n - \lceil \frac{n}{a} \rceil}}$ state. We remark that GHZ state preparation saturates the scaling bound in Theorem~\ref{theo:adalowerboundgeneric} when either $m - n = \Theta(1)$ and $L = \Theta(\log n)$, or $m - n = \Theta(n)$ and $L = \Theta(1)$. But the intermediate regimes are not tight.

It is worth mentioning that Theorem~\ref{theo:adalowerboundgeneric} and Proposition~\ref{prop:adalowerboundstab} applies to the case that the state $\ket{\psi}$ is prepared probabilistically, because our proof only utilizes Eq.~\eqref{eq:adapt_state} that $\ket{\psi}$ is obtained by measuring an $m$-qubit state regardless of the probability to get $\ket{\psi}$. A deterministic state preparation case is automatically included as a stochastic case with probability $1$.




\subsection{State weight and correlation range}\label{ssc:corrrange}
In the following, we introduce the correlation range and show its relation with weight. We first define the correlation function of two operators $O_1,O_2$ on a state $\ket{\psi}$ as
\begin{equation}\label{eq:correlationfunction}
\mathrm{Cor}(O_1,O_2,\ket{\psi})=\bra{\psi}O_1O_2\ket{\psi}-\bra{\psi}O_1\ket{\psi}\bra{\psi}O_2\ket{\psi}.
\end{equation}


Based on the correlation function, we can define the correlation strength within region $A$. By choosing $O_1$ and $O_2$ as two operators on the support of two distinct subsets of qubits to maximize the correlation function, and considering the minimal one among the choices of two subsets, we obtain the following definition:
\begin{equation}
\mathrm{Cor}^A_w(\ket{\psi}) = \min_{\substack{A_1,A_2\subseteq A\\|A_1|\leq w,|A_2|\leq w\\A_1\cap A_2=\emptyset}}\max_{\substack{\norm{O_1}\leq 1,\norm{O_2}\leq 1\\\mathrm{supp}(O_1)=A_1,\mathrm{supp}(O_2)=A_2}}\left|\mathrm{Cor}(O_1,O_2,\ket{\psi})\right|.
\end{equation}
Here, we fix the maximum weight of two operators as $w$ and let it be a tunable parameter. The norms are operator norms, or Schatten $\infty$-norms.  A special case is when $A = [n]$ and $w = 1$. In this case, we call the corresponding correlation strength as the global correlation of state $\ket{\psi}$:
\begin{equation}
\mathrm{Cor}(\ket{\psi})=\mathrm{Cor}_1^{[n]}(\ket{\psi})=\min_{i\neq j}\max_{\norm{O_i}\leq 1,\norm{O_j}\leq 1}\left|\mathrm{Cor}(O_i,O_j,\ket{\psi})\right|
\end{equation}
where $O_i$ and $O_j$ are single-qubit operators on qubit $i$ and $j$ respectively.

The correlation range over a given constant $\delta$ is correspondingly defined as the maximal size of the region with a correlation strength larger than $\delta$:
\begin{equation}
\mathrm{CR}^{\delta}_w(\ket{\psi}) = \max_{\mathrm{Cor}^A_w(\ket{\psi}) > \delta} \abs{A}.
\end{equation}
It is clear that $\mathrm{CR}^{\delta}_w(\ket{\psi})$ monotonically decreases with $\delta$. Specifically, when $\delta=0$, we omit the superscript and define
\begin{equation}
\mathrm{CR}_w(\ket{\psi}) = \max_{\mathrm{Cor}^A_w(\ket{\psi}) > 0} \abs{A}.
\end{equation}

When the operators appearing in the correlation function are all chosen to be single-qubit Pauli operators, we can define the Pauli correlation strength within region $A$. It is the minimum two-point Pauli correlation function in this region:
\begin{equation}
\mathrm{Cor}^A_{P}(\ket{\psi})=\min_{i\neq j\in A}\max_{P_i,P_j\in \{I,X,Y,Z\}}\left|\mathrm{Cor}(P_i,P_j,\ket{\psi})\right|.
\end{equation}
The Pauli correlation range is correspondingly defined as the maximal size of the region with a non-zero Pauli correlation strength:
\begin{equation}
\mathrm{CR}_P(\ket{\psi}) = \max_{\mathrm{Cor}^A_{P}(\ket{\psi}) > 0} \abs{A}.
\end{equation}

Interestingly, we find that the Pauli correlation range provides a lower bound for the stabilizer weight of a state, which is given in the following lemma.
\begin{lemma}\label{lemma:weightandcorrrange}
Given an $n$-qubit stabilizer state, the Pauli correlation range provides a lower bound for the stabilizer weight of a state,
\begin{equation}\label{eq:wtcorrange}
\wts(\ket{\psi}) \geq \frac{\mathrm{CR}_P(\ket{\psi})}{\sqrt{n}}.
\end{equation}
\end{lemma}

The proof of Lemma~\ref{lemma:weightandcorrrange} is available in Appendix~\ref{appendssc:lemmapf:weightandcorrrange}. The key idea is that if Eq.~\eqref{eq:wtcorrange} is violated, then for any region $A$ with $\abs{A}\geq \mathrm{CR}_P(\ket{\psi})$, there will exist a qubit pair $(i,j)$ not covered by a single stabilizer generator. Their correlation will be 0 and contradict the definition of the correlation range. The above lemma manifests a close relationship between correlation range and weight, where a state with a large correlation range will have a large weight. This also gives an intuition of weight, which is roughly the maximal size of the region where any qubits correlate with each other.

Similar to state weight, we can derive a time-space trade-off relation based on the correlation range, shown in the following proposition. The main technique of the proof utilizes the lightcone arguments.

\begin{theorem}\label{thm:adalowerboundgenericcorr}
Suppose a depth-$L$ adaptive circuit with $K$-bounded fan-in gates and $m$-qubit initial state, $\ket{0^m}$, prepares $n$-qubit state $\ket{\psi}$, we have the following bound:
\begin{equation}\label{eq:correlationbound}
(m-n+w)K^{2L-1}+w-1\geq \mathrm{CR}_w(\ket{\psi}).
\end{equation}
\end{theorem}

Now, the state weight and correlation range can both provide the circuit complexity lower bound for preparing quantum states in an adaptive setting. The following lemma shows a continuous property of the correlation function, where two close states have close correlation function values. This would imply a robustness property of the correlation range, where the correlation strength measures the magnitude of the robustness.

\begin{lemma}
For two quantum states $\ket{\psi}$ and $\ket{\phi}$ satisfying $\abs{\braket{\psi}{\phi}}^2\geq 1-\epsilon$, and operators $O_1$, $O_2$ with infinite norms bounded by 1, $\norm{O_1}\leq 1$, $\norm{O_2}\leq 1$, we have that
\begin{equation}
\abs{\mathrm{Cor}(O_1,O_2,\ket{\psi}) - \mathrm{Cor}(O_1,O_2,\ket{\phi})}\leq 6\sqrt{\epsilon}.
\end{equation}
\end{lemma}

\begin{corollary}
For two quantum states $\ket{\psi}$ and $\ket{\phi}$ satisfying $\abs{\braket{\psi}{\phi}}^2\geq 1-\epsilon$, we have that
\begin{equation}
\abs{\mathrm{Cor}^A_w(\ket{\psi}) - \mathrm{Cor}^A_w(\ket{\phi})}\leq 6\sqrt{\epsilon}.
\end{equation}
\end{corollary}

\begin{corollary}
Given the parameter $\epsilon$, we set the quantum state $\ket{\psi}$ with the correlation range $\mathrm{CR}^{\delta}_w(\ket{\psi})$ maximized by region $A$. That is, $\mathrm{CR}^{\delta}_w(\ket{\psi}) = \abs{A}$. Then, for any state $\ket{\phi}$ satisfying $\abs{\braket{\psi}{\phi}}^2\geq 1-\delta^2/36$, we have that
\begin{equation}
\mathrm{CR}_w(\ket{\phi}) \geq \mathrm{CR}_w^{\delta}(\ket{\psi}).
\end{equation}
\end{corollary}

The above results manifest that the correlation strength provides a quantification for the robustness of the correlation range. This would allow one to derive approximate circuit complexity lower bounds for quantum states. For instance, given global correlation $\mathrm{Cor}(\ket{\psi})$, the state $\ket{\phi}$ with infidelity $\mathrm{Cor}^2(\ket{\psi})/36$ with respect to $\ket{\psi}$ will have a similar circuit complexity lower bound. More generally, $\ket{\phi}$ with fidelity $\abs{\braket{\psi}{\phi}}^2\geq 1-\delta^2/36$ has a circuit lower bound given by $(m-n+1)g_{K,2L-1}\geq \mathrm{CR}_w(\ket{\phi})\geq \mathrm{CR}_w^{\delta}(\ket{\psi})$. We will provide concrete examples of the approximate circuit complexity in Table~\ref{table:approxboundexamples} in the next subsection.

\subsection{Circuit complexity for permutation-invariant states}\label{ssc:PIstate}
Here, we analyze the resource requirements for preparing an important class of quantum states known as permutation-invariant states. These states serve as valuable resources in quantum information processing, supporting key protocols in quantum metrology~\cite{leibfried2004toward,lucke2011twin}, quantum communication~\cite{bose1998multiparticle,cleve1999how}, and certain quantum error correction schemes~\cite{ouyang2014permutation}. A quantum state $\ket{\psi}$ is said to be permutation-invariant if it is stabilized by any two-qubit swap gate, i.e., $SWAP(i,j)\ket{\psi} = \ket{\psi}$ for all $i,j$. Nontrivial permutation-invariant states, those not expressible as a tensor product of identical single-qubit states, exhibit long-range entanglement. Prior work has shown that such states cannot be prepared using shallow non-adaptive circuits under geometric locality constraints~\cite{friedman2023locality}. In this work, we establish more general lower bounds by considering a broader class of adaptive state preparation protocols, including those with all-to-all connectivity.

\begin{proposition}\label{prop:permutation-invariant}
Suppose a depth-$L$ adaptive circuit with $K$-bounded fan-in gates and $m$-qubit initial state, $\ket{0^m}$, prepares $n$-qubit permutation-invariant state $\ket{\psi}$. As long as $\ket{\psi}$ is not a fully separable state, there should be
\begin{equation}\label{eq:boundpermutation}
(m-n+1)K^{2L-1}\geq n.
\end{equation}
\end{proposition}

Proposition~\ref{prop:permutation-invariant} is a straightforward corollary of Theorem~\ref{thm:adalowerboundgenericcorr}. It can be proven by showing that the correlation range of any non-fully-separable state should be $n$. If it is not true, we can find two qubits $i$ and $j$ that are not correlated with each other. Since $\ket{\psi}$ is permutation-invariant, the choice of $i$ and $j$ is arbitrary, and this vanishing correlation must hold for all pairs of distinct qubits. In other words, the measurement results of any two qubits are statistically independent, which implies that $\ket{\psi}$ must be a product state, contradicting our assumption.

Representative examples of Proposition~\ref{prop:permutation-invariant} include GHZ states, W states, Dicke states, and symmetric hypergraph states. The W state is defined as $\ket{\mathrm{W}} = \frac{1}{\sqrt{n}} \sum_{i=1}^n \ket{1}_i \otimes \ket{\mathbf{0}}_{-i}$, where $-i$ denotes all qubits except qubit $i$. Dicke states generalize the W state and are given by $\ket{D_k} = \binom{n}{k}^{-1/2} \sum_{|A|=k} \ket{\mathbf{1}}_A \otimes \ket{\mathbf{0}}_{\bar{A}}$. The symmetric hypergraph state considered here is $\ket{\mathrm{HG}} = CZ_n \ket{+}^{\otimes n}$, where $CZ_n \equiv \id_n - 2\ketbra{0}^{\otimes n}$ is the generalized controlled-Z gate.

One can calculate the correlation strength of these states. For $\ket{\mathrm{GHZ}}$, it can be bounded by choosing $Z_i$ and $Z_j$,
\begin{equation}\label{eq:corrGHZ}
\mathrm{Cor}(\ket{\mathrm{GHZ}})=\mathrm{Cor}^{[n]}_1(\ket{\mathrm{GHZ}})\geq |\mathrm{Cor}(Z_i,Z_j,\ket{\mathrm{GHZ}})|=1.
\end{equation}
For W state $\ket{\mathrm{W}}$, we let $O_1=\prod_{i\in A_1}Z_i$ and $O_2=\prod_{j\in A_2}Z_j$ where $A_1$ and $A_2$ are disjoint subsets of $[n]$ with cardinality $w$. Thus,
\begin{equation}\label{eq:corrW}
\mathrm{Cor}^{[n]}_w(\ket{\mathrm{W}})\geq |\mathrm{Cor}(O_1,O_2,\ket{\mathrm{W}})|=\frac{w^2}{n^2}.
\end{equation}
For the Dicke state, the same choice of $O_1$ and $O_2$ gives
\begin{equation}\label{eq:corrDicke}
\mathrm{Cor}^{[n]}_w(\ket{D_k})\geq |\mathrm{Cor}(O_1,O_2,\ket{D_k})|=\left|\frac{\sum_{k_1}(-1)^{k_1}\binom{2w}{k_1}\binom{n-2w}{k-k_1}}{\binom{n}{k}}-\left(\frac{\sum_{k_1}(-1)^{k_1}\binom{w}{k_1}\binom{n-w}{k-k_1}}{\binom{n}{k}}\right)^2\right|=\Theta\left(\frac{kw^2}{n^2}\right)
\end{equation}
when $k$ is a constant and $k\ll w\ll n$. For the symmetric hypergraph state $\ket{\mathrm{HG}}$, we can choose $X_i$ and $X_j$ to get
\begin{equation}\label{eq:corrHG}
\mathrm{Cor}(\ket{\mathrm{HG}})=\mathrm{Cor}^{[n]}_1(\ket{\mathrm{HG}})\geq\mathrm{Cor}(X_i,X_j,\ket{\mathrm{HG}})=\frac{1}{2^{n-2}}\left(1-\frac{1}{2^{n-2}}\right).
\end{equation}
Proposition~\ref{prop:permutation-invariant} gives the circuit lower bound for preparing these states, and Eqs.~\eqref{eq:corrGHZ}-\eqref{eq:corrHG} can extend the results to approximate circuit complexity. We compare our lower bounds with known upper bounds for these states in Table~\ref{table:boundexamples}. Table~\ref{table:approxboundexamples} shows the tolerable error rates to make the lower bounds still hold. As previously discussed, GHZ state preparation saturates the scaling bound in Proposition~\ref{prop:permutation-invariant} in both the non-adaptive and shallow adaptive cases. Apart from GHZ states, only the non-adaptive preparation of W states is also known to saturate the bound~\cite{cruz2019efficient}. Moreover, we note that the state preparation complexity of $\ket{\mathrm{HG}}$ directly implies the circuit complexity of the unbounded fan-in gate $\mathrm{Toffoli}_n=\sum_{x_1,\cdots x_n}\ketbra{x_1,\cdots,x_{n-1},x_n\oplus\prod_{i=1}^{n-1}x_i}{x_1,\cdots,x_{n-1},x_n}$.

\begin{table}[!ht]
\centering
\caption{Summary of lower and upper bounds for permutation-invariant state preparation, including specific examples such as the GHZ state, W state, Dicke state, and hypergraph state $\ket{\mathrm{HG}}=CZ_n\ket{+}^{\otimes n}$. (a) For exact state preparation, we focus on circuit depth in the non-adaptive case and the number of ancillary qubits in the shallow adaptive case. Notably, preparing $\ket{\mathrm{HG}}$ with depth $O\left((\log n)^3\right)$ also requires a single borrowed ancilla~\cite{Claudon2024polylogarithmicdepth}. Currently, no universal adaptive state preparation scheme exists for arbitrary permutation-invariant states. (b) For approximate state preparation, we evaluate the tolerable approximation error to ensure that the circuit complexity lower bounds in Table~\ref{table:boundexamples} still hold. We consider two cases: non-adaptive circuits with $\Omega(\log n)$ depth; shallow adaptive circuits with $\Omega(n)$ ancilla. Here, $\epsilon$ can be any positive constant, and we take the Dicke states with a constant $k$ as a special case since it has a better bound.}
\subtable[]{
\resizebox{.7\textwidth}{!}{
\begin{tabular}{ccccc}\hline
& \multicolumn{2}{c}{Non-adaptive (depth)} & \multicolumn{2}{c}{Shallow adaptive (\# ancilla)} \\\hline
& Lower bound & Upper bound & Lower bound & Upper bound \\\hline
GHZ state & \multirow{5}{*}{$\log_K{n}$} & $\log_K{n}$ & \multirow{5}{*}{$\Omega(n)$} & $O(n)$ \\
W state &  & $O(\log{n})$~\cite{cruz2019efficient} &  & $O(n\log{n})$~\cite{Buhrman2024statepreparation} \\
Dicke state &  & $O(k\log(n/k))$~\cite{bartschi2022shortdepth} &  & $O(n^2\log{n})$~\cite{Buhrman2024statepreparation} \\
Hypergraph state &  & $O\left((\log{n})^3\right)$*~\cite{Claudon2024polylogarithmicdepth} &  & $O(n\log{n})$~\cite{Takahashi2016collapse}  \\
PI state &  & $O(n)$~\cite{bartschi2019deterministic} &  & N/A \\
\hline
\end{tabular}
}
\label{table:boundexamples}
}
\\
\subtable[]{
\resizebox{.7\textwidth}{!}{
\begin{tabular}{ccc}\hline
Tolerable approximation & $\Omega(\log n)$-depth lower bound & $\Omega(n)$-ancilla lower bound \\
error for: & for non-adaptive circuits & for adaptive circuits \\\hline
GHZ state & $1/36$ & $1/36$ \\
W state & $O\left(1/n^{\epsilon}\right)$ & $O(1)$ \\
Dicke state (constant $k$) & $O(1/n^{\epsilon})$ & $O(1)$ \\
Dicke state (general $k$) & $O(k^2/n^2)$ & $O(k^2/n^2)$ \\
Hypergraph state & $1/(9\cdot 4^n)$ & $1/(9\cdot 4^n)$ \\
\hline
\end{tabular}
}
\label{table:approxboundexamples}
}
\end{table}

\subsection{State preparation with geometrically constrained circuits}\label{ssc:geometric}
The trade-off theorems discussed above can be extended to settings with geometric constraints, where gates must not only have bounded fan-in but also act on nearest-neighbor qubits. We provide full details in Appendix~\ref{appendsc:geo} and briefly outline the results here. In this setting, a variant of Lemma~\ref{lemma:operatorweightgrowth} holds: for a geometrically local, $D$-depth quantum circuit $U$ composed of $K$-bounded fan-in gates, the lower bound of $K^D\geq \wt(\ket{\psi})$ changes to
\begin{equation}
g_{K,D}\geq \wt(\ket{\psi}),
\end{equation}
where $g_{K,D}$ depends on the specific geometric constraint. In the all-to-all architecture, $g_{K,D} = K^D$. Under geometric constraints, $g_{K,D}$ is normally smaller due to the limited light cone. For example, in an $r$-dimensional grid where each gate acts on a $K\times K\times\cdots\times K$ hypercube, $g_{K,D} = (2(K-1)D+1)^r$.

Based on $g_{K,D}$, a tighter bound applies for stabilizer states:
\begin{equation}
g_{K,D}\geq \wts(\ket{\psi}).
\end{equation}
For $L$-depth generic adaptive circuits, the bound in Theorem~\ref{theo:adalowerboundgeneric} becomes
\begin{equation}
(m-n+1)g_{K,2L-1}\geq \wt(\ket{\psi}),
\end{equation}
and the bound in Theorem~\ref{thm:adalowerboundgenericcorr} becomes
\begin{equation}
(m-n+w)g_{K,2L-1}+w-1\geq \mathrm{CR}_w(\ket{\psi}).
\end{equation}
For $L$-depth Clifford adaptive circuits, the bound in Proposition~\ref{prop:adalowerboundstab} becomes
\begin{equation}
(m-n+1)g_{K,L} \geq \wts(\ket{\psi}).
\end{equation}

\section{Anti-shallowness in adaptive state preparation}\label{sc:shallow}
In this section, we explore the anti-shallowness of quantum states generated by low-resource adaptive circuits. Specifically, we examine the fidelity between states prepared by adaptive circuits and those produced by shallow non-adaptive circuits, as introduced in Section~\ref{ssc:approcircuitcomplexity}. By fixing the depth of the non-adaptive circuits to be constant, we define anti-shallowness as a measure of a state's distance from trivial states. This quantity is invariant under shallow circuit transformations and can thus serve as an indicator of distinct quantum phases.

To understand what value of anti-shallowness can be achieved by resource-constrained adaptive circuits like shallow-adaptive circuits, Section~\ref{ssc:shallowcomplexity} analyzes its achievable range and relates its maximal scaling to the number of ancillary qubits available to shallow adaptive circuits. Thus, a high anti-shallowness scaling implies a high circuit complexity. The maximal scaling is achieved by good QLDPC code states. Section~\ref{ssc:antishallow} relates the anti-shallowness to the correlation strength of a quantum state, and enables evaluating the scaling of the anti-shallowness of certain states. We then provide a classification of quantum states via anti-shallowness with illustrative examples. In particular, Section~\ref{ssc:QLDPCprepare} presents a scheme for preparing QLDPC states via shallow adaptive circuits. These states require high circuit depths when prepared non-adaptively.

\subsection{Anti-shallowness and phase indicator}\label{ssc:approcircuitcomplexity}
We begin by introducing the fidelity measure to quantify the distance between a target state $\ket{\psi}$ and the states prepared by depth-$D$ non-adaptive circuits:
\begin{equation}
F_{D,n}(\psi) = \max_{V:\mathrm{depth}(V)=D} \abs{\bra{\psi}V\ket{0^n}}^2,
\end{equation}
where $V$ is an $n$-qubit unitary implemented by a non-adaptive circuit of depth $D$. The quantity $F_{D,n}(\psi)$ reflects how well $\ket{\psi}$ can be approximated by such a circuit. We use $-\log F_{D,n}(\psi)$ as a measure of distance. When $D$ is constant, this value captures the distance from $\ket{\psi}$ to the set of trivial states, those preparable by shallow, non-adaptive circuits. In this case, we refer to the quantity as anti-shallowness and denote it simply by $-\log F(\psi)$.

We show that anti-shallowness serves as an indicator for distinguishing quantum phases: if two states have different scalings of anti-shallowness, they must belong to different phases. This follows from the fact that the fidelity is maximized over all trivial states.

To see this, suppose $\ket{\psi_2} = U\ket{\psi_1}$ for some shallow-circuit unitary $U$. Then,
\begin{equation}
\frac{-\log F(\psi_2)}{-\log F(\psi_1)} = \Theta(1),
\end{equation}
which shows that anti-shallowness is invariant under shallow unitaries, up to constant factors. The argument proceeds as follows:
\begin{equation}
\begin{split}
-\log F(\psi_2) &= \min_{V:\mathrm{depth}(V)=\Theta(1)} -\log\abs{\bra{\psi_2}V\ket{0^n}}^2\\
&= \min_{V:\mathrm{depth}(V)=\Theta(1)} -\log\abs{\bra{\psi_1}U^\dagger V\ket{0^n}}^2\\
&= \min_{V:\mathrm{depth}(V)=\Theta(1)} -\log\abs{\bra{\psi_1}V\ket{0^n}}^2\\
&= -\log F(\psi_1).
\end{split}
\end{equation}
Here, the equalities hold in the asymptotic sense, i.e., up to constant scaling. This completes the argument. Thus, anti-shallowness provides an effective tool for indicating quantum phases.

\subsection{Anti-shallowness of shallow-adaptive-circuit states}\label{ssc:shallowcomplexity}
It is important to note that the anti-shallowness of two states prepared using the same resources can differ in their scaling. To better understand the range of anti-shallowness of states prepared by adaptive circuits, we consider the minimum fidelity $F_{D,n}(\psi)$ over all $n$-qubit states $\ket{\psi}$ produced by an adaptive circuit of depth $L$ acting on $m$ input qubits:
\begin{equation}
F_{D,n}^{L,m} = \min_{\substack{\ket{\psi}=(\bra{\mathbf{s}}\otimes \id_n) U\ket{0^m}\\ \mathrm{depth}(\bra{\mathbf{s}}\otimes \id_n) U)=L,\mathbf{s}\in \{0,1\}^{m-n}}} F_{D,n}(\psi).
\end{equation}
Intuitively, $F_{D,n}^{L,m}$ describes how difficult non-adaptive circuits of depth $D$ are to approximate the most complicated $n$-qubit state prepared by adaptive circuits of depth $L$ with $m-n$ ancillas. In the following, we focus on the regime where $D = O(1)$ and $L = O(1)$, and establish quantitative results linking the power of shallow adaptive circuits to the number of ancillary qubits $m - n$. The more general case beyond the shallow-circuit regime is analyzed in Appendix~\ref{appendssc:thmpf:adaptpower}.

\begin{theorem}\label{thm:adaptpower}
Suppose $m-n = \omega(1)$. There exists an $n$-qubit state prepared by an $m$-qubit shallow adaptive circuit such that its distance to any $n$-qubit state prepared by an $n$-qubit shallow non-adaptive circuit is $\Omega(m-n)$. Meanwhile, the distance between any $n$-qubit state prepared by an $m$-qubit shallow adaptive circuit and an $n$-qubit shallow non-adaptive circuit is $O(m-n)$. Mathematically, $\exists L\in O(1), \forall D\in O(1)$,
\begin{equation}
-\log F_{D,n}^{L,m} = \Theta(\min(m-n, n)).
\end{equation}
If $m-n = O(1)$, then all shallow-adaptive-circuit states can be generated by shallow circuits. Mathematically, $\forall L\in O(1), \exists D\in O(1)$,
\begin{equation}\label{eq:Fconstantregion}
-\log F_{D,n}^{L,m} = 0.
\end{equation}
\end{theorem}

Theorem~\ref{thm:adaptpower} addresses two regimes of the ancillary qubit overhead $m-n$. When the number of ancillary qubits is constant $m - n = O(1)$, adaptive strategies offer no advantage in the shallow-circuit regime. This result is intuitive: the scaling of free parameters does not improve if only a constant number of ancillary qubits are used. To manifest the advantage of adaptive strategies, the number of ancillary qubits must scale with the system size.

In the regime where $m-n = \omega(1)$, Theorem~\ref{thm:adaptpower} provides both upper and lower bounds on fidelity: $-\log F_{D,n}^{L,m} = O(\min(m-n, n))$ and $-\log F_{D,n}^{L,m} = \Omega(\min(m-n, n))$. Note that for any $n$-qubit state $\ket{\psi}$, $F_{D,n}(\psi)$ can reach $2^{-n}$ by considering the inner product of all computational-basis states and $\ket{\psi}$. Thus, the upper bound of $-\log F_{D,n}^{L,m}$ is $n$. That is why one needs to minimize between $m-n$ and $n$. To show another upper bound of $-\log F_{D,n}^{L,m} = O(m-n)$, we utilize the idea of the light cone. Utilizing the expression of Eq.~\eqref{eq:adapt_state}, we map the states prepared by an $L$-depth adaptive circuit to a $\Theta((m-n)K^L)$-qubit state acted on by an $L$-depth non-adaptive circuit. The upper bound holds hereafter for $L=O(1)$. This idea also explains Eq.~\eqref{eq:Fconstantregion}.

To prove $-\log F_{D,n}^{L,m} = \Omega(\min(m-n, n))$, we construct a state prepared by a shallow adaptive circuit. Specifically, as shown in Section~\ref{ssc:QLDPCprepare}, with $m-n$ ancillary qubits, we can prepare a $\Theta(\min(m-n, n))$-qubit good QLDPC code state. When $m-n < n$, the additional qubits are all set as $\ket{0}$. The QLDPC code state is in a code space with distance $\Theta(\min(m-n, n))$. By generalizing the results of Ref.~\cite{bravyi2024entanglement}, we show that $-\log F_{D,n}^{L,m}$ is at least the same order as this code distance. The full proof of Theorem~\ref{thm:adaptpower} is shown in Appendix~\ref{appendsc:antishallow}.

The fact that a good QLDPC code state achieves the upper bound also confirms that preparing an $n$-qubit good QLDPC code state requires $\Omega(n)$ ancillary qubits for shallow adaptive circuits. This provides a circuit complexity lower bound for good QLDPC code states. For states with exponentially small infidelity to the $n$-qubit good QLDPC code states, the lower bound of the anti-shallowness is still $\Omega(n)$ by continuity of anti-shallowness. Particularly, given two states $\ket{\psi}$ and $\ket{\phi}$ such that $\abs{\braket{\phi}{\psi}}^2 = 1-\epsilon$, we can express $\ket{\phi} = \sqrt{1-\epsilon}\ket{\psi}+\sqrt{\epsilon}\ket{\psi^{\perp}}$, where $\ket{\psi^{\perp}}$ is orthogonal to $\ket{\psi}$. Then,
\begin{equation}
\begin{split}
-\log F(\phi) &= -\log \max_V \abs{\bra{\phi}V\ket{0^n}}^2\\
&=-\log \max_V \abs{\sqrt{1-\epsilon}\bra{\psi}V\ket{0^n}+\sqrt{\epsilon}\bra{\psi^{\perp}}V\ket{0^n}}^2\\
&= -\log \max_V [(1-\epsilon) \abs{\bra{\psi}V\ket{0^n}}^2 + \epsilon \abs{\bra{\psi^{\perp}}V\ket{0^n}}^2  + 2\sqrt{\epsilon(1-\epsilon)}\text{Re}(\bra{\psi}V\ket{0^n}\bra{\psi^{\perp}}V\ket{0^n})].
\end{split}
\end{equation}
Since $0\leq \abs{\bra{\psi}V\ket{0^n}}\leq 1$ and $0\leq \abs{\bra{\psi^{\perp}}V\ket{0^n}}\leq 1$, we have that
\begin{equation}
-\log F(\phi) \geq -\log[(1-\epsilon)2^{\log F(\psi)} + \epsilon + 2\sqrt{\epsilon(1-\epsilon)}].
\end{equation}
Thus, given $-\log F(\psi) = \Theta(n)$, we can ensure that $-\log F(\phi) = \Omega(n)$ if $\epsilon = \exp(-\Theta(n))$. As a result, the approximate preparation for good QLDPC code states with exponentially high fidelity still requires $\Omega(n)$ ancillary qubits. Following the proof of anti-shallowness lower bounds for $n$-qubit good QLDPC code states~\cite{bravyi2024entanglement}, one can also show that the approximate preparation for good QLDPC code states with constant fidelity requires $\Omega(\log n)$ depth. We provide more details in Appendix~\ref{appendsc:antishallow}.

From Theorem~\ref{thm:adaptpower}, we observe that the power of a shallow adaptive circuit is directly related to the number of qubits measured within the circuit. Specifically, the adaptive strategy only plays a role with a non-constant number of ancillary qubits. With respect to the fidelity measure, the distance between a state prepared by a shallow adaptive circuit and one prepared by a non-adaptive circuit is proportional to the number of measured qubits, as shown in Fig.~\ref{fig:VennAncilla}. This advantage comes from certain QLDPC code states $\ket{\psi}$, which are challenging to approximate using non-adaptive circuits. At the same time, the upper bound of $O(\min(m - n, n))$ limits the advantage of shallow adaptive circuits to scale proportionally with the number of measured qubits. A natural question is whether this relationship holds for general adaptive circuit states beyond the shallow regime.

\begin{figure}[htbp]
\centering
\includegraphics[width=6cm]{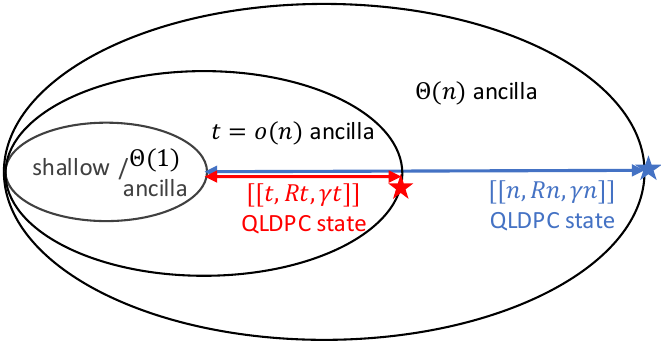}
\caption{The relationship between shallow-circuit states and shallow-adaptive-circuit states with $t$ ancillary qubits is illustrated. The shallow-circuit regime is equal to the ``$\Theta(1)$ ancilla'' regime. The maximum distance between the two classes, measured by the negative logarithm of fidelity, scales proportionally with the number of ancillary qubits. This maximum distance is achieved by a QLDPC code state with code parameters $[[t, Rt, \gamma t]]$, where $R$ and $\gamma$ are constants.}
\label{fig:VennAncilla}
\end{figure}

\subsection{Anti-shallowness and correlation strength}\label{ssc:antishallow}
In this part, we relate the anti-shallowness to the correlation strength of quantum states. This relation enables one to evaluate the scaling of anti-shallowness of quantum states. Based on the phase indicator property of anti-shallowness, we use this quantity to distinguish quantum states prepared by shallow-adaptive circuits. Notably, states prepared using different numbers of ancillary qubits belong to distinct quantum phases, e.g., states preparable with $o(n)$ ancillas and those requiring $\Omega(n)$ ancillas must lie in different phases. Therefore, we need to examine shallow-adaptive-circuit states under different scalings of ancillary qubits separately. As a representative case, we focus here on the regime where the number of ancillas scales as $\Omega(n)$, since this regime contains states that achieve a maximal scaling of anti-shallowness.

Though many quantum states need $\Omega(n)$ ancillas for preparation, their anti-shallowness can be much different. Here, we provide examples to show that the anti-shallowness in this regime can be linearly large or exponentially small. As discussed in the previous subsection, the maximal anti-shallowness scaling, i.e., $\Theta(n)$, is achieved by good QLDPC code states. Below this maximal scaling, we show that the anti-shallowness of the GHZ state and the hypergraph state $\ket{\mathrm{HG}} = CZ_n \ket{+}^{\otimes n}$ scales as $\Theta(1)$ and $e^{-\Theta(n)}$, respectively, as corollaries by Lemma~\ref{lemma:correlation}. As a remark, both the QLDPC and GHZ states lie exactly in the $\Theta(n)$-ancilla regime, as established by matching upper bounds from known shallow-adaptive state preparation protocols. Nonetheless, the hypergraph state currently has only an $O(n \log n)$ upper bound, so we can only conclude that it belongs to the $\Omega(n)$-ancilla regime.

In the following, we present the lemma showing a close relationship between anti-shallowness and global correlation of quantum states.

\begin{lemma}\label{lemma:correlation}
The anti-shallowness of an $n$-qubit state, $\ket{\psi}$, can be bounded with its global correlation:
\begin{equation}
-\log F(\psi)\geq -\log \left(1-\frac{1}{36}\mathrm{Cor}(\ket{\psi})^2\right)
\end{equation}
in the asymptotic case.
\end{lemma}

The proof utilizes the fact that the global correlation of states preparable by shallow non-adaptive circuits should be $0$. Detailed proof is shown in Appendix~\ref{appendssc:thmpf:correlation}. This lemma provides a tool for us to evaluate the lower bound of anti-shallowness for certain quantum states, as demonstrated below.

For the GHZ state, the global correlation satisfies $\mathrm{Cor}(\ket{\mathrm{GHZ}}) \geq \mathrm{Cor}(Z_i, Z_j, \ket{\mathrm{GHZ}}) = 1$, which implies a lower bound on the anti-shallowness: $-\log F(\ket{\mathrm{GHZ}}) \geq \log\frac{36}{35}$. On the other hand, its fidelity with the product state $\ket{0}^{\otimes n}$ provides an upper bound: $-\log F(\ket{\mathrm{GHZ}}) \leq \log 2$. Together, these bounds show that $-\log F(\ket{\mathrm{GHZ}}) = \Theta(1)$.

For hypergraph states, the global correlation $\mathrm{Cor}(\ket{\mathrm{HG}})\geq\mathrm{Cor}(X_i,X_j,\ket{\mathrm{HG}})=\frac{1}{2^{n-2}}\left(1-\frac{1}{2^{n-2}}\right)$, implying the anti-shallowness $-\log F(\ket{\mathrm{HG}})\geq -\log\left(1-\frac{1}{36\cdot 2^{2n-4}}\left(1-\frac{1}{2^{n-2}}\right)^2\right)$. On the other hand, the fidelity to $\ket{+}^{\otimes n}$ gives an upper bound $-\log F(\ket{\mathrm{HG}})\leq -\log\left(1-\frac{1}{2^{n-1}}\right)$. Thus, we conclude that $-\log F(\ket{\mathrm{HG}})=e^{-\Theta(n)}$.


We illustrate the anti-shallowness of the good QLDPC code state, GHZ state, and hypergraph state in Fig.~\ref{fig:antishallow}. Although we present only three representative examples, they reveal that the anti-shallowness can range from asymptotically maximal to vanishingly small. The gap between the extremes, $\Theta(n)$ and $e^{-\Theta(n)}$, can potentially be filled by constructing additional examples, which we leave for future work. The different scalings of anti-shallowness of these three examples indicate that they are in different quantum phases. This phase distinction is not feasible with state weight. In this sense, anti-shallowness provides a refined phase indication with state weight, and we expect more indicators to further explore quantum phase structure in the future.

\begin{figure}[htbp]
\centering
\includegraphics[width=7cm]{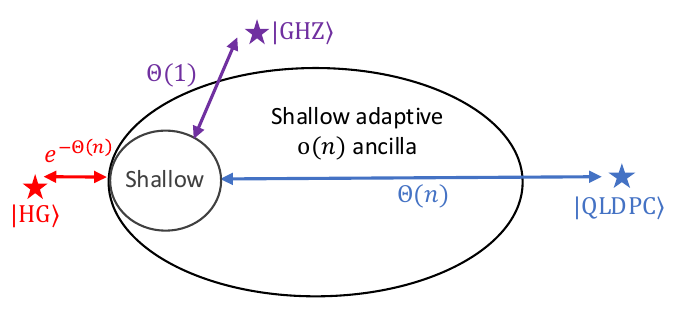}
\caption{Anti-shallowness of three different quantum states preparable by shallow adaptive circuits with $\Omega(n)$ ancillas. Their anti-shallownesses are $\Theta(n)$, $\Theta(1)$, and $e^{-\Theta(n)}$ for the good QLDPC code state, GHZ state, and hypergraph state, respectively.}
\label{fig:antishallow}
\end{figure}

\subsection{Efficient QLDPC state preparation with shallow adaptive circuits}\label{ssc:QLDPCprepare}
In this part, we construct protocols for deterministically preparing QLDPC code states with shallow adaptive circuits, which contributes to the proof of Theorem~\ref{thm:adaptpower}.
Note that the deterministic state preparation protocol is stronger than the probabilistic case.
The result of this part might be additionally interesting in quantum computing.

Particularly, our state preparation scheme is not restricted to QLDPC code states but includes this case. The scheme is preparing an $n$-qubit stabilizer state with an $m$-qubit shallow adaptive circuit. In this scheme, $m \leq 2n$, and the measurement is at the end of the circuit. The target state $\ket{\psi}$ has to satisfy the conditions below.

Given the $n$-qubit stabilizer state, $\ket{\psi}$, with stabilizer generator set as $\mathcal{S}_{\psi} = \{S_1, S_2, \cdots, S_n\}$, one can divide the stabilizer generators into two disjoint groups, $\mathcal{S}_{\psi} = \mathcal{S}_{\psi}^1 + \mathcal{S}_{\psi}^2$, such that
\begin{enumerate}
\item $\forall S_i\in \mathcal{S}_{\psi}^1$, $\wt(S_i)\leq s$ with $s = O(1)$;
\item There exists a trivial state, $\ket{\phi}$, or a state with a constant circuit complexity, in the code space defined by stabilizer generators within $\mathcal{S}_{\psi}^2$.
\end{enumerate}

Note that the set $\mathcal{S}_{\psi}^1$ defines a quantum stabilizer LDPC code, $\mathcal{C}$, with itself as the stabilizer check. Denote $\abs{\mathcal{S}_{\psi}^2} = k$, then $\abs{\mathcal{S}_{\psi}^1} = n-k$, and the stabilizer code defined by $\mathcal{S}_{\psi}^1$ is a $[[n, k]]$ code. Without loss of generality, let $\mathcal{S}_{\psi}^1 = \{S_1, S_2, \cdots, S_{n-k}\}$ composed of the former $n-k$ stabilizers in $\mathcal{S}_{\psi}$. The target state $\ket{\psi}$ is a state in the code space. If $k = 0$, the stabilizers of the target state $\ket{\psi}$ are all sparse. The following lemma shows that all the stabilizers in $\mathcal{S}_{\psi}^1$ can be measured simultaneously in an $O(s^2)$-depth circuit with $n-k$ ancillary qubits.

\begin{lemma}\label{lemma:qldpcmeasure}
Given a set of independent $n$-qubit Pauli operators, $\mathcal{S} = \{S_1, S_2, \cdots, S_t\}$, with $\forall j, \wt(S_j)\leq s$ and $\forall i, \# S^i = \abs{\{S_j | S_j(i)\neq \id\}}\leq s$, one can design a Clifford circuit with depth $2+s+s^2$ and $t$ ancillary qubits to measure all Pauli operators simultaneously. Here, $S_j(i)$ denotes the Pauli gate on the $i$-th qubit of $S_j$, and $\# S^i$ represents the number of Pauli operators acting nontrivially on qubit $i$.
\end{lemma}

The detailed proof of Lemma~\ref{lemma:qldpcmeasure} is a construction shown in Appendix~\ref{appendsc:lemmapf:qldpcmeasure}. With Lemma~\ref{lemma:qldpcmeasure}, we prepare $\ket{\psi}$ using ancillas and measurement feedback with the following procedure.
\begin{enumerate}
\item Prepare the $n$-qubit state $\ket{\phi}$ with a shallow circuit.
\item Measure each stabilizer generator in $\mathcal{S}_{\psi}^1$ with an independent ancillary qubit and obtain the syndrome. This step can be done in a depth $O(s^2) = O(1)$.\label{item:measure}
\item Correct the state with local Pauli operations according to the error syndrome.
\end{enumerate}

Note that initially, the state $\ket{\phi}$ is stabilized by the stabilizers in $\mathcal{S}_{\psi}^2$. Since all stabilizers in $\mathcal{S}_{\psi}^1$ commute with those in $\mathcal{S}_{\psi}^2$, after measurement, the state is still stabilized by the stabilizers in $\mathcal{S}_{\psi}^2$. Denote the measurement outcome of stabilizer $S_i$ as $s_i$ where $s_i\in \{\pm 1\}$, the stabilizer generators of the state after step~\ref{item:measure} are $\{s_iS_i| S_i\in \mathcal{S}_{\psi}^1 \}\cup \mathcal{S}_{\psi}^2$. We further denote $\mathcal{S}_{\psi}^+ = \{S_i| S_i\in \mathcal{S}_{\psi}^1, s_i=1 \}\cup \mathcal{S}_{\psi}^2$ and $\mathcal{S}_{\psi}^- = \{S_i| S_i\in \mathcal{S}_{\psi}^1, s_i=-1\}$. The local Pauli correction operation $P_c$ satisfies the following conditions,
\begin{align}
\forall S^+\in \mathcal{S}_{\psi}^+, [P_c, S^+] &= 0;\\
\forall S^-\in \mathcal{S}_{\psi}^-, \{P_c, S^-\} &= 0.
\end{align}
Note that $\mathcal{S}_{\psi}^+$ and $\mathcal{S}_{\psi}^-$ are both known sets of $n$-qubit Pauli operators. Finding a Pauli operator satisfying the above conditions is equivalent to solving a binary linear equation and can be done with a classical algorithm of complexity $O(n^2)$. This operator is unique in the multiplication of stabilizer generators. We provide more discussions in Appendix~\ref{appendsc:decode}.

After applying $P_c$ to the measured state after step~\ref{item:measure}, the final state will be stabilized by all stabilizers in $\mathcal{S}_{\psi}^+\cup \mathcal{S}_{\psi}^- = \mathcal{S}_{\psi}$. Equivalently, the final state is in the common eigenspace of $\mathcal{S}_{\psi}$ with all eigenvalues $+1$ and hence equal to $\ket{\psi}$. Since all three steps can be done at a constant depth, the whole circuit is shallow. Note that the number of ancillas equals $\abs{\mathcal{S}_{\psi}^1}\leq n$. Thus, the total number of qubits of the adaptive circuit $m = \abs{\mathcal{S}_{\psi}^1} + n$ satisfies $m \leq 2n$. Two special cases are when $m = n$ and $m = 2n$. In the former case, the state $\ket{\phi}$ is just $\ket{\psi}$ since $\mathcal{S}_{\psi}^2$ determines a unique state when $\abs{\mathcal{S}_{\psi}^2} = n$. In the latter case, one just prepares $\ket{0^n}$ as $\ket{\phi}$ and measure all stabilizers in $\mathcal{S}_{\psi}^1$ to get target state $\ket{\psi}$.

It is worth noting that for any QLDPC stabilizer code $\mathcal{C}$, there exists a logical state $\ket{\psi}_L\in \mathcal{C}$ satisfying the two conditions of shallow adaptive circuit preparation. That is, one part of stabilizer generators has constant weight, and the other part of stabilizer generators stabilizes a trivial state.

The QLDPC code checks form the first set $\mathcal{S}_{\psi}^1$, with each check having weight bounded by the sparsity $s$. The second set, $\mathcal{S}_{\psi}^2$, can be chosen as a subset of either $\{\id, X\}^{\otimes n}$ or $\{\id, Z\}^{\otimes n}$. For example, taking $\{\id, X\}^{\otimes n}$, it suffices to show the existence of $k$ independent Pauli operators in this set that commute with, yet are independent of, the elements in $\mathcal{S}_{\psi}^1$. We prove such a construction always exists in Appendix~\ref{appendsc:xtypelogical}. The trivial state $\ket{\phi}$ can then be taken as $\ket{+}^{\otimes n}$.


More generally, if the measurement is not restricted to the last step, the initial state $\ket{\phi}$ is not restricted to being prepared by a shallow non-adaptive circuit. We can first prepare $\ket{\phi}$ with a shallow adaptive circuit and then prepare $\ket{\psi}$ with another shallow adaptive circuit. This step can be nested for constant times, and the total circuit maintains a shallow adaptive circuit.

\section{Discussion}\label{sc:discussion}
In this work, we introduce two key quantities, state weight and anti-shallowness, and show their close relationships with correlation functions. We use these quantities to provide an approximate circuit complexity lower bound in an adaptive setting and identify quantum phases of given quantum states. Important examples include GHZ states, W states, hypergraph states, and good QLDPC code states. These quantities are also phase indicators that distinguish quantum states from different phases.

Beyond theoretical analysis, our explicit preparation schemes for QLDPC code states hold promise for practical implementation and fault-tolerant realization~\cite{Bergamaschi2025Fault}. In particular, our method for preparing certain stabilizer states within constant depth using $O(n)$ ancillary qubits offers a viable path toward efficient and large-scale quantum computing, where minimizing circuit depth is essential for mitigating noise and decoherence. In practice, one can also flexibly balance ancillary qubit usage and circuit depth based on the capabilities of the underlying hardware.

Several open questions remain for future research. First, it is interesting to extend and improve our results about approximate circuit complexity~\cite{aaronson2016complexityquantumstatestransformations,Brandao2021Complexity,du2024embeddedcomplexityquantumcircuit} in adaptive circuits, the minimum resources of a circuit needed to prepare a state within a small error of the target. In this study, we have already shown a resource lower bound for permutation-invariant states and QLDPC code states, and the tolerable approximate error is different for different states. For GHZ states, the lower bound still holds for a constant infidelity. In contrast, for W states, the lower bound only holds for an $O(\frac{1}{n^{\alpha}})$-scaling infidelity where $\alpha$ is a constant. It is interesting to improve this approximation error to a better scaling, like a constant. The same questions remain for many other states. Studying trade-off relations and preparation protocols in an approximate regime~\cite{piroli2024approximatingmanybodyquantumstates} with different approximation errors, and how they behave under realistic noise, is essential for bridging theory with experiment.

In addition, the existing trade-off bounds may not be tight. For instance, in Theorem~\ref{theo:adalowerboundgeneric}, the lightcone size with measurements scales as $K^{2L-1}$, compared to $K^L$ without measurements. Whether these bounds can be unified or tightened remains an open question. Although we have identified specific examples, such as GHZ states and QLDPC code states, that saturate the trade-offs, it is unclear whether these bounds hold universally or if there exist states that do not saturate them. Investigating such cases could yield deeper insights into the study of adaptive quantum circuits.

Finally, exploring more indicators beyond state weight and anti-shallowness to indicate quantum phases is a promising direction. Such refinement is particularly important for the regime of states with weight $\Theta(n)$ and shallow-adaptive-circuit states that use $\Omega(n)$ ancillas. This regime contains a rich variety of states, including GHZ states and Dicke states, and exhibits diverse entanglement structures. However, different phases cannot be distinguished using state weight in this regime, and anti-shallowness does not suffice for complete phase classification. Moreover, any classification method based solely on the size of the lightcone does not distinguish the phases within this regime, since the lightcone criteria cannot distinguish between states that require $\Theta(n)$ ancillas and those that require $\omega(n)$ ancillas. To resolve this, new tools beyond anti-shallowness are needed to capture the refined phase structure within this class of states.

\begin{acknowledgements}
We express our gratitude to Yuxuan Yan, You Zhou, Xingjian Zhang, and Nathanan Tantivasadakarn for the insightful discussions. This work was supported by the National Natural Science Foundation of China Grant No.~12174216 and the Innovation Program for Quantum Science and Technology Grant No.~2021ZD0300804 and No.~2021ZD0300702.
\end{acknowledgements}

\newpage

\appendix

\section{Equivalent definitions of generalized stabilizers}\label{appendsc:genestabilizer}
In this appendix, we prove the following claim: for the group $\mathbb{T}_{\psi} = \langle T_1, T_2, \ldots, T_n \rangle$, any element $T_j \in \mathbb{T}_{\psi} \backslash \mathbb{I}$ is a generalized stabilizer if and only if there exists a unitary $U$ such that $\forall i\in [n]$, $T_i = U Z_i U^\dagger$, where $Z_i$ is the Pauli-$Z$ operator on qubit $i$.

First, suppose $T_i=UZ_iU^{\dagger}$ stabilizes the state, $T_i\ket{\psi}=\ket{\psi}$ for all $i$. One can check that $T_i=T_i^{\dagger}=T_i^{-1}$ and $\tr{T_i}=\tr{Z_i}=0$ hold for all $i$, and these $T_i$ mutually commute with each other. In addition, any element generated by $T_i$, $T(\vec{t})=\prod_j T_i^{t_i}=U\left(\prod_i Z_i^{t_i}\right)U^{\dagger}$, also stabilizes the state $\ket{\psi}$ and satisfies those conditions.

Next, suppose the operators $T_i$ satisfy $T_i = T_i^\dagger = T_i^{-1}$ and $\tr{T_i} = 0$, and each stabilizes the state $\ket{\psi}$. Since the $T_i$ are Hermitian and mutually commute, there exists a unitary $U'$ that simultaneously diagonalizes them: $T'_i = U' T_i U'^\dagger$ are diagonal. Given $T_i^2 = I$ and $\tr{T_i} = 0$, each $T'_i$ has eigenvalues $\pm 1$ in equal numbers, i.e., they are balanced.

We now show how to iteratively transform $T'_i$ into the Pauli $Z_i$ operators. Since $T'_1$ and $Z_1$ share the same spectrum, there exists a permutation matrix $P_1$ such that
\begin{equation}
T'_1 = P_1^\dagger Z_1 P_1.
\end{equation}
For $T'_2$, decompose it as
\begin{equation}
T'_2 = \frac{I + T'_1}{2} T'_2 + \frac{I - T'_1}{2} T'_2,
\end{equation}
acting on the $\pm 1$ eigenspaces of $T'_1$. Using $\tr{T'_1 T'_2} = \tr{T'_2} = 0$, we deduce that both components are traceless and balanced. After conjugating by $P_1$, we write
\begin{equation}
P_1 T'_2 P_1^\dagger = \frac{I + Z_1}{2} P_1 T'_2 P_1^\dagger + \frac{I - Z_1}{2} P_1 T'_2 P_1^\dagger \equiv T'_{2,+} + T'_{2,-},
\end{equation}
where $T'_{2,+}$ and $T'_{2,-}$ act on the $\pm 1$ eigenspaces of $Z_1$ and remain balanced with $\pm 1$ eigenvalues.

Similarly, decompose $Z_2$ as
\begin{equation}
Z_2 = \frac{I + Z_1}{2} Z_2 + \frac{I - Z_1}{2} Z_2,
\end{equation}
with each part also balanced. We can then find permutation matrices $P_{2,+}$ and $P_{2,-}$, acting nontrivially only on the $+1$ and $-1$ eigenspaces of $Z_1$, respectively, such that
\begin{equation}
P_{2,+} T'_{2,+} P_{2,+}^\dagger = \frac{I + Z_1}{2} Z_2,
\end{equation}
\begin{equation}
P_{2,-} T'_{2,-} P_{2,-}^\dagger = \frac{I - Z_1}{2} Z_2.
\end{equation}
Letting $P_2 = P_{2,+} P_{2,-}$, where $P_{2,+}$ and $P_{2,-}$ commute since they act on orthogonal subspaces, we find
\begin{equation}
P_2 P_1 T'_2 P_1^\dagger P_2^\dagger = Z_2,\quad P_2 P_1 T'_1 P_1^\dagger P_2^\dagger = Z_1.
\end{equation}
Thus, $P_2 P_1$ transforms $\{T'_1, T'_2\}$ to $\{Z_1, Z_2\}$. Repeating this procedure, we construct permutation matrices $P_3, P_4, \ldots$ such that the product $P_n \cdots P_1$ maps $\{T'_1, T'_2, \ldots, T'_n\}$ to $\{Z_1, Z_2, \ldots, Z_n\}$.


\section{Proof of trade-off theorems related to state weight and correlation range}\label{appendsc:weight}
\subsection{Non-adaptive strategy}\label{ssc:nonadt}
We begin by establishing a lower bound on the circuit depth required to prepare any quantum state using non-adaptive circuits, based on the weight of the state. The result is summarized below.

\begin{proposition}\label{prop:nonadaptivegeneral}
Preparing an $n$-qubit state $\ket{\psi}$ with $K$-bounded fan-in gates in a non-adaptive circuit necessitates circuit depth
\begin{equation}\label{eq:boundnonadaptive}
D\geq \frac{\log{\wt(\ket{\psi})}}{\log{K}}.
\end{equation}
\end{proposition}

Proposition~\ref{prop:nonadaptivegeneral} has been proved in the main text. For stabilizer states, we can derive a tighter lower bound on circuit depth using the stabilizer weight. However, since the preparation circuit $U$ is not necessarily Clifford, the proof of Proposition~\ref{prop:nonadaptivegeneral} or Lemma~\ref{lemma:operatorweightgrowth} does not directly apply. In particular, $\ket{\psi} = U\ket{\mathbf{0}}$ does not guarantee that $UZ_iU^\dagger$ corresponds to a stabilizer generator of $\ket{\psi}$, so a different proof is needed.

\begin{proposition}\label{prop:nonadaptivegeneralClifford}
Preparing $n$-qubit stabilizer state $\ket{\psi}$ with $K$-bounded fan-in gates in a non-adaptive circuit necessitates circuit depth
\begin{equation}\label{eq:boundnonadaptivestab}
D\geq \frac{\log{\wts(\ket{\psi})}}{\log{K}}.
\end{equation}
\end{proposition}

\begin{proof}
Our proof utilizes the following lemma introduced in~\cite{onorati2023efficient,yu2023learning}.

\begin{lemma}\label{le:marginal}\cite{yu2023learning}
Suppose $n$-qubit quantum state $\ket{\psi}$ has non-adaptive circuit complexity $D$. For any $n$-qubit state $\rho$, either $\norm{\psi-\rho}_1<\epsilon$ or $\norm{\psi_s-\rho_s}>\frac{\epsilon^2}{4n}$ for some subset $s\subset[n]$ with $|s|=K^D$, where $\psi=\ketbra{\psi}$ and $\psi_s$ and $\rho_s$ denote the reduced density matrices on the qubits whose labels are in $s$.
\end{lemma}

To apply the lemma and identify a lower bound on the circuit complexity of $\ket{\psi}$, we construct a state that is locally indistinguishable from $\ket{\psi}$ but globally distinct. The construction proceeds as follows.

Given a stabilizer state $\ket{\psi}$ with minimal stabilizer generators $\{S_1^*, S_2^*, \cdots, S_n^*\}$, we define a new stabilizer state $\ket{\psi_-}$ whose stabilizer generators are $\{-S_1^*, S_2^*, \cdots, S_n^*\}$. Suppose the non-adaptive circuit complexity of $\ket{\psi}$ is $D$. Then, as we show below, any $k$-local observable cannot distinguish between $\ket{\psi}$ and $\ket{\psi_-}$ as long as $k < \wts(\ket{\psi})$.
\begin{equation}
\begin{split}
\tr_{\Bar{k}}(\ketbra{\psi} - \ketbra{\psi_-}) &= \tr_{\Bar{k}}\left(\frac{1}{2^n}\prod_{i=1}^n(\mathbb{I}+S_i^*) - \frac{1}{2^n}(\mathbb{I}-S_1^*)\prod_{i=2}^{n}(\mathbb{I}+S_i^*)\right)\\
&= \frac{1}{2^{n-1}}\tr_{\Bar{k}}\left(S_1^*\prod_{i=2}^{n}(\mathbb{I}+S_i^*)\right)\\
&= \frac{1}{2^{n-1}}\sum_{j_2, j_3, \cdots, j_{n}=0}^1\tr_{\Bar{k}}\left(S_1^* \prod_{i=2}^{n} S_i^{*j_i}\right)\\
&= 0.
\end{split}
\end{equation}
Here, in the fourth line, we utilize the fact that for any list $\{j_2, j_3, \cdots, j_{n}\}$, $\wts(S_1^*\prod_{i=2}^{n} S_i^{*j_i}) \geq \wts(\ket{\psi}) = \wts(S^*_1) > k$. That is, all the $S_1^*\prod_{i=2}^{n} S_i^{*j_i}$ has support in any $\Bar{k}$, and therefore the trace will be zero. The fact can be proved by contradiction. Suppose there exists an operator $\Tilde{S_1^*}\equiv S_1^*\prod_{i=2}^{n} S_i^{*j_i}$ such that $\wts(\Tilde{S_1^*}) < \wts(S^*_1)$.
Since $\Tilde{S_1^*}$ is still a valid stabilizer generator independent of the others, we can construct a new weight vector $\Tilde{\mathcal{S}^*}_{\psi}$ by replacing $S^*_1$ with $\Tilde{S_1^*}$. It then follows that $\wts(\Tilde{\mathcal{S}^*}_{\psi})<\wts(\mathcal{S}^*_{\psi})$, which is a contradiction to Eq.~\eqref{eq:minwtvec}.

Notice that $\norm{\psi_{-}-\psi}_1>0$ since they are different states. Let $\epsilon=\norm{\psi_{-}-\psi}_1$, and the requirement in Lemma~\ref{le:marginal} is not satisfied when $|s|<\wts(\ket{\psi})$. This implies that $K^D\geq\wts(\ket{\psi})$ and $D\geq\frac{\log{\wts(\ket{\psi})}}{\log{K}}$.
\end{proof}

\subsection{Proof of Theorem~\ref{theo:adalowerboundgeneric}}\label{appendsc:thmpf:adalowerboundgeneric}
Recall Eq.~\eqref{eq:adapt_state}, any state $\ket{\psi}$ prepared by the depth-$L$ adaptive circuit has the following expression.
\begin{equation}
\ket{\psi} \propto (\id_n\otimes \bra{\mathbf{s}}) U_{L'}\cdots U_2U_1\ket{0^m},
\end{equation}
where $\id_n$ denotes the identity operator acting on the first $n$ qubits; $\mathbf{s}=s_{n+1}s_{n+2}\cdots s_{m}$ denotes all measurement outcomes; $U_{L'},\cdots, U_2, U_1$ are all single layers of $K$-bounded fan-in quantum gates, which depend on the measurement outcomes $s_1, s_2, \cdots, s_{m-n}$. Note that $L'\leq L$. In the following, we use the following expression of $\ket{\psi}$:
\begin{equation}\label{eq:adapt_state_short}
\ket{\psi} = c(\id_n\otimes \bra{\mathbf{s}})\left(\prod_{i=1}^L U_i\right)\ket{0^m},
\end{equation}
where $c$ is a normalization factor and $\prod_{i=1}^L U_i$ is short for $U_L\cdots U_2U_1$. Note that we can always set $U_{L'+1}, U_{L'+2}, \cdots, U_L$ as identity operations.

Then, we consider the backward lightcone from the measured qubits to the first layer. Formally, for each layer of unitary $U_i$, we recursively decompose the qubits into two disjoint parts $[m]=Q_i^{(1)}+Q_i^{(2)}$ such that $U_i$ can be decomposed as $U_i = U_i^{(1)}\otimes U_i^{(2)}$, where $U_i^{(1)}$ acts only on $Q_i^{(1)}$, and $U_i^{(2)}$ acts only on $Q_i^{(2)}$. The decomposition is recursively defined with the constraints of $Q_i^{(1)}\cap Q_{i+1}^{(2)}=\emptyset$, and making $Q_i^{(2)}$ as small as possible. The recursion starts from the last layer and gradually constructs layer $i$ from layer $i+1$. The initialization is given by $Q_{L+1}^{(2)}=\{n+1,\cdots,m\}$. An example of such decomposition is shown in Fig.~\ref{fig:measurelightcone}. Notice that we also use $U_i^{(1)}$ ($U_i^{(2)}$) to denote $U_i^{(1)}\otimes\id_i^{(2)}$ ($\id_i^{(1)}\otimes U_i^{(2)}$) for simplicity. Since $Q_{i+1}^{(2)}$ and $Q_i^{(1)}$ are disjoint, the corresponding $U_{i+1}^{(2)}$ and $U_i^{(1)}$ commute with each other. Furthermore, $U_{i+1}^{(2)}$ should also commute with all $U_j^{(1)}$ with $j<i$ since $Q_j^{(1)}$ should be a subset of $Q_i^{(1)}$ and also be disjoint to $Q_{i+1}^{(2)}$. Therefore,
\begin{equation}\label{eq:decomposeU1U2}
\begin{split}
\ket{\psi} &= c(\id_n\otimes \bra{\mathbf{s}}) \prod_{i=1}^L (U_i^{(1)}\cdot U_i^{(2)}) \ket{0^m}\\
&= c\left(\prod_{i=1}^L U_i^{(1)}\right) (\id_n\otimes \bra{\mathbf{s}}) \left(\prod_{i=1}^L U_i^{(2)}\right) \ket{0^m}.
\end{split}
\end{equation}
Note that $(\id_n\otimes \bra{\mathbf{s}}) \left(\prod_{i=1}^L U_i^{(2)}\right)$ only acts on the qubits in $\bigcup_{i=1}^{L+1}Q_i^{(2)}=Q_1^{(2)}$. To estimate the size of $Q_1^{(2)}$, we come back to analyze the conditions of the decomposition $[m]=Q_i^{(1)}+Q_i^{(2)}$. To give the smallest $Q_i^{(2)}$, the support of each nontrivial unitary in $U_i^{(2)}$ should contain at least one qubit in $Q_{i+1}^{(2)}$, and the support of any unitary in $U_i^{(1)}$ does not. It shows that $|Q_{i+1}^{(2)}|\leq|Q_i^{(2)}|\leq K|Q_{i+1}^{(2)}|$, and $|Q_1^{(2)}|\leq K^L|Q_{L+1}^{(2)}|=K^L(m-n)$.

\begin{figure}[!htbp]
 \centering
  \begin{quantikz}[row sep={0.6cm,between origins}, column sep=0.5cm]
  \lstick{$\ket{0}$} & \gate[wires=2]{} & \qw & \gate[wires=2]{} & \qw & \qw\\
  \lstick{$\ket{0}$} &  & \gate[wires=2]{} &  & \gate[wires=2]{} & \qw  \\
  \lstick{$\ket{0}$} & \gate[wires=2]{} &  & \gate[wires=2]{} &  & \qw  \\
  \lstick{$\ket{0}$} &  & \gate[wires=2]{} &  & \gate[wires=2]{} & \qw  \\
  \lstick{$\ket{0}$} & \gate[style={fill=gray!30},wires=2]{} &  & \gate[wires=2]{} &  & \qw  \\
  \lstick{$\ket{0}$} &  & \gate[style={fill=gray!30},wires=2]{} &  & \gate[wires=2]{} & \qw  \\
  \lstick{$\ket{0}$} & \gate[style={fill=gray!30},wires=2]{} &  & \gate[style={fill=gray!30},wires=2]{} &  & \qw  \\
  \lstick{$\ket{0}$} &  & \gate[style={fill=gray!30},wires=2]{} &  & \gate[style={fill=gray!30},wires=2]{} & \qw  \\
  \lstick{$\ket{0}$} & \gate[style={fill=gray!30},wires=2]{} &  & \gate[style={fill=gray!30},wires=2]{} &  & \meter{} \\
  \lstick{$\ket{0}$} &  & \gate[style={fill=gray!30},wires=2]{} &  & \gate[style={fill=gray!30},wires=2]{} & \meter{}  \\
  \lstick{$\ket{0}$} & \gate[style={fill=gray!30},wires=2]{} &  & \gate[style={fill=gray!30},wires=2]{} &  & \qw  \\
  \lstick{$\ket{0}$} &  & \gate[style={fill=gray!30},wires=2]{} &  & \gate[wires=2]{} & \qw  \\
  \lstick{$\ket{0}$} & \gate[style={fill=gray!30},wires=2]{} &  & \gate[wires=2]{} &  & \qw  \\
  \lstick{$\ket{0}$} &  & \gate[wires=2]{} &  & \gate[wires=2]{} & \qw  \\
  \lstick{$\ket{0}$} & \gate[wires=2]{} &  & \gate[wires=2]{} &  & \qw  \\
  \lstick{$\ket{0}$} &  & \gate[wires=2]{} &  & \gate[wires=2]{} & \qw  \\
  \lstick{$\ket{0}$} & \gate[wires=2]{} &  & \gate[wires=2]{} &  & \qw  \\
  \lstick{$\ket{0}$} &  & \qw &  & \qw & \qw  \\
  \end{quantikz}
  \begin{tikzpicture}[overlay,x=4mm,y=4mm]
    \draw[red, dashed, very thick]
    (-10.1,-6.95)
    -- ++(0,15.15)
    -- ++(2,0) -- ++(0,-1.48)
    -- ++(2,0) -- ++(0,-1.48)
    -- ++(2,0) -- ++(0,-1.48)
    -- ++(2,0) -- ++(0,-1.48)
    -- ++(2.4,0)
    -- ++(0,-3.31)
    -- ++(-2.4,0) -- ++(0,-1.48)
    -- ++(-2,0) -- ++(0,-1.48)
    -- ++(-2,0) -- ++(0,-1.48)
    -- ++(-2,0) -- ++(0,-1.48)
    -- cycle;
  \end{tikzpicture}
  \hspace{5mm}
  $\Longrightarrow$
  \hspace{5mm}
  \begin{quantikz}[row sep={0.6cm,between origins}, column sep=0.5cm]
  \lstick{$\ket{0}$} & \gate[wires=2]{} & \qw & \gate[wires=2]{} & \qw & \qw\\
  \lstick{$\ket{0}$} &  & \gate[wires=2]{} &  & \gate[wires=2]{} & \qw  \\
  \lstick{$\ket{0}$} & \gate[wires=2]{} &  & \gate[wires=2]{} &  & \qw  \\
  \lstick{$\ket{0}$} &  & \gate[wires=2]{} &  & \gate[wires=2]{} & \qw  \\
  \lstick{$\ket{0}$} & \gate[style={fill=gray!30},wires=10]{} &  & \gate[wires=2]{} &  & \qw  \\
  \lstick{$\ket{0}$} &  & \qw  &  & \gate[wires=2]{} & \qw  \\
  \lstick{$\ket{0}$} &  & \qw  & \qw  &  & \qw  \\
  \lstick{$\ket{0}$} &  & \qw  & \qw  & \qw  & \qw  \\
  \lstick{$\ket{0}$} &  \\
  \lstick{$\ket{0}$} &  \\
  \lstick{$\ket{0}$} &  & \qw  & \qw  & \qw  & \qw  \\
  \lstick{$\ket{0}$} &  & \qw  & \qw  & \gate[wires=2]{} & \qw  \\
  \lstick{$\ket{0}$} &  & \qw  & \gate[wires=2]{} &  & \qw  \\
  \lstick{$\ket{0}$} &  & \gate[wires=2]{} &  & \gate[wires=2]{} & \qw  \\
  \lstick{$\ket{0}$} & \gate[wires=2]{} &  & \gate[wires=2]{} &  & \qw  \\
  \lstick{$\ket{0}$} &  & \gate[wires=2]{} &  & \gate[wires=2]{} & \qw  \\
  \lstick{$\ket{0}$} & \gate[wires=2]{} &  & \gate[wires=2]{} &  & \qw  \\
  \lstick{$\ket{0}$} &  & \qw &  & \qw & \qw  \\
  \end{quantikz}
  \begin{tikzpicture}[overlay,x=4mm,y=4mm]
    \draw[red, dashed, very thick] (-10.5,13.85) -- ++(0,-2.55) -- ++(3.6,0) -- ++(0,2.55) -- cycle;
    \draw[red, dashed, very thick] (-10.5,10.85) -- ++(0,-2.55) -- ++(3.6,0) -- ++(0,2.55) -- cycle;
    \draw[red, dashed, very thick] (-10.5,7.85) -- ++(0,-14.55) -- ++(3.6,0) -- ++(0,14.55) -- cycle;
    \draw[red, dashed, very thick] (-10.5,-7.15) -- ++(0,-2.55) -- ++(3.6,0) -- ++(0,2.55) -- cycle;
    \draw[red, dashed, very thick] (-10.5,-10.15) -- ++(0,-2.55) -- ++(3.6,0) -- ++(0,2.55) -- cycle;
  \end{tikzpicture}
  \hspace{5mm}
  $\Longrightarrow$
  \hspace{5mm}
  \begin{quantikz}[row sep={0.6cm,between origins}, column sep=0.5cm]
  & \qw & \gate[wires=2]{} & \qw & \qw\\
  & \gate[wires=2]{} &  & \gate[wires=2]{} & \qw  \\
  &  & \gate[wires=2]{} &  & \qw  \\
  & \gate[wires=2]{} &  & \gate[wires=2]{} & \qw  \\
  &  & \gate[wires=2]{} &  & \qw  \\
  & \qw  &  & \gate[wires=2]{} & \qw  \\
  & \qw  & \qw  &  & \qw  \\
  & \qw  & \qw  & \qw  & \qw  \\
  \\
  \\
  & \qw  & \qw  & \qw  & \qw  \\
  & \qw  & \qw  & \gate[wires=2]{} & \qw  \\
  & \qw  & \gate[wires=2]{} &  & \qw  \\
  & \gate[wires=2]{} &  & \gate[wires=2]{} & \qw  \\
  &  & \gate[wires=2]{} &  & \qw  \\
  & \gate[wires=2]{} &  & \gate[wires=2]{} & \qw  \\
  &  & \gate[wires=2]{} &  & \qw  \\
  & \qw &  & \qw & \qw  \\
  \end{quantikz}
  \begin{tikzpicture}[overlay,x=4mm,y=4mm]
    \draw[black, thick] (-7.15,13.8) -- ++(0,-2.5) -- ++(-0.8,1.25) -- cycle;
    \draw[black, thick] (-7.15,10.8) -- ++(0,-2.5) -- ++(-0.8,1.25) -- cycle;
    \draw[black, thick] (-7.15,7.8) -- ++(0,-14.5) -- ++(-0.8,7.25) -- cycle;
    \draw[black, thick] (-7.15,-7.2) -- ++(0,-2.5) -- ++(-0.8,1.25) -- cycle;
    \draw[black, thick] (-7.15,-10.2) -- ++(0,-2.5) -- ++(-0.8,1.25) -- cycle;
    \draw[black, dashed, thick] (-6.8,8.3) -- ++(6.3,4.7);
    \draw[black, dashed, thick] (-6.8,-7.2) -- ++(6.3,-4.7);
    \begin{scope}[on background layer]
    \fill[gray, fill opacity=0.3, draw=none] (-7.15,7.8) -- ++(0,-14.5) -- ++(-0.8,7.25) -- cycle;
    \fill[yellow, fill opacity=0.3, draw=none] (-6.8,8.3) -- (-6.8,-7.2) -- (-0.5,-11.9) -- (-0.5,13) -- cycle;
    \end{scope}
  \end{tikzpicture}
\caption{``Backward lightcone" from the measurements and ``forward lightcone" from the first layer. In the circuit on the left, the gates that constitute $U_i^{(1)}$ are shown as white blocks, whereas those that constitute \(U_i^{(2)}\) are gray blocks. Consequently, the entire collection of gates and measurements enclosed by the dashed red box can be replaced by a single composite gray operation in the first layer of the circuit in the middle, denoted as $E$ in the main text. The state after that first layer is in a tensor-product form. The forward lightcone from the largest non-decomposable subsystem is shown in the circuit on the right.}
\label{fig:measurelightcone}
\end{figure}
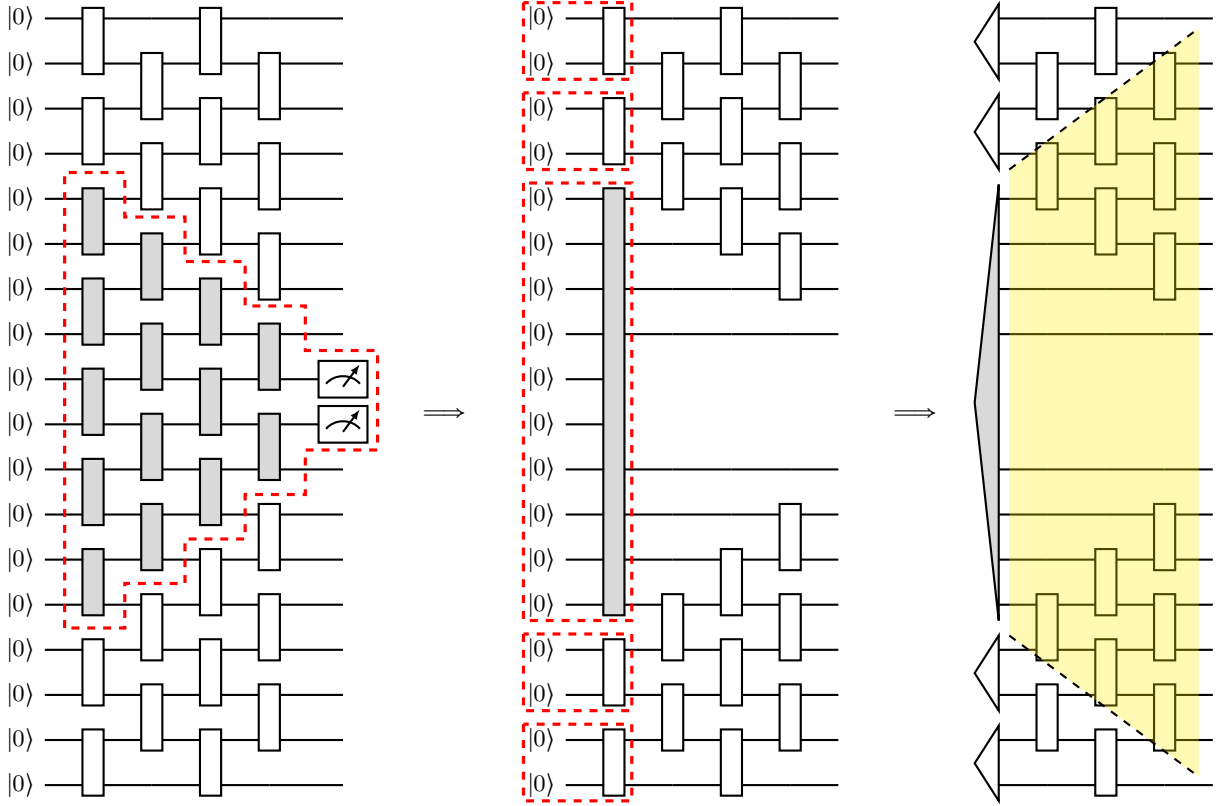

Now, we define
\begin{equation}\label{eq:definitionE}
E\equiv c(\id_n\otimes \bra{\mathbf{s}}) \left(\prod_{i=1}^L U_i^{(2)}\right).
\end{equation}
The final state can be expressed as $\ket{\psi}= \prod_{i=2}^L U_i^{(1)} (U_1^{(1)}\otimes E)\ket{0^m}$ where $U_1^{(1)} E$ can be regarded as a single layer of operations. Let $\ket{\tilde{\psi}}=U_1^{(1)} E\ket{0^m}$. It can be decomposed into a tensor product of states on smaller subsystems $\ket{\tilde{\psi}}=\bigotimes_{\bigcup A=[n]}\ket{\tilde{\psi}_A}$ according to the tensor-product structure of $U_1^{(1)}\otimes E$. The generalized stabilizer generators of $\ket{\tilde{\psi}}$ can be chosen by combining the set of generalized stabilizer generators of each $\ket{\tilde{\psi}_A}$, i.e., we can construct generator set $\Tilde{\mathcal{T}}_{\tilde{\psi}}=\sum_{A}\mathcal{T}^*_{\tilde{\psi}_A}$. Thus, $\left(\prod_{i=2}^L U_i^{(1)}\right)\Tilde{\mathcal{T}}_{\tilde{\psi}}\left(\prod_{i=2}^L U_i^{(1)}\right)^{\dagger}$ should be a set of generalized stabilizer generators of $\ket{\psi}$, which implies
\begin{equation}
\wt(\ket{\psi})\leq\wt(T_1)=\wt(U\tilde{T}U^{\dagger})\leq K^{L-1}\wt(\tilde{T})\leq K^{L-1}\max_j(\wt(\psi_j))\leq K^{L-1}\max(K,K^L(m-n))\leq K^{2L-1}(m-n+1).
\end{equation}
Here, $T_1$ is the one with maximal weight in $\left(\prod_{i=2}^L U_i^{(1)}\right)\Tilde{\mathcal{T}}_{\tilde{\psi}}\left(\prod_{i=2}^L U_i^{(1)}\right)^{\dagger}$, and $\tilde{T}$ is the one in $\Tilde{\mathcal{T}}_{\tilde{\psi}}$ corresponding to $T_1$. In the final step, we use the fact that the weight of a state is always no larger than its qubit number.

\subsection{Proof of Proposition~\ref{prop:adalowerboundstab}}\label{appendssc:proofadalowerboundstab}
In the first step of our proof, we use Eq.~\eqref{eq:adapt_state} to postpone the mid-circuit measurements to the end of the circuit. The final state $\ket{\psi}$ can be written as
\begin{equation}\label{eq:Clifadapt_state}
\ket{\psi} = c(\id_n \otimes \bra{\mathbf{s}}) C_{L}\cdots C_2C_1\ket{0^m}.
\end{equation}
Here, $\id_n$ denotes the identity operator acting on the first $n$ qubits; $\mathbf{s} = s_1s_2 \cdots s_{m-n}$ denotes the bit-string of the measured qubits; $C_L, \cdots, C_2, C_1$ are single layers of $K$-bounded fan-in Clifford gates that may depend on $\mathbf{s}$; and $c$ is a normalization factor. The measurement outcome is $+1$ for the state $\ket{0}$ and $-1$ for the state $\ket{1}$. Equation~\eqref{eq:Clifadapt_state} shows that the final state is obtained by measuring a stabilizer state, $C_{L}\cdots C_2C_1\ket{0^m}$. Notably, we can always ensure that all measurement outcomes are $+1$ by inserting appropriate Pauli $X$ gates before the measurements. These additional Pauli $X$ gates can be absorbed into the final layer $C_L$ without increasing the circuit depth. Therefore, it suffices to consider the case where all measurement outcomes are $+1$.


Let $\ket{\Psi}$ be the $m$-qubit state before measurement, and $\ket{\psi}$ be the final $n$-qubit state after measurement. We provide a necessary and sufficient condition for transforming $\ket{\Psi}$ into $\ket{\psi}$ as below.

\begin{lemma}\label{le:transform}
One can obtain $n$-qubit stabilizer state $\ket{\psi}$ by measuring $Z$ on the final $m-n$ qubits of $m$-qubit stabilizer state $\ket{\Psi}$ when getting all the outcomes $+1$ iff for all $S\in \mathbb{S}_{\psi}$, there exists a binary vector $\vec{z}$ of $m-n$ bits such that $S\otimes Z^{\vec{z}}\in \mathbb{S}_{\Psi}$, where $Z^{\vec{z}}$ denotes $\bigotimes_{i=1}^{m-n}Z_i^{z_i}$ for short.
\end{lemma}

\begin{proof}
First, suppose we can transform $\ket{\Psi}$ into $\ket{\psi}$ by measuring $Z$ on the final $m-n$ qubits of $m$-qubit state $\ket{\Psi}$ and get all outcomes $+1$. Since $\ket{\Psi}$ is a stabilizer state, it can be decomposed as a summation over the elements in its corresponding stabilizer group:
\begin{equation}
\ketbra{\Psi}=\frac{1}{2^m}\sum_{T\in \mathbb{S}_{\Psi}}T.
\end{equation}
Then, $\ket{\psi}$ should be given by
\begin{equation}\label{eq:transform}
\ketbra{\psi}=a\bra{0}_{m-n}\ketbra{\Psi}_m\ket{0}_{m-n}=\frac{a}{2^m}\sum_{T\in \mathbb{S}_{\Psi}}\bra{0}_{m-n}T\ket{0}_{m-n}=\frac{1}{2^n}\sum_{S\in \mathbb{S}_{\Psi}}S
\end{equation}
where $a$ is a normalization factor. For Pauli operators, we have $\bra{0}\id\ket{0}=\bra{0}Z\ket{0}=1$ and $\bra{0}X\ket{0}=\bra{0}Y\ket{0}=0$. Thus, $\bra{0}_{m-n}T\ket{0}_{m-n}\neq0$ iff the final $m-n$ Pauli operators of $T$ are $\id$ or $Z$. For those $T$ whose final $m-n$ Pauli operators are $\id$ or $Z$, $\bra{0}_{m-n}T\ket{0}_{m-n}$ just equals to the first $n$ Pauli operators of $T$, denoted by $T|_{n}$. Given that Pauli operators form a set of the basis of the operator space, Eq.~\eqref{eq:transform} indicates that for all $S\in \mathbb{S}_{\psi}$, there should be some $T\in \mathbb{S}_{\Psi}$ such that $T|_{n}=S$. The term $T|_{n}$ represents the former $n$ Pauli operators inside $T$. By combining this with the requirement that the final $m-n$ Pauli operators of $T$ should be $\id$ or $Z$, we finish the proof of the necessary condition.

Next, we come to the proof of the sufficient condition. Suppose for all $S\in \mathbb{S}_{\psi}$, there exists a binary vector $\vec{z}$ of $m-n$ bits such that $S\otimes Z^{\vec{z}}\in \mathbb{S}_{\Psi}$. Since $\ket{\Psi}$ is a stabilizer state, for any other binary vector $\vec{z'}$, $\tr{\ketbra{\Psi}S\otimes Z^{\vec{z'}}}$ can only be $0$ or $\pm 1$. If it is $-1$, $-S\otimes Z^{\vec{z'}}$ and $(S\otimes Z^{\vec{z}})\cdot (-S\otimes Z^{\vec{z'}})=-Z^{\vec{z}+\vec{z'}}$ should be stabilizers of $\ket{\Psi}$. In this case, there should be $\bra{0}_{m-n}\ketbra{\Psi}_m\ket{0}_{m-n}=0$ since $\bra{0}_{m-n}T\ket{0}_{m-n}+\bra{0}_{m-n}T\cdot(-Z^{\vec{z}+\vec{z'}})\ket{0}_{m-n}=0$ for any $T\in \mathbb{S}_{\Psi}$. Therefore, if we get all $+1$ outcomes from measurements, $\tr{\ketbra{\Psi}S\otimes Z^{\vec{z'}}}$ can only be $0$ or $+1$.
Denote the post-measured state with $\ketbra{\tilde{\psi}}\equiv a\bra{0}_{m-n}\ketbra{\Psi}_m\ket{0}_{m-n}$ where $a$ is a normalization factor and larger than 0. For any $S\in \mathbb{S}_{\psi}$,
\begin{equation}\label{eq:transform2}
\begin{aligned}
\tr{\ketbra{\tilde{\psi}}S}&=a\tr{\ketbra{\Psi} \left(S\otimes\bigotimes_{i=n+1}^{m}\frac{\id+Z_i}{2}\right)}=\frac{a}{2^{m-n}}\tr{\ketbra{\Psi}S\otimes Z^{\vec{z}}}+\frac{a}{2^{m-n}}\sum_{\vec{z'}\neq\vec{z}}\tr{\ketbra{\Psi}S\otimes Z^{\vec{z'}}}>0.
\end{aligned}
\end{equation}
Notice that $\ket{\Psi}$ is a stabilizer state, and our operations are Clifford. Therefore, the final state should also be a stabilizer state, whose stabilizer group contains $\mathbb{S}_{\psi}$. Since the final state is an $n$-qubit state, the size of its stabilizer group is $2^n = \abs{\mathbb{S}_{\psi}}$. Thus, the stabilizer group of the final state can only be $\mathbb{S}_{\psi}$. This implies that the final state is just $\ket{\psi}$, which finishes the proof.
\end{proof}

In the following proof, we will only use the necessary condition in Lemma~\ref{le:transform}, whereas the sufficient condition may be useful in designing adaptive state preparation protocols for specific states.

By limiting $T$ on the first $n$ qubits and only considering stabilizer generators, i.e., only discussing $T|_{n}$ where $T$ is a stabilizer generator of $\ket{\Psi}$, and applying Lemma~\ref{lemma:operatorweightgrowth}, we obtain the following necessary but not sufficient condition for an $m$-qubit Clifford adaptive circuit to prepare an $n$-qubit stabilizer state.

\begin{lemma}\label{lemma:Clifadap}
Suppose a depth-$L$ Clifford adaptive circuit with $K$-bounded fan-in gates and $m$-qubit initial state, $\ket{0^m}$, prepares $n$-qubit stabilizer state $\ket{\psi}$. There exists a set of $n$-qubit Pauli operators $P_1,P_2,\cdots,P_m$ such that
\begin{enumerate}[1.]
\item $\wt(P_j)\leq K^L$;\label{item:weight}
\item For any stabilizer $S_i^{*}$ of $\ket{\psi}$, one can find a subset $s_i\subseteq [m]$, s.t. $S_i^{*}=\prod_{j\in s_i}P_j$.\label{item:stabilizer}
\end{enumerate}
\end{lemma}

The Pauli operators $P_1,P_2,\cdots,P_m$ can be viewed as the stabilizer generators of $\ket{\Psi}$ restricted on the first $n$ qubits. The point~\ref{item:weight} results from a lightcone argument based on Lemma~\ref{lemma:operatorweightgrowth}. As $C_{L},\cdots, C_2, C_1$ are all layers of $K$-bounded fan-in Clifford gates, $\wt(P_j)\leq K^L$. The point~\ref{item:stabilizer} holds based on Lemma~\ref{le:transform}.

Below, we investigate the lower bound of $m$ based on the above lemma and prove Proposition~\ref{prop:adalowerboundstab}. Set $m' = \abs{\bigcup_{i=2}^n s_i}$. Without loss of generality, we assume $\bigcup_{i=2}^n s_i = [m']$. That is, the first $m'$ Pauli operators can generate $S_2^*, \cdots, S_n^*$.
We can prove that $m'\geq n-1$ by contradiction. If the reverse side holds, $S_2^*, \cdots, S_n^*$ can be generated by less than $n-1$ Pauli operators, which violates the independence of $S_i^*$. Then, we represent $P_1, \cdots, P_{n-1}$ with $S_2^*, \cdots, S_n^*$ and $P_n, \cdots, P_{m}$. This step is feasible thanks to the independence of $S_i^*$. We denote
\begin{equation}
P_i = \prod_{j=2}^n (S_j^*)^{a_{ij}} \prod_{j=n}^m (P_j)^{b_{ij}},
\end{equation}
where $a_{ij}, b_{ij} \in \{0, 1\}$ and $1\leq i\leq n-1$. Next, we analyze the decomposition of $S_1^*$ as below.
\begin{equation}
\begin{split}
S_1^* &= \prod_{j\in s_1, 1\leq j\leq n-1} P_j \prod_{j\in s_1, n\leq j\leq m} P_j\\
&= \prod_{j\in s_1, 1\leq j\leq n-1} \prod_{k=2}^n (S_k^*)^{a_{jk}} \prod_{k=n}^m (P_k)^{b_{jk}} \prod_{j\in s_1, n\leq j\leq m} P_j\\
&= \prod_{k=2}^n (S_k^*)^{c_k} \prod_{n\leq j\leq m} P_j^{d_j}.
\end{split}
\end{equation}
Here, $c_k = \bigoplus_{i\in s_1, 1\leq i\leq n-1} a_{ik}$ and $d_j = \mathbf{1}_{j\in s_1}\oplus \left(\bigoplus_{i\in s_1, 1\leq i\leq n-1}  b_{ij}\right)$ are both in $\{0, 1\}$; $\mathbf{1}_{j\in s_1}$ takes 1 if $j\in s_1$ and 0 otherwise. This step shows that we can decompose $S_1^*$ with all the other $S_k^*$ and the last $m-n+1$ $P_j$. We can further deduce that
\begin{equation}
S_1^* \prod_{k=2}^n (S_k^*)^{c_k} = \prod_{n\leq j\leq m} P_j^{d_j}.
\end{equation}
Denote $\tilde{S}_1^{*} = S_1^* \prod_{k=2}^n (S_k^*)^{c_k}$ where $\wt(\tilde{S}_1^{*})\geq \wt(S_1^{*})$ with the definition of $S_1^*$. Since the weight of each Pauli operator is no more than $K^L$, we can get
\begin{equation}
(m-n+1)K^L \geq \wt(\tilde{S}_1^{*})\geq \wt(S_1^{*}) = \wts(\ket{\psi}),
\end{equation}
which completes the proof.

\subsection{Proof of Lemma~\ref{lemma:weightandcorrrange}}\label{appendssc:lemmapf:weightandcorrrange}
We prove the result by contradiction. Recall that
\begin{equation}
\mathrm{Cor}^A_{P}(\ket{\psi})=\min_{i\neq j\in A}\max_{P_i,P_j\in \{I,X,Y,Z\}}\left|\mathrm{Cor}(P_i,P_j,\ket{\psi})\right|,
\end{equation}
and
\begin{equation}
\mathrm{CR}_P(\ket{\psi}) = \max_{\mathrm{Cor}^A_{P}(\ket{\psi}) > 0} \abs{A}.
\end{equation}
Without loss of generality, we denote $A^*$ to be the region saturating the above equality. That is, $\mathrm{CR}_P(\ket{\psi}) = \abs{A^*}.$

By contradiction, assuming that the condition
\begin{equation}
\wts(\ket{\psi}) \geq \frac{\abs{A^*}}{\sqrt{n}},
\end{equation}
is not satisfied, then we will have
\begin{equation}
n \frac{\wts(\ket{\psi})(\wts(\ket{\psi})-1)}{2} < \frac{\abs{A^*}(\abs{A^*}-1)}{2}.
\end{equation}
The above inequality implies $\exists i^*\neq j^* \in A^*$, such that no stabilizer generator of $\ket{\psi}$ will cover qubit $i^*$ and qubit $j^*$ simultaneously. Then, we can divide the stabilizer generators of $\ket{\psi}$ into three groups:
\begin{equation}
\mathcal{S}^{i^*} = \{S^{i^*}_k\}, \mathcal{S}^{j^*} = \{S^{j^*}_l\},
\mathcal{S}^*_{\psi}\backslash (\mathcal{S}^{i^*}\cup \mathcal{S}^{j^*}).
\end{equation}
Here, $S^{i^*}_k$ and $S^{j^*}_l$ are the stabilizer generators only covering qubits $i^*$ and $j^*$, respectively. Then, for any Pauli operator on qubit $i^*$, denoted as $P_{i^*}$, we have either
\begin{equation}
\bra{\psi} P_{i^*} \ket{\psi} = 0\Rightarrow \exists k, \{S_k^{i^*}, P_{i^*}\} = 0\Rightarrow \exists k, \{S_k^{i^*}, P_{i^*}P_{j^*}\}=0\Rightarrow \bra{\psi} P_{i^*}P_{j^*} \ket{\psi} = 0\Rightarrow \mathrm{Cor}(P_{i^*},P_{j^*},\ket{\psi}) = 0,
\end{equation}
or
\begin{equation}
\bra{\psi} P_{i^*} \ket{\psi} \neq 0\Rightarrow P_{i^*} = \prod_k c_k S^{i^*}_k
\Rightarrow \mathrm{Cor}(P_{i^*},P_{j^*},\ket{\psi}) = \prod_k c_k \bra{\psi} P_{j^*} \ket{\psi} - \prod_k c_k \bra{\psi} P_{j^*} \ket{\psi} = 0.
\end{equation}
Then, we will always get $\mathrm{Cor}(P_{i^*},P_{j^*},\ket{\psi}) = 0$, which contradicts the assumption that any qubit pair within $A^*$ has non-zero correlation. By contradiction, the condition $\wts(\ket{\psi}) \geq \frac{\abs{A^*}}{\sqrt{n}}$ holds, and the lemma is proved.

\subsection{Proof of Theorem~\ref{thm:adalowerboundgenericcorr}}\label{appendssc:prop2pf:permutation-invariant}

Recall the correlation function defined in Eq.~\eqref{eq:correlationfunction},
\begin{equation}
\mathrm{Cor}(O_1,O_2,\ket{\psi})=\bra{\psi}O_1O_2\ket{\psi}-\bra{\psi}O_1\ket{\psi}\bra{\psi}O_2\ket{\psi}.
\end{equation}
When $\ket{\psi}$ is separable under some decomposition $\ket{\psi}=\ket{\psi_1}_B\otimes\ket{\psi_2}_{\bar{B}}$ such that $\mathrm{supp}(O_1)\subseteq B$ and $\mathrm{supp}(O_1)\subseteq \bar{B}$, one can find $\bra{\psi}O_1O_2\ket{\psi}=\bra{\psi_1}_BO_1\ket{\psi_1}_B\bra{\psi_2}_{\bar{B}}O_2\ket{\psi_2}_{\bar{B}}=\bra{\psi}O_1\ket{\psi}\bra{\psi}O_2\ket{\psi}$ and thus $\mathrm{Cor}(O_1,O_2,\ket{\psi})=0$.

We will prove Theorem~\ref{thm:adalowerboundgenericcorr} by contradiction. Suppose that one can find a depth-$L$ adaptive circuit with $K$-bounded fan-in gates and $m$-qubit initial state to prepare $n$-qubit state $\ket{\psi}$, whose parameters violate Eq.~\eqref{eq:correlationbound}. Recall the proof of Theorem~\ref{theo:adalowerboundgeneric}, the state prepared by such an adaptive circuit can be written as $\ket{\psi}=\tilde{U}\ket{\tilde{\psi}}$, where $\tilde{U}=\prod_{i=2}^L U_i^{(1)}$ and $\ket{\tilde{\psi}}$ is a product state over some decomposition according to the backward lightcone from measured qubits:
\begin{equation}
\ket{\tilde{\psi}}=\bigotimes_{\bigcup_j B_j=[n]}\ket{\tilde{\psi}^{B_j}}.
\end{equation}
Here, only one of those $B_j$ corresponds to the backward lightcone from measured qubits, denoted as $B_1=\mathrm{supp}(E)$, whose cardinality is bounded by $|B_1|=|\mathrm{supp}(E)|\leq K^L(m-n)$. $E$ is defined in Eq.~\eqref{eq:definitionE}. For any other $B_j$ with $j\neq 1$, the cardinality is bounded by $|B_j|\leq K$.

Note that
\begin{equation}
\mathrm{Cor}(O_1,O_2,\ket{\psi})=\bra{\tilde{\psi}}\tilde{U}^{\dagger}O_1\tilde{U}\tilde{U}^{\dagger}O_2\tilde{U}\ket{\tilde{\psi}}-\bra{\tilde{\psi}}\tilde{U}^{\dagger}O_1\tilde{U}\ket{\tilde{\psi}}\bra{\tilde{\psi}}\tilde{U}^{\dagger}O_2\tilde{U}\ket{\tilde{\psi}}=\mathrm{Cor}(\tilde{U}^{\dagger}O_1\tilde{U},\tilde{U}^{\dagger}O_2\tilde{U},\ket{\tilde{\psi}})
\end{equation}
for any operators $O_1$ and $O_2$. We use $A^*$ to denote the set $A$ with the largest cardinality which satisfies $\mathrm{Cor}^A_w(\ket{\psi})>0$. That is, $A^*=\arg\max_{\mathrm{Cor}^A_w(\ket{\psi})>0}|A|$. For any operator $O_1$ whose support is covered by $A^*$, i.e., $\mathrm{supp}(O_1)\subseteq A^*$, we first consider the support of $\tilde{O}_1\equiv\tilde{U}^{\dagger}O_1\tilde{U}$. The backward lightcone of $\mathrm{supp}(O_1)$ gives an upper bound on the size of it, showing $|\mathrm{supp}(\tilde{O}_1)|\leq K^{L-1}\mathrm{supp}(O_1)\leq w\cdot K^{L-1}$. Rigorous analysis of the backward lightcone is similar to that in the proof of Theorem~\ref{theo:adalowerboundgeneric}, and a diagram is shown in Fig.~\ref{fig:correlationlightcone}. Let $\tilde{B}=\bigcup_{B_j\cap\mathrm{supp}(\tilde{O}_1)\neq\emptyset}B_j$ whose cardinality can be bounded by $|B|\leq |B_1|+(|\mathrm{supp}(\tilde{O}_i)|-1)K\leq (m-n+w)K^L$. $\ket{\tilde{\psi}}$ can be decomposed as $\ket{\tilde{\psi}}=\ket{\tilde{\psi}_1}_{\tilde{B}}\otimes\ket{\tilde{\psi}_2}_{[n]\setminus\tilde{B}}$. Therefore, as long as we can find $O_2\subseteq A^*$ with $\wt(O_2)\leq w$ such that the support of $\tilde{O}_2\equiv\tilde{U}^{\dagger}O_2\tilde{U}$ is a subset of $[n]\setminus\tilde{B}$, the correlation function $\mathrm{Cor}(O_1,O_2,\ket{\psi})=\mathrm{Cor}(\tilde{U}^{\dagger}O_1\tilde{U},\tilde{U}^{\dagger}O_2\tilde{U},\ket{\tilde{\psi}})$ should be $0$.

To find such $O_2$, consider the forward lightcone from qubits in $\tilde{B}$. The width of the lightcone at the end of the circuit is bounded by $K^{L-1}|\tilde{B}|\leq (m-n+w)K^{2L-1}$, as shown in Fig.~\ref{fig:correlationlightcone}. Thus, there should be at least $|A^*|-K^{L-1}|\tilde{B}|\geq \mathrm{CR}_w(\ket{\psi})-(m-n+w)K^{2L-1}\geq w$ qubits in $A^*$ outside the lightcone at the end of the circuit. Thus, the backward lightcone of those $w$ qubits should be disjoint from the forward lightcone of qubits in $|\tilde{B}|$, which implies that any $O_2$ on those $w$ qubits should give $\mathrm{supp}(\tilde{O}_2)\cap \tilde{B}=\emptyset$. The correlation function of $O_1$ and such $O_2$ should be $0$, where $O_1$ and $O_2$ can be arbitrary operators acting on qubit $i$ and $j$ respectively, which contradicts the condition in Theorem~\ref{thm:adalowerboundgenericcorr}. Therefore, $(m-n+w)K^{2L-1}\geq \mathrm{CR}_w(\ket{\psi})+w-1$ must be satisfied.

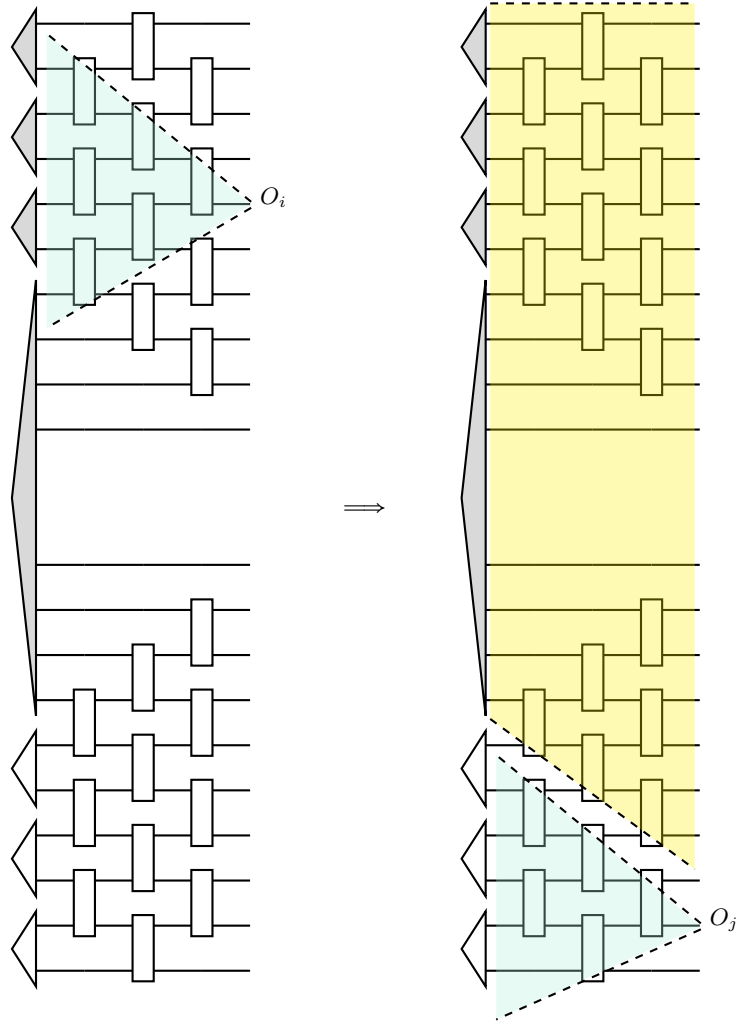
\begin{figure}[!htbp]
 \centering
  \begin{quantikz}[row sep={0.6cm,between origins}, column sep=0.5cm]
  & \qw & \gate[wires=2]{} & \qw & \qw\\
  & \gate[wires=2]{} &  & \gate[wires=2]{} & \qw  \\
  &  & \gate[wires=2]{} &  & \qw  \\
  & \gate[wires=2]{} &  & \gate[wires=2]{} & \qw  \\
  &  & \gate[wires=2]{} &  & \qw  \\
  & \gate[wires=2]{} &  & \gate[wires=2]{} & \qw  \\
  &  & \gate[wires=2]{} &  & \qw  \\
  & \qw  &  & \gate[wires=2]{} & \qw  \\
  & \qw  & \qw  &  & \qw  \\
  & \qw  & \qw  & \qw  & \qw  \\
  \\
  \\
  & \qw  & \qw  & \qw  & \qw  \\
  & \qw  & \qw  & \gate[wires=2]{} & \qw  \\
  & \qw  & \gate[wires=2]{} &  & \qw  \\
  & \gate[wires=2]{} &  & \gate[wires=2]{} & \qw  \\
  &  & \gate[wires=2]{} &  & \qw  \\
  & \gate[wires=2]{} &  & \gate[wires=2]{} & \qw  \\
  &  & \gate[wires=2]{} &  & \qw  \\
  & \gate[wires=2]{} &  & \gate[wires=2]{} & \qw  \\
  &  & \gate[wires=2]{} &  & \qw  \\
  & \qw &  & \qw & \qw  \\
  \end{quantikz}
  \begin{tikzpicture}[overlay,x=4mm,y=4mm]
  \draw[black, thick] (-7.15,16.8) -- ++(0,-2.5) -- ++(-0.8,1.25) -- cycle;
  \draw[black, thick] (-7.15,13.8) -- ++(0,-2.5) -- ++(-0.8,1.25) -- cycle;
  \draw[black, thick] (-7.15,10.8) -- ++(0,-2.5) -- ++(-0.8,1.25) -- cycle;
  \draw[black, thick] (-7.15,7.8) -- ++(0,-14.5) -- ++(-0.8,7.25) -- cycle;
  \draw[black, thick] (-7.15,-7.2) -- ++(0,-2.5) -- ++(-0.8,1.25) -- cycle;
  \draw[black, thick] (-7.15,-10.2) -- ++(0,-2.5) -- ++(-0.8,1.25) -- cycle;
  \draw[black, thick] (-7.15,-13.2) -- ++(0,-2.5) -- ++(-0.8,1.25) -- cycle;
  \draw[black, dashed, thick] (0,10.4) -- ++(-6.8,5.6);
  \draw[black, dashed, thick] (0,10.2) -- ++(-6.8,-4);
  \node[fit={(0.3,10.3)}, text=black, font=\small\bfseries] {$O_i$};
  \begin{scope}[on background layer]
  \fill[gray, fill opacity=0.3, draw=none] (-7.15,16.8) -- ++(0,-2.5) -- ++(-0.8,1.25) -- cycle;
  \fill[gray, fill opacity=0.3, draw=none] (-7.15,13.8) -- ++(0,-2.5) -- ++(-0.8,1.25) -- cycle;
  \fill[gray, fill opacity=0.3, draw=none] (-7.15,10.8) -- ++(0,-2.5) -- ++(-0.8,1.25) -- cycle;
  \fill[gray, fill opacity=0.3, draw=none] (-7.15,7.8) -- ++(0,-14.5) -- ++(-0.8,7.25) -- cycle;
  \fill[mintcyan, fill opacity=0.3, draw=none] (-6.8,16) -- (-6.8,6.2) -- (0,10.3) -- cycle;
  \end{scope}
  \end{tikzpicture}
  \hspace{10mm}
  $\Longrightarrow$
  \hspace{10mm}
  \begin{quantikz}[row sep={0.6cm,between origins}, column sep=0.5cm]
  & \qw & \gate[wires=2]{} & \qw & \qw\\
  & \gate[wires=2]{} &  & \gate[wires=2]{} & \qw  \\
  &  & \gate[wires=2]{} &  & \qw  \\
  & \gate[wires=2]{} &  & \gate[wires=2]{} & \qw  \\
  &  & \gate[wires=2]{} &  & \qw  \\
  & \gate[wires=2]{} &  & \gate[wires=2]{} & \qw  \\
  &  & \gate[wires=2]{} &  & \qw  \\
  & \qw  &  & \gate[wires=2]{} & \qw  \\
  & \qw  & \qw  &  & \qw  \\
  & \qw  & \qw  & \qw  & \qw  \\
  \\
  \\
  & \qw  & \qw  & \qw  & \qw  \\
  & \qw  & \qw  & \gate[wires=2]{} & \qw  \\
  & \qw  & \gate[wires=2]{} &  & \qw  \\
  & \gate[wires=2]{} &  & \gate[wires=2]{} & \qw  \\
  &  & \gate[wires=2]{} &  & \qw  \\
  & \gate[wires=2]{} &  & \gate[wires=2]{} & \qw  \\
  &  & \gate[wires=2]{} &  & \qw  \\
  & \gate[wires=2]{} &  & \gate[wires=2]{} & \qw  \\
  &  & \gate[wires=2]{} &  & \qw  \\
  & \qw &  & \qw & \qw  \\
  \end{quantikz}
  \begin{tikzpicture}[overlay,x=4mm,y=4mm]
  \draw[black, thick] (-7.15,16.8) -- ++(0,-2.5) -- ++(-0.8,1.25) -- cycle;
  \draw[black, thick] (-7.15,13.8) -- ++(0,-2.5) -- ++(-0.8,1.25) -- cycle;
  \draw[black, thick] (-7.15,10.8) -- ++(0,-2.5) -- ++(-0.8,1.25) -- cycle;
  \draw[black, thick] (-7.15,7.8) -- ++(0,-14.5) -- ++(-0.8,7.25) -- cycle;
  \draw[black, thick] (-7.15,-7.2) -- ++(0,-2.5) -- ++(-0.8,1.25) -- cycle;
  \draw[black, thick] (-7.15,-10.2) -- ++(0,-2.5) -- ++(-0.8,1.25) -- cycle;
  \draw[black, thick] (-7.15,-13.2) -- ++(0,-2.5) -- ++(-0.8,1.25) -- cycle;
  \draw[black, dashed, thick] (-7,17) -- ++(6.8,0);
  \draw[black, dashed, thick] (-7,-6.8) -- ++(6.8,-5);
  \draw[black, dashed, thick] (0,-13.6) -- ++(-6.8,5.6);
  \draw[black, dashed, thick] (0,-13.8) -- ++(-6.8,-3);
  \node[fit={(0.3,-13.7)}, text=black, font=\small\bfseries] {$O_j$};
  \begin{scope}[on background layer]
  \fill[gray, fill opacity=0.3, draw=none] (-7.15,16.8) -- ++(0,-2.5) -- ++(-0.8,1.25) -- cycle;
  \fill[gray, fill opacity=0.3, draw=none] (-7.15,13.8) -- ++(0,-2.5) -- ++(-0.8,1.25) -- cycle;
  \fill[gray, fill opacity=0.3, draw=none] (-7.15,10.8) -- ++(0,-2.5) -- ++(-0.8,1.25) -- cycle;
  \fill[gray, fill opacity=0.3, draw=none] (-7.15,7.8) -- ++(0,-14.5) -- ++(-0.8,7.25) -- cycle;
  \fill[mintcyan, fill opacity=0.3, draw=none] (-6.8,-8) -- (-6.8,-16.8) -- (0,-13.7) -- cycle;
  \fill[yellow, fill opacity=0.3, draw=none] (-7,17) -- (-7,-6.8) -- (-0.2,-11.8) -- (-0.2,17) -- cycle;
  \end{scope}
  \end{tikzpicture}
\caption{Illustration of the backward lightcones of qubits $i$ and $j$, and the forward lightcone of region $B$. Our goal is to determine $O_j$ given $O_i$. Any adaptive circuit can be represented by the circuit on the left, whose initial state is separable based on the backward lightcone from measurements, as in Fig.~\ref{fig:measurelightcone}. The largest set of states covered by the backward lightcone from $O_i$, denoted as $B$ in the proof, is highlighted in gray. The circuit on the right demonstrates that if one can identify an operator $O_j$ supported on qubit $j$, which lies outside the forward lightcone of the gray qubits, then it follows that the gray qubits must also lie outside the backward lightcone of qubit $j$.
}
\label{fig:correlationlightcone}
\end{figure}

\subsection{Discussion of cases with geometric constraints}\label{appendsc:geo}
We start with the formal definition of the operator diffusion function.
\begin{definition}\label{defi:GKD}
The operator diffusion function $G_{K,D}(A)$ is a function of support depending on the support, $A$, with parameters $K$, $D$. It should be the smallest support that satisfies: for any operator $O$ whose support is a subset of $A$ and any $D$-depth $K$-bounded fan-in unitary gate $U$, we have $\mathrm{supp}(UOU^{\dagger})\subseteq G_{K,D}(A)$. Equivalently,
\begin{equation}
G_{K,D}(A)=\bigcup_{\substack{\mathrm{supp}(O)\subseteq A\\ U\in\mathcal{U}_{K,D}}}\mathrm{supp}(UOU^{\dagger})
\end{equation}
where $\mathcal{U}_{K,D}$ denotes the set of depth-$D$ unitary circuits with $K$-bounded fan-in, possibly subject to geometric constraints.
\end{definition}
One can find that $G_{K,D_2}(G_{K,D_1}(A))=G_{K,D_1+D_2}(A)$ by definition. For any operator $O$ whose support is a subset of $A_1+A_2$, where $A_1$ and $A_2$ are two disjoint supports, one can decompose $O$ into a product form, $O=\sum_i O_i^{(1)}\cdot O_i^{(2)}$ Here, the support of $O_i^{(1)}$ ($O_i^{(2)}$) is a subset of $A_1$ ($A_2$). Thus, for any $U\in\mathcal{U}_{K,D}$,
\begin{equation}\label{eq:supportadd}
\mathrm{supp}(UOU^{\dagger})\subseteq\bigcup_i\mathrm{supp}(UO_i^{(1)}O_i^{(2)}U^{\dagger})\subseteq\left(\bigcup_i\mathrm{supp}(UO_i^{(1)}U^{\dagger})\right)\cup\left(\bigcup_i\mathrm{supp}(UO_i^{(2)}U^{\dagger})\right)\subseteq G_{K,D}(A_1)\cup G_{K,D}(A_2).
\end{equation}
It is based on the fact that $\mathrm{supp}(A+B)\subseteq\mathrm{supp}(A)\cup\mathrm{supp}(B)$ and $\mathrm{supp}(AB)\subseteq\mathrm{supp}(A)\cup\mathrm{supp}(B)$. Since Eq.~\eqref{eq:supportadd} is true for any operator $O$, one can conclude that
\begin{equation}\label{eq:GKDunion}
G_{K,D}(A_1+A_2)\subseteq G_{K,D}(A_1)\cup G_{K,D}(A_2).
\end{equation}
One can optimize the choice of support $A$ to get $G_{K,D}(A)$ with maximal cardinality when the number of qubits in the support is fixed. Define
\begin{equation}
g_{K,D}(a)=\max_{|A|=a}|G_{K,D}(A)|.
\end{equation}
When $A$ only includes a single qubit, we omit the parameter $a$ and just use $g_{K,D}\equiv g_{K,D}(1)$. With the property of Eq.~\eqref{eq:GKDunion}, one can find that $g_{K,D}(a)\leq a\cdot g_{K,D}$.

Based on the former definition, one can generalize the proof of Proposition~\ref{prop:nonadaptivegeneral} by replacing Eq.~\eqref{eq:prooftheo1} with
\begin{equation}
\wt(\ket{\psi})=\wt(T^*_1)\leq\wt(T_1)=\wt(UZ_{i_1}U^{\dagger})\leq g_{K,D}.
\end{equation}
Moreover, Lemma~\ref{le:marginal} can also be generalized with geometric constraints by requiring the cardinality of the subset $|s|=g_{K,D}$. The cases for the one-dimensional chain and the square lattice were considered in Ref.~\cite{yu2023learning}, where $g_{K,D}$ is bounded by $O(KD)$ and $O(K^2D^2)$, respectively. Then, the result of Proposition~\ref{prop:nonadaptivegeneralClifford} changes to be
\begin{equation}
g_{K,D}\geq\wts(\ket{\psi}).
\end{equation}

For Theorem~\ref{theo:adalowerboundgeneric}, one can go through the proof and can find that we need to bound the weight of $E=c(\id_n\otimes \bra{\mathbf{s}}) \left(\prod_{i=1}^L U_i^{(2)}\right)$. We can solve this problem recursively by considering the support of $O_{j+1}^{(2)}U_j^{(2)}$ in each step, where $O_{j+1}^{(2)}=c(\id_n\otimes \bra{\mathbf{s}})\left(\prod_{i=j+1}^L U_i\right)$. Suppose the support of $O_{j+1}^{(2)}$ is $A_{j+1}$, the support of $O_j^{(2)}$ is
\begin{equation}
A_j=\mathrm{supp}(O_{j+1}^{(2)}U_j^{(2)})\subseteq\mathrm{supp}(O_{j+1}^{(2)})\cup\mathrm{supp}(U_j^{(2)})=A_{j+1}\cup Q_j^{(2)}.
\end{equation}
To bound the set $A_{j+1}\cup Q_j^{(2)}$, we construct an operator whose support is it. To do this, we first find an operator $O'_{j+1}$ as a tensor product of Pauli $X$ whose support is $A_{j+1}\cup Q_{j+1}^{(2)}$. Then, we construct a new depth-$1$ $K$-bounded fan-in unitary circuit $U'_j$ by replacing each tensor product term of $U_j^{(2)}$ with generalized $CZ$ gates. Specifically, suppose $U_j^{(2)}$ can be decomposed as $U_j^{(2)}=\bigotimes_{\abs{R}\leq K} U_{R}$, $U'_j$ is defined as $U'_j=\bigotimes_{\abs{R}\leq K} CZ_R$, where $CZ_R\equiv\id_R-2\ketbra{1}_R$ is the multi-qubit controlled-$Z$ gate on the qubits in $R$. Since $CZ_RX_iCZ_R^{\dagger}=X_iCZ_{R\backslash\{i\}}$ as long as $i\in R$, and $R\cap Q_{j+1}^{(2)}\neq\emptyset$ for all $R$ by definition, the whole support of $U'_j$ should be added to the support of $O'_{j+1}$ after the evolution $U'_jO'_{j+1}{U'}_j^{\dagger}$, i.e., $\mathrm{supp}(U'_jO'_{j+1}{U'}_j^{\dagger})=\mathrm{supp}(O'_{j+1})\cup\mathrm{supp}(U'_j)$. Then, Definition~\ref{defi:GKD} shows that
\begin{equation}
\mathrm{supp}(U'_jO'_{j+1}{U'}_j^{\dagger})=\mathrm{supp}(O'_{j+1})\cup\mathrm{supp}(U'_j)=A_{j+1}\cup Q_{j+1}^{(2)}\cup Q_j^{(2)}\subseteq G_{K,1}(A_{j+1}).
\end{equation}
Therefore, by combining the former two equations, one can derive $A_j\subseteq G_{K,1}(A_{j+1})$. If we start with $O_{L+1}=\id_n\otimes\bra{\mathbf{s}}$ whose support is denoted as $A_L$ and repeat the procedure, we will finally get $\mathrm{supp}(E)=A_0\subseteq G_{K,L}(A_L)$. Therefore,
\begin{equation}
\begin{split}
\wt(\ket{\psi})&\leq \max\left(|G_{K,L}(\{i\})|,|G_{K,L-1}(G_{K,L}(A_L))|\right)\\
&= \max( g_{K,L}, |G_{K,2L-1}(A_L)| )\\
&\leq \max( g_{K,L}, (m-n)g_{K,2L-1} )\\
&\leq (m-n+1)g_{K,2L-1}.
\end{split}
\end{equation}
Here, $\{i\}$ represents the support of qubit $i$.

For Proposition~\ref{prop:adalowerboundstab}, we only need to replace $K^L$ with $g_{K,L}$, as it follows the same weight growth used in Proposition~\ref{prop:nonadaptivegeneral}. Specifically, one only needs to modify the first condition in Lemma~\ref{lemma:Clifadap} from $\wt(P_j)\leq K^L$ to $\wt(P_j)\leq g_{K,L}$. The final result shows
\begin{equation}
(m-n+1)g_{K,L}\geq\wts(\ket{\psi}).
\end{equation}

For Theorem~\ref{thm:adalowerboundgenericcorr}, we replace $K^L$ with $g_{K,L}$ throughout the proof wherever forward or backward lightcones are used. When analyzing the forward lightcone from region $\tilde{B}$, we consider two cases of $\tilde{B}$, which consists of all product-state components overlapping with the backward lightcone of $O_1$.

If $\tilde{B}$ does not contain the qubits in the backward lightcone from measurements, i.e., $\mathrm{supp}(E)$ defined in Eq.~\eqref{eq:definitionE}, we have $B\subseteq G_{K,L}(\mathrm{supp}(O_1))$, since any separable part within $\tilde{B}$ is caused by a single $K$-qubit gate. In this case, the forward lightcone is bounded as $|G_{K,L-1}(\tilde{B})|\leq |G_{K,L-1}(G_{K,L}(\mathrm{supp}(O_1)))|\leq w\cdot g_{K,2L-1}$.

If $\tilde{B}$ does include $\mathrm{supp}(E)$, we decompose $\tilde{B}$ into two parts: $\mathrm{supp}(E)$ and $\tilde{B} \backslash \mathrm{supp}(E)$. As established earlier, $\mathrm{supp}(E) \subseteq G_{K,L}(A_L)$, where $A_L$ denotes the set of measured qubits. The remaining part, $\tilde{B} \backslash \mathrm{supp}(E)$, includes only product-state components arising from local $K$-qubit gates and hence satisfies $\tilde{B} \backslash \mathrm{supp}(E) \subseteq G_{K,L}({\mathrm{supp}(O_1)})$. Thus, the size of the forward lightcone at the end of the circuit is bounded by
\begin{equation}
\begin{aligned}
|G_{K,L-1}(\tilde{B})|&\leq |G_{K,L-1}(\mathrm{supp}(E))|+|G_{K,L-1}(\tilde{B}\backslash\mathrm{supp}(E))|\\
&\leq |G_{K,L-1}(G_{K,L}(A_L))|+|G_{K,L-1}(G_{K,L}(\mathrm{supp}(O_1)))|\\
&\leq (m-n)g_{K,2L-1}+w\cdot g_{K,2L-1}.
\end{aligned}
\end{equation}

Note that we must have $|G_{K,L-1}(\tilde{B})| \geq \mathrm{CR}_w(\ket{\psi})-(w-1)$; otherwise, there would exist at least two qubits such that the correlation function vanishes for all operators supported on them. Combining the arguments above, we obtain
\begin{equation}
(m-n+w)g_{K,2L-1}+w-1\geq \mathrm{CR}_w(\ket{\psi}).
\end{equation}

For Proposition~\ref{prop:permutation-invariant}, it is a direct corollary of Theorem~\ref{thm:adalowerboundgenericcorr}.

\section{Proof of results related to anti-shallowness}\label{appendsc:antishallow}
\subsection{Proof of Theorem~\ref{thm:adaptpower}}\label{appendssc:thmpf:adaptpower}
We first present a strong and more detailed version of Theorem~\ref{thm:adaptpower}.

\begin{theorem}\label{thm:adaptpowerFull}
Suppose $m$ and $n$ are two positive integers with $m - n = \omega(1)$. The distance between any $n$-qubit state prepared by an $m$-qubit shallow adaptive circuit and an $n$-qubit $D$-depth non-adaptive circuit is $O(m-n)$. Meanwhile, there exists an $n$-qubit state prepared by an $m$-qubit shallow adaptive circuit such that its distance to any state prepared by an $n$-qubit $D$-depth non-adaptive circuit is $\Omega((m-n)/K^{5D})$. More precisely,
\begin{equation}
-\log F_{D,n}^{L,m} = O(\min((m-n)K^L, n)) \Rightarrow \forall L\in O(1), -\log F_{D,n}^{L,m} = O(\min(m-n, n)),
\end{equation}
and $\exists L\in O(1)$,
\begin{equation}
-\log F_{D,n}^{L,m} = \Omega(\min(m-n, n)/K^{5D}).
\end{equation}
If $m-n = O(1)$, then all shallow-adaptive circuit states can be generated by shallow circuits. Mathematically, $\forall L\in O(1), \exists D\in O(1)$,
\begin{equation}
-\log F_{D,n}^{L,m} = 0.
\end{equation}
\end{theorem}

The main difference between Theorem~\ref{thm:adaptpowerFull} and Theorem~\ref{thm:adaptpower} is that the former does not require the circuit depth $D$ to be constant. In the following, we present the full proof.

\subsection{Upper bound}
We first prove that $-\log F_{D,n}^{L,m} = O(\min(m - n, n))$ when the number of ancillary qubits is non-constant, and that $-\log F_{D,n}^{L,m} = 0$ when the number of ancillary qubits is constant.
As elaborated in the main text, if a state $\ket{\psi}$ can be written as $\ket{\phi}\otimes \ket{0^{n-t}}$ where $\ket{\phi}$ is a $t$-qubit state, the value of $F_{D,n}(\ket{\psi})$ is at least $2^{-t}$ resulting from the following equation.
\begin{equation}
\begin{split}
F_{D,n}(\ket{\psi}) &= \max_{V:\mathrm{depth}(V)=D} \abs{\bra{\psi}V\ket{0^n}}^2\\
&\geq \max_{V:\mathrm{depth}(V)=D} \abs{\bra{\phi}V\ket{0^t}}^2\\
&\geq \max_{\mathbf{s}\in \{0,1\}^t} \abs{\bra{\phi}\ket{\mathbf{s}}}^2\\
&\geq 2^{-t}\sum_{\mathbf{s}\in \{0,1\}^t} \abs{\bra{\phi}\ket{\mathbf{s}}}^2\\
&= 2^{-t}.
\end{split}
\end{equation}
Thus, $-\log F_{D,n}^{L,m} = O(n)$ holds automatically. The remaining thing is to prove $-\log F_{D,n}^{L,m} = O(m-n)$ when $m-n<n$. We will show that any state $\ket{\psi}$ prepared by shallow adaptive circuits with $m-n$ ancillary qubits satisfies $-\log F_{D,n}(\ket{\psi}) = O(m-n)$. We investigate the state $\ket{\psi}$ with the following expression.
\begin{equation}\label{eq:append_adapt_state}
\ket{\psi} = c(\id_n\otimes \bra{\mathbf{s}}) U_{L}\cdots U_2U_1\ket{0^m},
\end{equation}
where $\id_n$ denotes the identity operator acting on the first $n$ qubits; $\mathbf{s}=s_1s_2\cdots s_{m-n}$ denotes all measurement outcomes; $U_{L},\cdots, U_2, U_1$ are all single layers of $K$-bounded fan-in quantum gates, which depend on the measurement outcomes $s_1, s_2, \cdots, s_{m-n}$; $c$ is a normalization factor. For each layer of unitary $U_i$, we decompose it into two parts $U_i = U_i^{(1)}\otimes U_i^{(2)}$ where $\mathrm{supp}(U_i^{(2)})$ contains last $m-n$ qubits and $\mathrm{supp}(U_i^{(1)})$ does not. The decomposition should make $\wt(U_i^{(2)})$ as small as possible. Thus, the fan-in of $U_i^{(2)}$ is upper bounded by $K(m-n)$ from the definition. With this decomposition, we rewrite the expression of $\ket{\psi}$ below.
\begin{equation}
\begin{split}
\ket{\psi} &= c(\id_n\otimes \bra{\mathbf{s}}) \prod_{i=1}^{L} (U_i^{(1)}\otimes U_i^{(2)}) \ket{0^m}\\
&= c(\prod_{i=1}^{L} U_i^{(1)}) (\id_n\otimes \bra{\mathbf{s}}) \prod_{i=1}^{L} \left( (\prod_{j=1}^{i-1} U_i^{(1)})^{\dagger} U_i^{(2)} (\prod_{j=1}^{i-1} U_i^{(1)}) \right) \ket{0^m}.
\end{split}
\end{equation}
Note that $\prod_{i=1}^{L} U_i^{(1)}$ is a depth-$L$ shallow circuit only containing unitary operations. The fan-in of unitary operation $(\prod_{j=1}^{i-1} U_i^{(1)})^{\dagger} U_i^{(2)} (\prod_{j=1}^{i-1} U_i^{(1)})$ will be upper bounded by $K^{i-1}K(m-n) = K^i(m-n)$ based on the Lemma~\ref{lemma:operatorweightgrowth}.

Lemma~\ref{lemma:operatorweightgrowth} indicates that the weight of $(\prod_{j=1}^{i-1} U_i^{(1)})^{\dagger} U_i^{(2)} (\prod_{j=1}^{i-1} U_i^{(1)})$ is still proportional to $m-n$ with a constant factor $K^i$ after shallow circuits. As a result,
\begin{equation}
\wt(\prod_{i=1}^{L} \left( (\prod_{j=1}^{i-1} U_i^{(1)})^{\dagger} U_i^{(2)} (\prod_{j=1}^{i-1} U_i^{(1)}) \right))\leq m-n+\sum_{i=1}^{L}(K^i-1)(m-n) = (\frac{K^{L+1}-1}{K-1}-L)(m-n).
\end{equation}
Hence, the state after measurement and before $\prod_{i=1}^{L} U_i^{(1)}$,
\begin{equation}
c(\id_n\otimes \bra{\mathbf{s}}) \prod_{i=1}^{L} \left( (\prod_{j=1}^{i-1} U_i^{(1)})^{\dagger} U_i^{(2)} (\prod_{j=1}^{i-1} U_i^{(1)}) \right) \ket{0^m},
\end{equation}
must have a decomposition of a $a(m-n)$-qubit state, $\ket{\phi}$, and a product state $\ket{0^{n-a(m-n)}}$ where $a = \frac{K^{L+1}-1}{K-1}-L-1$ is a constant. Thus,
\begin{equation}
F_{D,n}(\ket{\psi}) \geq F_{D-L, a(m-n)} (\ket{\phi}) \geq 2^{-a(m-n)},
\end{equation}
which is equivalent to $-\log F_{D,n}(\ket{\psi}) = O((m-n)K^L)$. Hence, $-\log F_{D,n}^{L,m}(\ket{\psi}) = O(m-n)$ for $L = O(1)$.

Notice that the final state $\ket{\psi}$ is connected to $\ket{\phi}\ket{0^{n - a(m - n)}}$ via a depth-$L$ or shallow circuit. If $m-n$ is a constant, then the state $\ket{\phi}$ can also be prepared by a shallow circuit. Consequently, $\ket{\psi}$ itself can be generated by a shallow circuit. In this case, there exists a constant $D$ such that $F_{D,n}(\ket{\psi}) = 1$, implying that $-\log F_{D,n}(\ket{\psi}) = 0$. We summarize that for shallow adaptive circuits or $L = O(1)$,
\begin{equation}\label{eq:Fupperbound}
-\log F_{D,n}^{L,m} =
\begin{cases}
0 & \text{ if } m-n=O(1),\\
O(\min(m-n, n)) & \text{ if } m-n=\omega(1).
\end{cases}
\end{equation}

\subsection{Lower bound}
Next, we prove that $\exists L\in O(1), -\log F_{D,n}^{L,m} = \Omega(\min(m-n, n)/K^{5D})$. We consider the case of $m-n = \omega(1)$ to let $-\log F_{D,n}^{L,m} > 0$.

Recall that in Section~\ref{ssc:QLDPCprepare} in the main text, we introduce a deterministic scheme to prepare a QLDPC code state. Specifically, given a $[[t, Rt, \gamma t]]$ QLDPC code with sparsity $s$, we can prepare a logical state with $(1-R)t$ ancillary qubits and depth $O(s^2)$. Thanks to the existence of good QLDPC codes~\cite{Panteleev2022goodQLDPC}, for any constant $R\in (0, 1)$, there exists a $[[t, Rt, \gamma t]]$ QLDPC code with $\gamma$ also a constant. Hence, with $m-n$ ancillary qubits, one is able to realize a $t$-qubit QLDPC code state with $t = \frac{m-n}{1-R}$ and code distance $d=\gamma t = \frac{\gamma (m-n)}{1-R}$ proportional to $m-n$. Since the target is an $n$-qubit state, one can realize an $n$-qubit QLDPC state $\ket{\psi}$ with code distance $d = \Theta (n)$ when $m-n = \Omega(n)$. When $m-n=o(n)$, one can realize the tensor product of a $t$-qubit QLDPC code state $\ket{\phi}$ and $\ket{0^{n-t}}$, and we denote $\ket{\psi} = \ket{\phi}\ket{0^{n-t}}$ here.

Below, we present a lemma showing that $-\log F_{D,n}(\ket{\phi}\ket{0^{n-t}}) = \Omega(d/K^{5D})$ when $\ket{\phi}$ is a $t$-qubit QLDPC state with code distance $d$ and $D = O(\log \min(m-n, n))$. This result is a generalization of that in Ref.~\cite{bravyi2024entanglement} showing that $-\log F_{D,n}(\ket{\phi}) = \Omega(d)$ when $\ket{\phi}$ is an $n$-qubit QLDPC code state with code distance $d$. Our results allow the state to add additional ancillas.

\begin{lemma}\label{lemma:qlpdcpower}
Given a $t$-qubit, $s$-sparse, and $d$-distance QLDPC code, $\mathcal{C}$, for any logical state $\ket{\phi}\in \mathcal{C}$, we have
\begin{equation}
-\log F_{D,n}(\ket{\phi}\ket{0^{n-t}}) = \Omega (\frac{d}{K^{5D}}),
\end{equation}
as long as $d > (sK^D)^4$.
\end{lemma}

Note that Lemma~\ref{lemma:qlpdcpower} requires the code distance of $\ket{\phi}$ to be larger than $(sK^D)^4$, which can always be fulfilled with $D < \log_K d/4 - \log_K s$. The maximum value of $D$ gives the bound proportional to $s^4/K^D$, which cannot surpass a constant.


Utilizing Lemma~\ref{lemma:qlpdcpower} and the facts that one can prepare $\frac{m-n}{1-R}$-qubit $\frac{\gamma (m-n)}{1-R}$-distance QLDPC state with $m-n$ ancillary qubits and that the code distance of an $n$-qubit code is always smaller than $n$, we have $d = \min(m-n, n)$. Then, we prove the lower bound summarized as follows.

\begin{equation}
-\log F_{D,n}^{L,m} =
\begin{cases}
0 & \text{ if } m-n=O(1),\\
O(\min(m-n, n)/K^{5D}) & \text{ if } m-n=\omega(1) \text{ and } m - n > \frac{1-R}{\gamma}s^4K^{4D}.
\end{cases}
\end{equation}

Particularly, when $K$ and $D$ are both constant, the condition that $m-n = \omega(1)$ automatically implies $m - n > \frac{1-R}{\gamma}s^4K^{4D}$. We can simplify the above equation to
\begin{equation}\label{eq:Flowerbound}
-\log F_{D,n}^{L,m} =
\begin{cases}
0 & \text{ if } m-n=O(1),\\
O(\min(m-n, n)) & \text{ if } m-n=\omega(1).
\end{cases}
\end{equation}
Combining Eqs.~\eqref{eq:Fupperbound} and~\eqref{eq:Flowerbound} gives Theorem~\ref{thm:adaptpower}.


Below, we present the proof of Lemma~\ref{lemma:qlpdcpower}. Our target is to get the lower bound of
\begin{equation}
-\log F_{D,n}(\ket{\phi}\ket{0^{n-t}}) = \min_{V:\mathrm{depth}(V)=D}  -\log \abs{\bra{\phi}\bra{0^{n-t}} V \ket{0^{n}}}^2.
\end{equation}
For further elaboration, we denote $\ket{\phi_V} = V^{\dagger} \ket{\phi}\ket{0^{n-t}}$. Since $\ket{\phi}$ is a $s$-sparse QLDPC code state, and $V$ has depth $D$, $\ket{\phi_V}$ is a $sK^D$-sparse QLDPC code state. Also, $\abs{\bra{\phi}\bra{0^{n-t}} V \ket{0^{n}}}^2 = \abs{\braket{0^n}{\phi_V}}^2$ is a computational-basis measurement outcome of $\ket{\phi_V}$. Denote
\begin{equation}
\Pr_{\phi_V}(x) = \abs{\braket{x}{\phi_V}}^2
\end{equation}
as the distribution of the computational-basis measurement outcomes of $\ket{\phi_V}$. To give an upper bound of $-\log \abs{ \bra{\phi} \bra{0^{n-t}} V \ket{0^{n}} }^2$, we will use three facts about $\Pr_{\phi_V}(x)$ as shown in the following lemma, which follows from Lemma~2 in Ref.~\cite{bravyi2024entanglement}.

\begin{lemma}\label{lemma:qldpcdistribution}
Suppose $d > (sK^D)^4$. Then for probability distribution $\Pr_{\phi_V}(x) = \abs{\braket{x}{\phi_V}}^2$, we have
\begin{enumerate}
\item\label{item:independence}
Each bit of $x$ is independent of all but at most $J$ other bits,
\item\label{item:bound}
$\mathbb{E}_{\phi_V}(\abs{x})\leq -(J+1)\log \abs{\braket{0^n}{\phi_V}}^2$.
\item\label{item:chernoff}
There exists a function $c(s)$ such that $\Pr_{\phi_V}(\abs{x}\geq t)\leq e^{\mathbb{E}_{\phi_V}(\abs{x})-t}$ for any $t\geq c(sK^D)\mathbb{E}_{\phi_V}(\abs{x})$.
\end{enumerate}
Here, $J = (sK^D)^2+(sK^D)^4$ is a constant; $\abs{x} = \sum_{i=1}^n x_i$ is the Hamming weight of bit string $x$.
\end{lemma}

With Lemma~\ref{lemma:qldpcdistribution}, we prove Lemma~\ref{lemma:qlpdcpower}. Denote $A = -(J+1)\log \abs{\braket{0^n}{\phi_V}}^2$. If $d < c(sK^D) AK^D$, then $-\log \abs{\braket{0^n}{\phi_V}}^2 > \frac{d}{(J+1)K^Dc(sK^D)}$, and the results hold. So, we only need to consider the case of $d \geq c(sK^D)AK^D$. Define a quantity of $S_i = \Pr[\abs{x}=i]$,
\begin{equation}
S_i = \sum_{1\leq p_1 < p_2 < \cdots < p_i\leq n} \tr{ V\ketbra{1_{p_1}1_{p_2}\cdots 1_{p_i}}V^{\dagger} (\ketbra{\phi}\otimes\ketbra{0^{n-t}}) },
\end{equation}
which is the probability of the weight of the measurement result being $i$. Note that $S_i$ is the same for all logical states $\ket{\phi}\in \mathcal{C}$ as long as $i < \frac{d}{K^D}$. This can be derived from the following expression of $S_i$.
\begin{align}
S_i &= \sum_{1\leq p_1 < p_2 < \cdots < p_i\leq n} \tr_{\phi}(O_{p_1,p_2,\cdots, p_i} \ketbra{\phi});\\
O_{p_1,p_2,\cdots, p_i} &= \tr_{\Bar{\phi}}(V\ketbra{1_{p_1}1_{p_2}\cdots 1_{p_i}}V^{\dagger} \ketbra{0^{n-t}}),
\end{align}
where $\tr_{\phi}$ and $\tr_{\Bar{\phi}}$ means the trace over the space of the support of $\ket{\phi}$ and the trace over the space outside the support of $\ket{\phi}$, respectively. Since the depth of $V$ is $D$, based on Lemma~\ref{lemma:operatorweightgrowth}, we can deduce $\wt(O_{p_1,p_2,\cdots, p_i})\leq iK^D$. Thus, $i < \frac{d}{K^D}$ implies that $\wt(O_{p_1,p_2,\cdots, p_i}) < d$. Note that due to the local indistinguishability property of logical states, all $d-1$-local density matrices of the logical states within code space $\mathcal{C}$ of code distance $d$ are the same~\cite{bravyi2024entanglement}. As a result, $S_i$ would be the same for all logical states $\ket{\phi}\in \mathcal{C}$ for $i < \frac{d}{K^D}$.

Denote $d_K^D=\frac{d}{K^D}$. Note that based on Bonferroni's inequality, we have~\cite{bravyi2024entanglement}
\begin{equation}
\begin{split}
&(-1)^{d_K^D-1}(\Pr[\abs{x} > 0] - \sum_{i=1}^{d_K^D-1}S_i)\\
=&\sum_{t = d_K^D}^n \binom{t-1}{d_K^D-1}\Pr[\abs{x}=t]\\
=&\sum_{t = d_K^D}^n \binom{t-2}{d_K^D-2}\Pr[\abs{x}\geq t]\\
\leq &\sum_{t = d_K^D}^{\infty} \binom{t-2}{d_K^D-2}e^{A-t}\\
=&e^{A-1}(e-1)^{1-d_K^D}.
\end{split}
\end{equation}
The inequality in the fourth line utilizes conditions~\ref{item:bound} and~\ref{item:chernoff}. The requirement $t\geq c(sK^D)\mathbb{E}(\abs{x})$ in condition~\ref{item:chernoff} satisfies as $t\geq d_K^D = \frac{d}{K^D}\geq c(sK^D)A$ and $A\geq \mathbb{E}(\abs{x})$. Substituting $\abs{\braket{0^n}{\phi_V}}^2 = \Pr[\abs{x}=0] = 1-\Pr[\abs{x} > 0]$ to the above equation, we get
\begin{equation}
1-\sum_{i=1}^{d_K^D-1}S_i-e^{A-1}(e-1)^{1-d_K^D}\leq \abs{\braket{0^n}{\phi_V}}^2\leq 1-\sum_{i=1}^{d_K^D-1}S_i+e^{A-1}(e-1)^{1-d_K^D}.
\end{equation}
Note that the above inequality holds for all logical states $\ket{\phi}$, and the left side and the right side are independent of the specific logical state. Since $\dim \mathcal{C} = 2^{t} \geq 2$, there exists logical state $\ket{\phi'}$ such that $\abs{\braket{0^n}{\phi'_V}}^2 = 0$. As a result, we obtain
\begin{equation}
-e^{A-1}(e-1)^{1-d_K^D}\leq 1-\sum_{i=1}^{d_K^D-1}S_i\leq e^{A-1}(e-1)^{1-d_K^D}.
\end{equation}
Thus,
\begin{equation}
\abs{\braket{0^n}{\phi_V}}^2\leq 2e^{A-1}(e-1)^{1-d_K^D}.
\end{equation}
Equivalently,
\begin{equation}
-\log \abs{\braket{0^n}{\phi_V}}^2 \geq \frac{1-\log2+(d_K^D-1)\log(e-1)}{J+2}.
\end{equation}
In summary, we get
\begin{equation}
\begin{split}
-\log \abs{\braket{0^n}{\phi_V}}^2 &\geq \min(\frac{1-\log2+(\frac{d}{K^D}-1)\log(e-1)}{J+2}, \frac{d}{(J+1)K^Dc(s)})\\
&= \min(\frac{1-\log2+(\frac{d}{K^D}-1)\log(e-1)}{(sK^D)^2+(sK^D)^4+2}, \frac{d}{((sK^D)^2+(sK^D)^4+1)K^Dc(s)}).
\end{split}
\end{equation}
Since $s=O(1)$, we obtain the result of Lemma~\ref{lemma:qlpdcpower}, $-\log F_{D,n}(\ket{\phi}\ket{0^{n-t}}) = \Omega(\frac{d}{K^{5D}})$.

Note that $d = \Theta(n)$ for good QLDPC codes. Based on the above derivation, we obtain that there exists a small constant $c$ such that if $D\leq c\log n$, the term $-\log \abs{\braket{0^n}{\phi_V}}^2$ of the $n$-qubit good QLDPC code state will be larger than a positive constant. This implies that preparing the $n$-qubit good QLDPC code state above a certain constant fidelity threshold requires a circuit lower bound $\Omega(\log n)$.


\subsection{Proof of Lemma~\ref{lemma:correlation}}\label{appendssc:thmpf:correlation}
Suppose the anti-shallowness of $\ket{\psi}$ is given by $-\log f$, meaning that the maximal fidelity between $\ket{\psi}$ and any state $\ket{\phi}$ prepared by a shallow non-adaptive circuit is $f$. The difference in their correlation functions satisfies
\begin{equation}
\begin{aligned}
&\left|\mathrm{Cor}(O_i,O_j,\ket{\psi})-\mathrm{Cor}(O_i,O_j,\ket{\phi})\right|\\
\leq&\left|\bra{\psi}O_1O_2\ket{\psi}-\bra{\phi}O_1O_2\ket{\phi}\right|+\left|\bra{\psi}O_1\ket{\psi}\bra{\psi}O_2\ket{\psi}-\bra{\phi}O_1\ket{\phi}\bra{\phi}O_2\ket{\phi}\right|\\
\leq&\left|\bra{\psi}O_1O_2\ket{\psi}-\bra{\phi}O_1O_2\ket{\phi}\right|+\left|\bra{\psi}O_1\ket{\psi}\right|\left|\bra{\psi}O_2\ket{\psi}-\bra{\phi}O_2\ket{\phi}\right|+\left|\bra{\phi}O_2\ket{\phi}\right|\left|\bra{\psi}O_1\ket{\psi}-\bra{\phi}O_1\ket{\phi}\right|\\
\leq&2\sqrt{1-f}+4\sqrt{1-f}=6\sqrt{1-f}.
\end{aligned}
\end{equation}

In the final step, we use the following bound, derived via H$\ddot{\text{o}}$lder's inequality:
\begin{equation}\label{eq:Holdereq}
\begin{aligned}
\left|\bra{\psi}O\ket{\psi}-\bra{\phi}O\ket{\phi}\right|&\leq\tr(\left|O\left(\ketbra{\psi}-\ketbra{\phi}\right)\right|)\\
&\leq \norm{O}_{\infty}\cdot\norm{\ketbra{\psi}-\ketbra{\phi}}_1\\
&\leq \norm{O}_{\infty}\cdot 2\sqrt{1-f},
\end{aligned}
\end{equation}
for any operator $O$, where $\norm{\cdot}_{\infty}$ and $\norm{\cdot}_{1}$ are the Schatten infinity norm and trace norm, respectively. Note that we have $\norm{O_i}_{\infty}, \norm{O_j}_{\infty} \leq 1$.

Since $\ket{\phi}$ can be prepared by a shallow non-adaptive circuit, $(m-n+1)K^{2L-1}$ is a constant. Thus, as shown in the proof of Theorem~\ref{thm:adalowerboundgenericcorr} in Appendix~\ref{appendssc:prop2pf:permutation-invariant}, there must be two qubits labeled by $i$ and $j$ such that $\mathrm{Cor}(O_i,O_j,\ket{\phi})=0$ for any $O_i$ acting on $i$ and $O_j$ acting on $j$. Thus, the corresponding correlation function on $\ket{\psi}$ can be bounded: $\left|\mathrm{Cor}(O_i,O_j,\ket{\psi})\right|\leq 6\sqrt{1-f}$. By definition, the global correlation should also satisfies $\mathrm{Cor}(\ket{\psi})\leq 6\sqrt{1-f}$, which solves $f\leq 1-\frac{1}{36}\mathrm{Cor}(\ket{\psi})^2$ and
\begin{equation}
-\log F_{D,n}(\psi)\geq -\log \left(1-\frac{1}{36}\mathrm{Cor}(\ket{\psi})^2\right).
\end{equation}
As a remark, substituting global correlation $\mathrm{Cor}(\ket{\psi})$ with correlation strength $\mathrm{Cor}_w^{[n]}(\ket{\psi})$ where $w = o(n)$ still makes the above inequality hold. For simplicity, we

\section{Adaptive QLDPC code state preparation}\label{appendsc:qlpdc}

\subsection{Proof of Lemma~\ref{lemma:qldpcmeasure}}\label{appendsc:lemmapf:qldpcmeasure}
The proof of Lemma~\ref{lemma:qldpcmeasure} is given by the construction of the measurement circuit. Recall that the stabilizer $S_j$ is a tensor product of Pauli operators, $S_j = \bigotimes_{i=1}^n P_i$, where $P_i\in \{\id, X, Y, Z\}$. One can use an ancillary qubit and controlled Pauli operations between the ancillary qubit and the data qubits to measure observable $S_j$, as shown in Fig.~\ref{fig:Paulimeasure}. The depth of the measurement circuit equals $\wt(S_j)+1$, including the contribution of both gates and measurements. Note that the depth of the Hadamard gate can be ignored by merging it into the controlled Pauli operation. Also, the order of the controlled Pauli operations does not matter in this case.

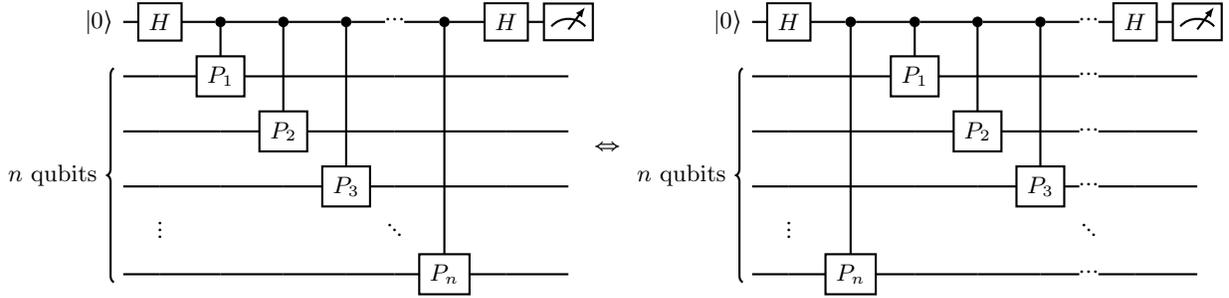
\begin{figure}[!htbp]
 \centering
  \begin{quantikz}[row sep=0.2cm, column sep=0.2cm]
  \lstick{$\ket{0}$} & \gate{H} & \ctrl{1} & \ctrl{2}& \ctrl{3} & \qw \cdots & \ctrl{5}  & \gate{H} & \meter{} \\
  \lstick[wires=5]{$n$ qubits} & \qw & \gate{P_1} & \qw & \qw & \qw & \qw  & \qw & \qw \\
  \lstick{} & \qw & \qw  & \gate{P_2} & \qw & \qw & \qw & \qw & \qw \\
  \lstick{} & \qw & \qw  & \qw & \gate{P_3} & \qw & \qw & \qw & \qw \\
  \lstick{} & \vdots &  &  &  & \ddots &  &  & \\
  \lstick{} & \qw & \qw & \qw & \qw & \qw & \gate{P_n} & \qw & \qw \\
  \end{quantikz}
  $\Leftrightarrow$
  \begin{quantikz}[row sep=0.2cm, column sep=0.2cm]
  \lstick{$\ket{0}$} & \gate{H} & \ctrl{5} & \ctrl{1}& \ctrl{2} & \ctrl{3} & \qw \cdots & \gate{H} & \meter{} \\
  \lstick[wires=5]{$n$ qubits} & \qw & \qw & \gate{P_1} & \qw & \qw & \qw  \cdots & \qw & \qw \\
  \lstick{} & \qw & \qw & \qw  & \gate{P_2} & \qw & \qw \cdots & \qw & \qw \\
  \lstick{} & \qw & \qw & \qw  & \qw & \gate{P_3} & \qw \cdots & \qw & \qw \\
  \lstick{} & \vdots &  &  &  &  & \ddots &  & \\
  \lstick{} & \qw & \gate{P_n} & \qw & \qw & \qw & \qw \cdots & \qw & \qw \\
  \end{quantikz}
\caption{The circuit to measure a Pauli observable. One first prepares state $\ket{+}$ by adding a Hadamard gate, $H$, to the ancillary qubit. Then, a sequence of controlled Pauli operations between the ancillary qubit and data qubits is applied. Finally, one measures on $X$ basis with a Hadamard gate and computational-basis measurement. The final measurement result can be viewed as the outcome of measuring the $n$-qubit state with observable $\bigotimes_{i=1}^n P_i$. Note that if $P_i$ equals $\id$, then the controlled-$P_i$ operation is just identity and does not need to be realized. Note that when measuring a single Pauli operator, the order of the controlled Pauli operations does not matter, as shown by the two equivalent circuits.}
\label{fig:Paulimeasure}
\end{figure}

To measure two Pauli observables, we can use two ancillary qubits and implement two circuits like Fig.~\ref{fig:Paulimeasure} sequentially. Nonetheless, in this sense, the circuit depth will be proportional to the number of measured Pauli observables. If we try to measure them simultaneously at a short depth, a strategy is to combine the two measurement circuits and arrange the controlled Pauli operations parallelly. We depict an example of measuring $XXXX$ and $ZZZZ$ simultaneously within depth $5$ in Fig.~\ref{fig:simulPaulimeasure}.

\begin{figure}[!htbp]
 \centering
  \begin{quantikz}[row sep=0.2cm, column sep=0.2cm]
  \lstick{$\ket{0}$} & \gate{H} & \ctrl{2} & \qw  & \ctrl{3} & \qw  & \ctrl{4} & \qw & \ctrl{5} & \qw & \gate{H} & \meter{} \\
  \lstick{$\ket{0}$} & \gate{H} & \qw  & \ctrl{3} & \qw  & \ctrl{4} & \qw  & \ctrl{1} & \qw & \ctrl{2} & \gate{H} & \meter{} \\
  \lstick{} & \qw & \targ{} & \qw & \qw & \qw & \qw  & \ctrl{} & \qw & \qw & \qw & \qw \\
  \lstick{} & \qw & \qw  & \qw & \targ{} & \qw & \qw & \qw & \qw & \ctrl{} & \qw & \qw \\
  \lstick{} & \qw & \qw  & \ctrl{} & \qw & \qw & \targ{} & \qw & \qw & \qw & \qw & \qw \\
  \lstick{} & \qw & \qw & \qw & \qw & \ctrl{} & \qw & \qw & \targ{} & \qw & \qw & \qw \\
  \end{quantikz}
  \begin{tikzpicture}[overlay]
  \node[fit={(-4.79,1.6) (-4.39,-1.6)}, draw=red, dashed, line width=0.4mm, rounded corners, inner sep=0.2cm, label=above:{\textbf{\red{One depth}}}] {};
  \end{tikzpicture}
  $\Leftrightarrow$
  \begin{quantikz}[row sep=0.2cm, column sep=0.2cm]
  \lstick{$\ket{0}$} & \gate{H} & \ctrl{2} & \ctrl{3} & \ctrl{4} & \ctrl{5} & \qw & \qw & \qw & \qw & \gate{H} & \meter{} \\
  \lstick{$\ket{0}$} & \gate{H} & \qw & \qw & \qw & \qw & \ctrl{1} & \ctrl{2} & \ctrl{3} & \ctrl{4} & \gate{H} & \meter{} \\
  \lstick{} & \qw & \targ{} & \qw & \qw & \qw & \ctrl{}  & \qw & \qw & \qw & \qw & \qw \\
  \lstick{} & \qw & \qw  & \targ{} & \qw & \qw & \qw & \ctrl{} & \qw & \qw & \qw & \qw \\
  \lstick{} & \qw & \qw  & \qw & \targ{} & \qw & \qw & \qw & \ctrl{} & \qw & \qw & \qw \\
  \lstick{} & \qw & \qw & \qw & \qw & \targ{} & \qw & \qw & \qw & \ctrl{} & \qw & \qw \\
  \end{quantikz}
\caption{A depth-$5$ circuit to measure Pauli observable $XXXX$ and $ZZZZ$ simultaneously. One initializes $\ket{+}$ for each ancillary qubit. Then, the controlled Pauli operations between the ancillary qubits and data qubits are applied parallelly. For instance, the CNOT gate between the first ancillary qubit and the first data qubit and the CZ gate between the second ancillary qubit and the third data qubit can be implemented in one depth. Finally, one measures both ancillary qubits on $X$ basis.}
\label{fig:simulPaulimeasure}
\end{figure}
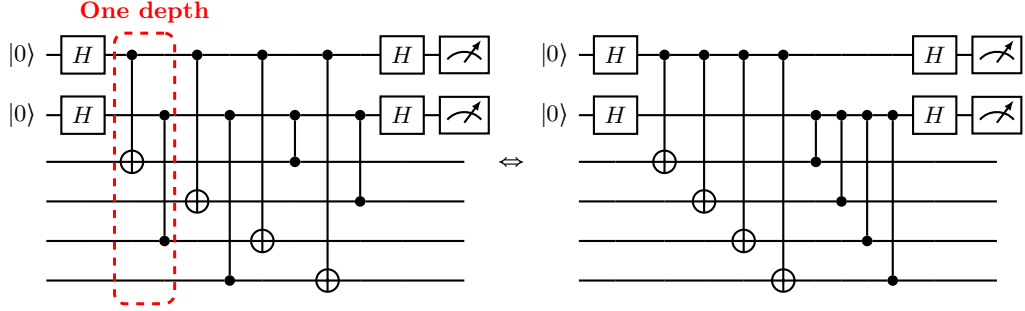

Note that there are different ways to arrange the controlled Pauli operations. To ensure that in one step, different ancillary qubits are all applied with controlled Pauli operations, the only requirement is that the data qubits associated with different ancillary qubits are different. Nonetheless, not all arrangements would make the circuit equal to the sequential measurement of two Pauli observables~\cite{Geher2024Tangling}, as demonstrated in Fig.~\ref{fig:simulPaulimeasureNocommute}.

\begin{figure}[!htbp]
 \centering
  \begin{quantikz}[row sep=0.2cm, column sep=0.2cm]
  \lstick{$\ket{0}$} & \gate{H} & \qw & \ctrl{2} & \qw  & \ctrl{3} & \qw  & \ctrl{4} & \qw & \ctrl{5} & \gate{H} & \meter{} \\
  \lstick{$\ket{0}$} & \gate{H} & \ctrl{4} & \qw  & \ctrl{1} & \qw  & \ctrl{2} & \qw  & \ctrl{3} & \qw & \gate{H} & \meter{} \\
  \lstick{} & \qw & \qw & \targ{} & \ctrl{} & \qw & \qw & \qw  & \qw & \qw & \qw & \qw \\
  \lstick{} & \qw & \qw & \qw  & \qw & \targ{} & \ctrl{} & \qw & \qw & \qw & \qw & \qw \\
  \lstick{} & \qw & \qw & \qw  & \qw & \qw & \qw & \targ{} & \ctrl{} & \qw & \qw & \qw \\
  \lstick{} & \qw & \ctrl{} & \qw & \qw & \qw & \qw & \qw & \qw & \targ{} & \qw & \qw \\
  \end{quantikz}
  \begin{tikzpicture}[overlay]
  \node[fit={(-4.79,1.6) (-4.39,-1.6)}, draw=red, dashed, line width=0.4mm, rounded corners, inner sep=0.2cm, label=above:{\textbf{\red{One depth}}}] {};
  \end{tikzpicture}
  $\Leftrightarrow$
  \begin{quantikz}[row sep=0.2cm, column sep=0.2cm]
  \lstick{$\ket{0}$} & \gate{H} & \ctrl{2} & \ctrl{3} & \ctrl{4} & \ctrl{5} & \qw & \qw & \qw & \qw & \ctrl{1} & \gate{H} & \meter{} \\
  \lstick{$\ket{0}$} & \gate{H} & \qw & \qw & \qw & \qw & \ctrl{1} & \ctrl{2} & \ctrl{3} & \ctrl{4} & \ctrl{} & \gate{H} & \meter{} \\
  \lstick{} & \qw & \targ{} & \qw & \qw & \qw & \ctrl{}  & \qw & \qw & \qw & \qw & \qw \\
  \lstick{} & \qw & \qw  & \targ{} & \qw & \qw & \qw & \ctrl{} & \qw & \qw & \qw & \qw \\
  \lstick{} & \qw & \qw  & \qw & \targ{} & \qw & \qw & \qw & \ctrl{} & \qw & \qw & \qw \\
  \lstick{} & \qw & \qw & \qw & \qw & \targ{} & \qw & \qw & \qw & \ctrl{} & \qw & \qw \\
  \end{quantikz}
  \begin{tikzpicture}[overlay]
  \node[fit={(-1.73,1.6) (-1.715,0.7)}, draw=blue, dashed, line width=0.4mm, rounded corners, inner sep=0.2cm, label=above:{\textbf{\blue{Additional operation}}}] {};
  \end{tikzpicture}
\caption{A circuit equivalence to depict the influence of the gate order on the total operation of the circuit. One has to implement an additional CZ gate to make the circuit measure the two Pauli observables simultaneously.}
\label{fig:simulPaulimeasureNocommute}
\end{figure}
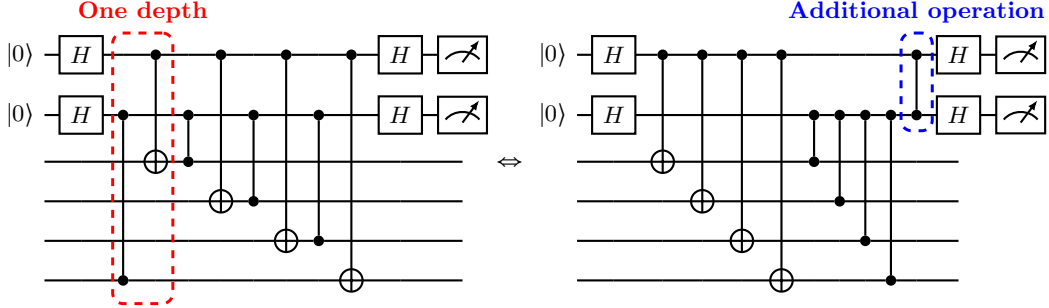

Denote $C_jP_i$ as the controlled-$P$ operation between the $i$-th ancillary qubit and $j$-th data qubit. Notice that $C_1X_4$ and $C_2Z_4$ both need to be implemented in the circuit, and the two gates do not commute. The order in which they are implemented influences the total operation of the circuit. The same is true for other non-commutative $C_jP_i$ operations. For ancillary qubits 1 and 2, one can find all the non-commutative controlled Pauli operation pairs $\{(C_1P_{i1}, C_2P_{i2})\}$. If the number of non-commutative pairs in which $C_1P_{i1}$ are implemented before $C_2P_{i2}$ is even, then the total operation of the circuit equals the sequential measurement of two Pauli observables. Otherwise, there will be an additional CZ gate acting on the two ancillary qubits. For further elaboration, we call the two ancillary qubits to be tangled. In this case, we have to apply an additional CZ gate as a correction before the $X$-basis measurement. The circuit depth will be 6 in total. Note that the additional operation acting on the ancillary qubit will always be identity or CZ gate due to the following commutation rule~\cite{Geher2024Tangling}.
\begin{equation}
[C_iX_k, C_jY_k] = [C_iX_k, C_jZ_k] = [C_iY_k, C_jZ_k] = C_iZ_j.
\end{equation}

Now, we turn to the case of measuring observables in $\mathcal{S} = \{S_1, S_2, \cdots, S_t\}$ simultaneously. We also use the idea of parallelizing the measurement circuit of each stabilizer and implementing additional CZ gates. We will use the definition of the Tanner graph of a set of $n$-qubit operators $\mathcal{O} = \{O_1, O_2, \cdots, O_t\}$ shown below.

\begin{definition}
The Tanner graph $G_{\mathcal{O}}$ of a set of $n$-qubit operators $\mathcal{O} = \{O_1, O_2, \cdots, O_t\}$ is a bipartite graph $G_{\mathcal{C}}=(U, V, E)$ defined as follows: The vertex sets $U=\{u_i\}_{i=1}^n$, $V=\{v_j\}_{j=1}^{t}$ represent the qubits and the operators $\{O_1, O_2, \cdots, O_{t}\}$ respectively, and the edge set $E=\{e_{ij}\}$ contains an edge $e_{ij}=(u_i,v_j)$ iff $i\in \mathrm{supp}(O_j)$.
\end{definition}

With the definition of the Tanner graph, each stabilizer $S_j$ has an expression, $S_j = \bigotimes_{e_{ij}\in E} S_j(i)$. Note that $\wt(S_i)\leq s$. For $\mathcal{S} = \{S_1, S_2, \cdots, S_t\}$, the maximum degree of Tanner graph $G_{\mathcal{S}}$, denoted by $\Delta(G_{\mathcal{S}})$, is no larger than the sparsity $s$. Since $G_{\mathcal{S}}$ is a bipartite graph, its edge chromatic number $\chi_E(G_{\mathcal{S}})=\Delta(G_{\mathcal{S}})=s$. Therefore, one can divide the edge set $E$ into $\chi_E(G_{\mathcal{S}})$ mutually disjoint matchings, $E=\sum_{l=1}^{\chi_E(G_{\mathcal{S}})}E_l$. The measurement circuit is defined as follows.

\begin{enumerate}
\item Denote the input $n$-qubit state as $\ket{\phi}$. One first initializes $t$-qubit state as $\ket{0^t}$ on the ancillary system. Adopt a layer of $H$ gates on each ancillary qubit. The index of the ancillary qubits is labeled as $\{1, 2, \cdots, t\}$ corresponding to the stabilizer set $\{S_1, S_2, \cdots, S_t\}$.
\item Recall that $E=\sum_{l=1}^{\chi_E(G_{\mathcal{S}})}E_l$. In depth $l$, apply a controlled-$S_j(i)$ gate on the $j$-th ancillary qubit and the $i$-th qubit for all $e_{ij}\in E_l$. Since $E_l$ is a matching of $G_{\mathcal{C}}$, these gates will not intersect with each other and can be implemented simultaneously in one depth. This procedure is realized from depth 1 to $\chi_E(G_{\mathcal{C}})=s$.
\item Define a new graph $G_T = (U, E')$ with the vertex set $U$ containing all ancillary qubits. For each pair of ancillary qubits $(i,j)$, determine whether the two qubits are tangled based on the circuit in the previous step. If they are tangled, then the edge $e'_{ij} = (u_i, u_j)$ is contained in the edge set $E'$. Note that the edge set $E'$ of $G_T$ represents all additional CZ gates that have to be implemented to make the whole circuit measure $\mathcal{S} = \{S_1, S_2, \cdots, S_t\}$ simultaneously. Since $\wt(S_j)\leq s$ and $\# S^i\leq s$, the maximum degree of $G_T$, $\Delta(G_T)\leq s^2$. The chromatic number $\chi_{E}(G_{T})\leq \Delta(G_{\mathcal{S}})+1\leq s^2+1$. Therefore, one can divide the edge set $E'$ into $\chi_E(G_{T})$ mutually disjoint matching, $E' = \sum_{l'=1}^{\chi_E(G_{T})} E'_{l'}$.
\item Applying additional CZ gates among ancillary qubits according to the mutually disjoint matching, $E' = \sum_{l'=1}^{\chi_E(G_{T})} E'_{l'}$. In depth $l'$, apply CZ gates between $i$-th and $j$-th ancillary qubits for all $e'_{ij}\in E'_{l'}$. This procedure is realized from depth $s+1$ to $s+\chi_E(G_{T})$.
\item Adopt a layer of $H$ gates on the $t$ ancillary qubits and measure them on the $Z$ basis.
\end{enumerate}
Considering the contribution of all gates and measurements, the total depth of the circuit is no larger than $2+s+s^2$, which finishes the proof.

\subsection{Complexity of decoding Pauli correction operations}\label{appendsc:decode}
Provided with $n$ independent and mutually commutative Pauli operators, $\{S_1, S_2, \cdots, S_n\}$ and divide them into two sets $\mathcal{S}^+$ and $\mathcal{S}^-$. Our first task is to find a Pauli operator $P_c$ satisfying the following conditions,
\begin{equation}\label{eq:commutation}
\begin{split}
&\forall S^+\in \mathcal{S}^+, [P_c, S^+] = 0 \Leftrightarrow P_cS^+P_cS^+=\id;\\
&\forall S^-\in \mathcal{S}^-, \{P_c, S^-\} = 0 \Leftrightarrow P_cS^-P_cS^-=-\id.
\end{split}
\end{equation}
For simplicity, we denote each stabilizer $S_i = a_i\prod_{j}X_j^{x^i_j}Z_j^{z^i_j}$ where $a_i\in \{\pm 1, \pm i\}$, $x^i_j, z^i_j\in \{0, 1\}$, and $X_j$ and $Z_j$ are Pauli $X$ and $Z$ operators on qubit $j$, respectively. Similarly, we denote $P_c = a \prod_{j}X_j^{x_j}Z_j^{z_j}$ where $a\in \{\pm 1, \pm i\}$, $x_j, z_j\in \{0, 1\}$. Then,
\begin{equation}
\begin{split}
[P_c, S_i] &= 0 \Leftrightarrow \bigoplus_{j} (z_j^i x_j\oplus x_j^i z_j) = 0;\\
\{P_c, S_i\} &= 0 \Leftrightarrow \bigoplus_{j} (z_j^i x_j\oplus x_j^i z_j) = 1.
\end{split}
\end{equation}
Thus, Eq.~\eqref{eq:commutation} is equivalent to a binary linear equation of $(x_j, z_j)$ over $\{0, 1\}^{2n}$:
\begin{equation}
\begin{pmatrix}
z^1_1 & z^1_2 & \cdots & z^1_n & x^1_1 & x^1_2 & \cdots & x^1_n\\
z^2_1 & z^2_2 & \cdots & z^2_n & x^2_1 & x^2_2 & \cdots & x^2_n\\
\vdots & \vdots & \ddots & \vdots & \vdots & \vdots & \ddots & \vdots\\
z^n_1 & z^n_2 & \cdots & z^n_n & x^n_1 & x^n_2 & \cdots & x^n_n\\
\end{pmatrix}
\begin{pmatrix}
x_1\\
x_2\\
\vdots\\
x_n\\
z_1\\
z_2\\
\vdots\\
z_n
\end{pmatrix}=b,
\end{equation}
where $b = (b_1, b_2, \cdots, b_n)^T$ is a bit string in $\{0, 1\}^n$ representing the commutation relation between $S_i$ and $P_c$, and $P_cS_iP_cS_i=b_i\id$. Note that the coefficient matrix is full row-rank due to the independence of the stabilizer generators $\{S_1, S_2, \cdots, S_n\}$. Row $i$ equaling the linear combination of other rows is equivalent to $S_i$ equaling the multiplication of other stabilizer generators. Thus, the solution exists, and a Pauli operator solution can be found with an algorithm of complexity $O(n^2)$ by solving the linear equation over $\mathbb{GF}(2)$~\cite{Cetin1991GF2}.

Since the dimension of the coefficient matrix is $n\times 2n$, the solution of $(x_j, z_j)$ is not unique. Any solution plus a vector in the null space of the coefficient matrix is still a solution. The null space of the coefficient matrix corresponds to all Pauli operators commuting with all stabilizer generators, which is just the stabilizer group formed by the stabilizer generators. Thus, the Pauli operator solution is unique in the multiplication of stabilizer generators.

\subsection{Finding $X$-type logical operations for a stabilizer code}\label{appendsc:xtypelogical}
Given a stabilizer code with generators ${S_1, S_2, \cdots, S_{n-k}}$, our second task is to find $k$ independent logical operators $L_1, L_2, \cdots, L_k$ from $\{I, X\}^{\otimes n}$. The stabilizer code can be represented by a stabilizer table, where each generator $S = \prod_j X_j^{x_j} Z_j^{z_j}$ corresponds to a row
$\begin{pmatrix}
x_1 & x_2 & \cdots & x_n & z_1 & z_2 & \cdots & z_n
\end{pmatrix}$. This table can be brought into the canonical form~\cite{Aaronson2004stabilizer}.
\begin{equation}
\begin{pmatrix}
A & B\\
C & 0
\end{pmatrix},
\end{equation}
where $B$ is an $r \times n$ full-rank matrix, $A$ is $r \times n$, and $C$ is $(n-k-r) \times n$ with rank $n-k-r$. Since $B$ has rank $r$, its row space has an orthogonal complement of dimension $n - r$. The rows of $C$, satisfying $B^T C = 0$, occupy $n-k-r$ dimensions of this orthogonal space, leaving $k$ dimensions. Let $D$ be a $k \times n$ matrix whose rows span the remaining orthogonal subspace so that $B^T D = 0$, and the matrix $\begin{pmatrix}
A & B\\
C & 0 \\
D & 0
\end{pmatrix}$ has full rank $n$. Then, the $X$-type logical operators can be defined as
\begin{equation}
L_i = \prod_{j=1}^n X_j^{D_{ij}}.
\end{equation}
The stabilizer table $\begin{pmatrix}
A & B\\
C & 0 \\
D & 0
\end{pmatrix}$ naturally defines the target stabilizer state within the given stabilizer code.

\bibliographystyle{apsrev4-1}

\bibliography{bibQLDPC}

\begin{thebibliography}{62}%
\makeatletter
\providecommand \@ifxundefined [1]{%
 \@ifx{#1\undefined}
}%
\providecommand \@ifnum [1]{%
 \ifnum #1\expandafter \@firstoftwo
 \else \expandafter \@secondoftwo
 \fi
}%
\providecommand \@ifx [1]{%
 \ifx #1\expandafter \@firstoftwo
 \else \expandafter \@secondoftwo
 \fi
}%
\providecommand \natexlab [1]{#1}%
\providecommand \enquote  [1]{``#1''}%
\providecommand \bibnamefont  [1]{#1}%
\providecommand \bibfnamefont [1]{#1}%
\providecommand \citenamefont [1]{#1}%
\providecommand \href@noop [0]{\@secondoftwo}%
\providecommand \href [0]{\begingroup \@sanitize@url \@href}%
\providecommand \@href[1]{\@@startlink{#1}\@@href}%
\providecommand \@@href[1]{\endgroup#1\@@endlink}%
\providecommand \@sanitize@url [0]{\catcode `\\12\catcode `\$12\catcode
  `\&12\catcode `\#12\catcode `\^12\catcode `\_12\catcode `\%12\relax}%
\providecommand \@@startlink[1]{}%
\providecommand \@@endlink[0]{}%
\providecommand \url  [0]{\begingroup\@sanitize@url \@url }%
\providecommand \@url [1]{\endgroup\@href {#1}{\urlprefix }}%
\providecommand \urlprefix  [0]{URL }%
\providecommand \Eprint [0]{\href }%
\providecommand \doibase [0]{http://dx.doi.org/}%
\providecommand \selectlanguage [0]{\@gobble}%
\providecommand \bibinfo  [0]{\@secondoftwo}%
\providecommand \bibfield  [0]{\@secondoftwo}%
\providecommand \translation [1]{[#1]}%
\providecommand \BibitemOpen [0]{}%
\providecommand \bibitemStop [0]{}%
\providecommand \bibitemNoStop [0]{.\EOS\space}%
\providecommand \EOS [0]{\spacefactor3000\relax}%
\providecommand \BibitemShut  [1]{\csname bibitem#1\endcsname}%
\let\auto@bib@innerbib\@empty
\bibitem [{\citenamefont {Piroli}\ \emph {et~al.}(2021)\citenamefont {Piroli},
  \citenamefont {Styliaris},\ and\ \citenamefont {Cirac}}]{Piroli2021adaptive}%
  \BibitemOpen
  \bibfield  {author} {\bibinfo {author} {\bibfnamefont {L.}~\bibnamefont
  {Piroli}}, \bibinfo {author} {\bibfnamefont {G.}~\bibnamefont {Styliaris}}, \
  and\ \bibinfo {author} {\bibfnamefont {J.~I.}\ \bibnamefont {Cirac}},\ }\href
  {\doibase 10.1103/PhysRevLett.127.220503} {\bibfield  {journal} {\bibinfo
  {journal} {Phys. Rev. Lett.}\ }\textbf {\bibinfo {volume} {127}},\ \bibinfo
  {pages} {220503} (\bibinfo {year} {2021})}\BibitemShut {NoStop}%
\bibitem [{\citenamefont {Verresen}\ \emph {et~al.}(2022)\citenamefont
  {Verresen}, \citenamefont {Tantivasadakarn},\ and\ \citenamefont
  {Vishwanath}}]{verresen2022efficientlypreparingschrodingerscat}%
  \BibitemOpen
  \bibfield  {author} {\bibinfo {author} {\bibfnamefont {R.}~\bibnamefont
  {Verresen}}, \bibinfo {author} {\bibfnamefont {N.}~\bibnamefont
  {Tantivasadakarn}}, \ and\ \bibinfo {author} {\bibfnamefont {A.}~\bibnamefont
  {Vishwanath}},\ }\href {https://arxiv.org/abs/2112.03061} {\enquote {\bibinfo
  {title} {Efficiently preparing schr\"odinger's cat, fractons and non-abelian
  topological order in quantum devices},}\ } (\bibinfo {year} {2022}),\ \Eprint
  {http://arxiv.org/abs/2112.03061} {arXiv:2112.03061 [quant-ph]} \BibitemShut
  {NoStop}%
\bibitem [{\citenamefont {Yan}\ \emph {et~al.}(2025)\citenamefont {Yan},
  \citenamefont {Ma}, \citenamefont {Zhou},\ and\ \citenamefont
  {Ma}}]{Yan2025Variational}%
  \BibitemOpen
  \bibfield  {author} {\bibinfo {author} {\bibfnamefont {Y.}~\bibnamefont
  {Yan}}, \bibinfo {author} {\bibfnamefont {M.}~\bibnamefont {Ma}}, \bibinfo
  {author} {\bibfnamefont {Y.}~\bibnamefont {Zhou}}, \ and\ \bibinfo {author}
  {\bibfnamefont {X.}~\bibnamefont {Ma}},\ }\href {\doibase
  10.1103/PhysRevLett.134.170601} {\bibfield  {journal} {\bibinfo  {journal}
  {Phys. Rev. Lett.}\ }\textbf {\bibinfo {volume} {134}},\ \bibinfo {pages}
  {170601} (\bibinfo {year} {2025})}\BibitemShut {NoStop}%
\bibitem [{\citenamefont {Malz}\ \emph {et~al.}(2024)\citenamefont {Malz},
  \citenamefont {Styliaris}, \citenamefont {Wei},\ and\ \citenamefont
  {Cirac}}]{Malz2024MPS}%
  \BibitemOpen
  \bibfield  {author} {\bibinfo {author} {\bibfnamefont {D.}~\bibnamefont
  {Malz}}, \bibinfo {author} {\bibfnamefont {G.}~\bibnamefont {Styliaris}},
  \bibinfo {author} {\bibfnamefont {Z.-Y.}\ \bibnamefont {Wei}}, \ and\
  \bibinfo {author} {\bibfnamefont {J.~I.}\ \bibnamefont {Cirac}},\ }\href
  {\doibase 10.1103/PhysRevLett.132.040404} {\bibfield  {journal} {\bibinfo
  {journal} {Phys. Rev. Lett.}\ }\textbf {\bibinfo {volume} {132}},\ \bibinfo
  {pages} {040404} (\bibinfo {year} {2024})}\BibitemShut {NoStop}%
\bibitem [{\citenamefont {Piroli}\ \emph {et~al.}(2024)\citenamefont {Piroli},
  \citenamefont {Styliaris},\ and\ \citenamefont
  {Cirac}}]{piroli2024approximatingmanybodyquantumstates}%
  \BibitemOpen
  \bibfield  {author} {\bibinfo {author} {\bibfnamefont {L.}~\bibnamefont
  {Piroli}}, \bibinfo {author} {\bibfnamefont {G.}~\bibnamefont {Styliaris}}, \
  and\ \bibinfo {author} {\bibfnamefont {J.~I.}\ \bibnamefont {Cirac}},\ }\href
  {\doibase 10.1103/PhysRevLett.133.230401} {\bibfield  {journal} {\bibinfo
  {journal} {Phys. Rev. Lett.}\ }\textbf {\bibinfo {volume} {133}},\ \bibinfo
  {pages} {230401} (\bibinfo {year} {2024})}\BibitemShut {NoStop}%
\bibitem [{\citenamefont {Sun}\ \emph {et~al.}(2023)\citenamefont {Sun},
  \citenamefont {Tian}, \citenamefont {Yang}, \citenamefont {Yuan},\ and\
  \citenamefont {Zhang}}]{Sun2023Preparation}%
  \BibitemOpen
  \bibfield  {author} {\bibinfo {author} {\bibfnamefont {X.}~\bibnamefont
  {Sun}}, \bibinfo {author} {\bibfnamefont {G.}~\bibnamefont {Tian}}, \bibinfo
  {author} {\bibfnamefont {S.}~\bibnamefont {Yang}}, \bibinfo {author}
  {\bibfnamefont {P.}~\bibnamefont {Yuan}}, \ and\ \bibinfo {author}
  {\bibfnamefont {S.}~\bibnamefont {Zhang}},\ }\href {\doibase
  10.1109/TCAD.2023.3244885} {\bibfield  {journal} {\bibinfo  {journal} {Trans.
  Comp.-Aided Des. Integ. Cir. Sys.}\ }\textbf {\bibinfo {volume} {42}},\
  \bibinfo {pages} {3301–3314} (\bibinfo {year} {2023})}\BibitemShut
  {NoStop}%
\bibitem [{\citenamefont {Zi}\ \emph {et~al.}(2025)\citenamefont {Zi},
  \citenamefont {Nie},\ and\ \citenamefont
  {Sun}}]{zi2025constantdepthquantumcircuitsarbitrary}%
  \BibitemOpen
  \bibfield  {author} {\bibinfo {author} {\bibfnamefont {W.}~\bibnamefont
  {Zi}}, \bibinfo {author} {\bibfnamefont {J.}~\bibnamefont {Nie}}, \ and\
  \bibinfo {author} {\bibfnamefont {X.}~\bibnamefont {Sun}},\ }\href
  {https://arxiv.org/abs/2503.16208} {\enquote {\bibinfo {title}
  {Constant-depth quantum circuits for arbitrary quantum state preparation via
  measurement and feedback},}\ } (\bibinfo {year} {2025}),\ \Eprint
  {http://arxiv.org/abs/2503.16208} {arXiv:2503.16208 [quant-ph]} \BibitemShut
  {NoStop}%
\bibitem [{\citenamefont {Cao}\ and\ \citenamefont
  {Eisert}(2025)}]{cao2025measurementdrivenquantumadvantagesshallow}%
  \BibitemOpen
  \bibfield  {author} {\bibinfo {author} {\bibfnamefont {C.}~\bibnamefont
  {Cao}}\ and\ \bibinfo {author} {\bibfnamefont {J.}~\bibnamefont {Eisert}},\
  }\href {https://arxiv.org/abs/2505.04705} {\enquote {\bibinfo {title}
  {Measurement-driven quantum advantages in shallow circuits},}\ } (\bibinfo
  {year} {2025}),\ \Eprint {http://arxiv.org/abs/2505.04705} {arXiv:2505.04705
  [quant-ph]} \BibitemShut {NoStop}%
\bibitem [{\citenamefont {Shor}(1995)}]{Shor1995code}%
  \BibitemOpen
  \bibfield  {author} {\bibinfo {author} {\bibfnamefont {P.~W.}\ \bibnamefont
  {Shor}},\ }\href {\doibase 10.1103/PhysRevA.52.R2493} {\bibfield  {journal}
  {\bibinfo  {journal} {Phys. Rev. A}\ }\textbf {\bibinfo {volume} {52}},\
  \bibinfo {pages} {R2493} (\bibinfo {year} {1995})}\BibitemShut {NoStop}%
\bibitem [{\citenamefont {Gottesman}(1997)}]{gottesman1997stabilizer}%
  \BibitemOpen
  \bibfield  {author} {\bibinfo {author} {\bibfnamefont {D.}~\bibnamefont
  {Gottesman}},\ }\href@noop {} {\emph {\bibinfo {title} {Stabilizer codes and
  quantum error correction}}}\ (\bibinfo  {publisher} {California Institute of
  Technology},\ \bibinfo {year} {1997})\BibitemShut {NoStop}%
\bibitem [{\citenamefont {Nielsen}\ and\ \citenamefont
  {Chuang}(2010)}]{nielsen2010quantum}%
  \BibitemOpen
  \bibfield  {author} {\bibinfo {author} {\bibfnamefont {M.~A.}\ \bibnamefont
  {Nielsen}}\ and\ \bibinfo {author} {\bibfnamefont {I.~L.}\ \bibnamefont
  {Chuang}},\ }\href@noop {} {\emph {\bibinfo {title} {Quantum computation and
  quantum information}}}\ (\bibinfo  {publisher} {Cambridge university press},\
  \bibinfo {year} {2010})\BibitemShut {NoStop}%
\bibitem [{\citenamefont {Raussendorf}\ \emph {et~al.}(2003)\citenamefont
  {Raussendorf}, \citenamefont {Browne},\ and\ \citenamefont
  {Briegel}}]{Raussendorf2003MBQC}%
  \BibitemOpen
  \bibfield  {author} {\bibinfo {author} {\bibfnamefont {R.}~\bibnamefont
  {Raussendorf}}, \bibinfo {author} {\bibfnamefont {D.~E.}\ \bibnamefont
  {Browne}}, \ and\ \bibinfo {author} {\bibfnamefont {H.~J.}\ \bibnamefont
  {Briegel}},\ }\href {\doibase 10.1103/PhysRevA.68.022312} {\bibfield
  {journal} {\bibinfo  {journal} {Phys. Rev. A}\ }\textbf {\bibinfo {volume}
  {68}},\ \bibinfo {pages} {022312} (\bibinfo {year} {2003})}\BibitemShut
  {NoStop}%
\bibitem [{\citenamefont {Briegel}\ \emph {et~al.}(2009)\citenamefont
  {Briegel}, \citenamefont {Browne}, \citenamefont {D{\"u}r}, \citenamefont
  {Raussendorf},\ and\ \citenamefont {Van~den Nest}}]{Briegel2009MBQC}%
  \BibitemOpen
  \bibfield  {author} {\bibinfo {author} {\bibfnamefont {H.~J.}\ \bibnamefont
  {Briegel}}, \bibinfo {author} {\bibfnamefont {D.~E.}\ \bibnamefont {Browne}},
  \bibinfo {author} {\bibfnamefont {W.}~\bibnamefont {D{\"u}r}}, \bibinfo
  {author} {\bibfnamefont {R.}~\bibnamefont {Raussendorf}}, \ and\ \bibinfo
  {author} {\bibfnamefont {M.}~\bibnamefont {Van~den Nest}},\ }\href {\doibase
  10.1038/nphys1157} {\bibfield  {journal} {\bibinfo  {journal} {Nature
  Physics}\ }\textbf {\bibinfo {volume} {5}},\ \bibinfo {pages} {19} (\bibinfo
  {year} {2009})}\BibitemShut {NoStop}%
\bibitem [{\citenamefont {Lu}\ \emph {et~al.}(2022)\citenamefont {Lu},
  \citenamefont {Lessa}, \citenamefont {Kim},\ and\ \citenamefont
  {Hsieh}}]{Lu2022LRE}%
  \BibitemOpen
  \bibfield  {author} {\bibinfo {author} {\bibfnamefont {T.-C.}\ \bibnamefont
  {Lu}}, \bibinfo {author} {\bibfnamefont {L.~A.}\ \bibnamefont {Lessa}},
  \bibinfo {author} {\bibfnamefont {I.~H.}\ \bibnamefont {Kim}}, \ and\
  \bibinfo {author} {\bibfnamefont {T.~H.}\ \bibnamefont {Hsieh}},\ }\href
  {\doibase 10.1103/PRXQuantum.3.040337} {\bibfield  {journal} {\bibinfo
  {journal} {PRX Quantum}\ }\textbf {\bibinfo {volume} {3}},\ \bibinfo {pages}
  {040337} (\bibinfo {year} {2022})}\BibitemShut {NoStop}%
\bibitem [{\citenamefont {Zhu}\ \emph {et~al.}(2023)\citenamefont {Zhu},
  \citenamefont {Tantivasadakarn}, \citenamefont {Vishwanath}, \citenamefont
  {Trebst},\ and\ \citenamefont {Verresen}}]{Zhu2023LRE}%
  \BibitemOpen
  \bibfield  {author} {\bibinfo {author} {\bibfnamefont {G.-Y.}\ \bibnamefont
  {Zhu}}, \bibinfo {author} {\bibfnamefont {N.}~\bibnamefont
  {Tantivasadakarn}}, \bibinfo {author} {\bibfnamefont {A.}~\bibnamefont
  {Vishwanath}}, \bibinfo {author} {\bibfnamefont {S.}~\bibnamefont {Trebst}},
  \ and\ \bibinfo {author} {\bibfnamefont {R.}~\bibnamefont {Verresen}},\
  }\href {\doibase 10.1103/PhysRevLett.131.200201} {\bibfield  {journal}
  {\bibinfo  {journal} {Phys. Rev. Lett.}\ }\textbf {\bibinfo {volume} {131}},\
  \bibinfo {pages} {200201} (\bibinfo {year} {2023})}\BibitemShut {NoStop}%
\bibitem [{\citenamefont {Tantivasadakarn}\ \emph {et~al.}(2023)\citenamefont
  {Tantivasadakarn}, \citenamefont {Vishwanath},\ and\ \citenamefont
  {Verresen}}]{Tantivasadakarn2023Finite}%
  \BibitemOpen
  \bibfield  {author} {\bibinfo {author} {\bibfnamefont {N.}~\bibnamefont
  {Tantivasadakarn}}, \bibinfo {author} {\bibfnamefont {A.}~\bibnamefont
  {Vishwanath}}, \ and\ \bibinfo {author} {\bibfnamefont {R.}~\bibnamefont
  {Verresen}},\ }\href {\doibase 10.1103/PRXQuantum.4.020339} {\bibfield
  {journal} {\bibinfo  {journal} {PRX Quantum}\ }\textbf {\bibinfo {volume}
  {4}},\ \bibinfo {pages} {020339} (\bibinfo {year} {2023})}\BibitemShut
  {NoStop}%
\bibitem [{\citenamefont {Iqbal}\ \emph {et~al.}(2024)\citenamefont {Iqbal},
  \citenamefont {Tantivasadakarn}, \citenamefont {Gatterman}, \citenamefont
  {Gerber}, \citenamefont {Gilmore}, \citenamefont {Gresh}, \citenamefont
  {Hankin}, \citenamefont {Hewitt}, \citenamefont {Horst}, \citenamefont
  {Matheny}, \citenamefont {Mengle}, \citenamefont {Neyenhuis}, \citenamefont
  {Vishwanath}, \citenamefont {Foss-Feig}, \citenamefont {Verresen},\ and\
  \citenamefont {Dreyer}}]{Iqbal2024Topological}%
  \BibitemOpen
  \bibfield  {author} {\bibinfo {author} {\bibfnamefont {M.}~\bibnamefont
  {Iqbal}}, \bibinfo {author} {\bibfnamefont {N.}~\bibnamefont
  {Tantivasadakarn}}, \bibinfo {author} {\bibfnamefont {T.~M.}\ \bibnamefont
  {Gatterman}}, \bibinfo {author} {\bibfnamefont {J.~A.}\ \bibnamefont
  {Gerber}}, \bibinfo {author} {\bibfnamefont {K.}~\bibnamefont {Gilmore}},
  \bibinfo {author} {\bibfnamefont {D.}~\bibnamefont {Gresh}}, \bibinfo
  {author} {\bibfnamefont {A.}~\bibnamefont {Hankin}}, \bibinfo {author}
  {\bibfnamefont {N.}~\bibnamefont {Hewitt}}, \bibinfo {author} {\bibfnamefont
  {C.~V.}\ \bibnamefont {Horst}}, \bibinfo {author} {\bibfnamefont
  {M.}~\bibnamefont {Matheny}}, \bibinfo {author} {\bibfnamefont
  {T.}~\bibnamefont {Mengle}}, \bibinfo {author} {\bibfnamefont
  {B.}~\bibnamefont {Neyenhuis}}, \bibinfo {author} {\bibfnamefont
  {A.}~\bibnamefont {Vishwanath}}, \bibinfo {author} {\bibfnamefont
  {M.}~\bibnamefont {Foss-Feig}}, \bibinfo {author} {\bibfnamefont
  {R.}~\bibnamefont {Verresen}}, \ and\ \bibinfo {author} {\bibfnamefont
  {H.}~\bibnamefont {Dreyer}},\ }\href {\doibase 10.1038/s42005-024-01698-3}
  {\bibfield  {journal} {\bibinfo  {journal} {Communications Physics}\ }\textbf
  {\bibinfo {volume} {7}},\ \bibinfo {pages} {205} (\bibinfo {year}
  {2024})}\BibitemShut {NoStop}%
\bibitem [{\citenamefont {Chen}\ \emph {et~al.}(2010)\citenamefont {Chen},
  \citenamefont {Gu},\ and\ \citenamefont {Wen}}]{Chen2010topological}%
  \BibitemOpen
  \bibfield  {author} {\bibinfo {author} {\bibfnamefont {X.}~\bibnamefont
  {Chen}}, \bibinfo {author} {\bibfnamefont {Z.-C.}\ \bibnamefont {Gu}}, \ and\
  \bibinfo {author} {\bibfnamefont {X.-G.}\ \bibnamefont {Wen}},\ }\href
  {\doibase 10.1103/PhysRevB.82.155138} {\bibfield  {journal} {\bibinfo
  {journal} {Phys. Rev. B}\ }\textbf {\bibinfo {volume} {82}},\ \bibinfo
  {pages} {155138} (\bibinfo {year} {2010})}\BibitemShut {NoStop}%
\bibitem [{\citenamefont {Bravyi}\ \emph {et~al.}(2025)\citenamefont {Bravyi},
  \citenamefont {Lee}, \citenamefont {Li},\ and\ \citenamefont
  {Yoshida}}]{bravyi2024entanglement}%
  \BibitemOpen
  \bibfield  {author} {\bibinfo {author} {\bibfnamefont {S.}~\bibnamefont
  {Bravyi}}, \bibinfo {author} {\bibfnamefont {D.}~\bibnamefont {Lee}},
  \bibinfo {author} {\bibfnamefont {Z.}~\bibnamefont {Li}}, \ and\ \bibinfo
  {author} {\bibfnamefont {B.}~\bibnamefont {Yoshida}},\ }\href {\doibase
  10.1103/PhysRevLett.134.210602} {\bibfield  {journal} {\bibinfo  {journal}
  {Phys. Rev. Lett.}\ }\textbf {\bibinfo {volume} {134}},\ \bibinfo {pages}
  {210602} (\bibinfo {year} {2025})}\BibitemShut {NoStop}%
\bibitem [{\citenamefont {Friedman}\ \emph {et~al.}(2023)\citenamefont
  {Friedman}, \citenamefont {Yin}, \citenamefont {Hong},\ and\ \citenamefont
  {Lucas}}]{friedman2023locality}%
  \BibitemOpen
  \bibfield  {author} {\bibinfo {author} {\bibfnamefont {A.~J.}\ \bibnamefont
  {Friedman}}, \bibinfo {author} {\bibfnamefont {C.}~\bibnamefont {Yin}},
  \bibinfo {author} {\bibfnamefont {Y.}~\bibnamefont {Hong}}, \ and\ \bibinfo
  {author} {\bibfnamefont {A.}~\bibnamefont {Lucas}},\ }\href@noop {} {\enquote
  {\bibinfo {title} {Locality and error correction in quantum dynamics with
  measurement},}\ } (\bibinfo {year} {2023}),\ \Eprint
  {http://arxiv.org/abs/2206.09929} {arXiv:2206.09929 [quant-ph]} \BibitemShut
  {NoStop}%
\bibitem [{\citenamefont {Brand\~ao}\ \emph {et~al.}(2021)\citenamefont
  {Brand\~ao}, \citenamefont {Chemissany}, \citenamefont {Hunter-Jones},
  \citenamefont {Kueng},\ and\ \citenamefont
  {Preskill}}]{Brandao2021Complexity}%
  \BibitemOpen
  \bibfield  {author} {\bibinfo {author} {\bibfnamefont {F.~G.}\ \bibnamefont
  {Brand\~ao}}, \bibinfo {author} {\bibfnamefont {W.}~\bibnamefont
  {Chemissany}}, \bibinfo {author} {\bibfnamefont {N.}~\bibnamefont
  {Hunter-Jones}}, \bibinfo {author} {\bibfnamefont {R.}~\bibnamefont {Kueng}},
  \ and\ \bibinfo {author} {\bibfnamefont {J.}~\bibnamefont {Preskill}},\
  }\href {\doibase 10.1103/PRXQuantum.2.030316} {\bibfield  {journal} {\bibinfo
   {journal} {PRX Quantum}\ }\textbf {\bibinfo {volume} {2}},\ \bibinfo {pages}
  {030316} (\bibinfo {year} {2021})}\BibitemShut {NoStop}%
\bibitem [{\citenamefont {Roberts}\ and\ \citenamefont
  {Yoshida}(2017)}]{Daniel2017complexity}%
  \BibitemOpen
  \bibfield  {author} {\bibinfo {author} {\bibfnamefont {D.~A.}\ \bibnamefont
  {Roberts}}\ and\ \bibinfo {author} {\bibfnamefont {B.}~\bibnamefont
  {Yoshida}},\ }\href {\doibase 10.1007/JHEP04(2017)121} {\bibfield  {journal}
  {\bibinfo  {journal} {Journal of High Energy Physics}\ }\textbf {\bibinfo
  {volume} {2017}},\ \bibinfo {pages} {121} (\bibinfo {year}
  {2017})}\BibitemShut {NoStop}%
\bibitem [{\citenamefont {Liu}\ \emph {et~al.}(2018)\citenamefont {Liu},
  \citenamefont {Lloyd}, \citenamefont {Zhu},\ and\ \citenamefont
  {Zhu}}]{Liu2018complexity}%
  \BibitemOpen
  \bibfield  {author} {\bibinfo {author} {\bibfnamefont {Z.-W.}\ \bibnamefont
  {Liu}}, \bibinfo {author} {\bibfnamefont {S.}~\bibnamefont {Lloyd}}, \bibinfo
  {author} {\bibfnamefont {E.}~\bibnamefont {Zhu}}, \ and\ \bibinfo {author}
  {\bibfnamefont {H.}~\bibnamefont {Zhu}},\ }\href {\doibase
  10.1007/JHEP07(2018)041} {\bibfield  {journal} {\bibinfo  {journal} {Journal
  of High Energy Physics}\ }\textbf {\bibinfo {volume} {2018}},\ \bibinfo
  {pages} {41} (\bibinfo {year} {2018})}\BibitemShut {NoStop}%
\bibitem [{\citenamefont {Anshu}\ and\ \citenamefont
  {Nirkhe}(2022)}]{anshu2022Bounds}%
  \BibitemOpen
  \bibfield  {author} {\bibinfo {author} {\bibfnamefont {A.}~\bibnamefont
  {Anshu}}\ and\ \bibinfo {author} {\bibfnamefont {C.}~\bibnamefont {Nirkhe}},\
  }in\ \href {\doibase 10.4230/LIPIcs.ITCS.2022.6} {\emph {\bibinfo {booktitle}
  {13th Innovations in Theoretical Computer Science Conference (ITCS 2022)}}},\
  \bibinfo {series} {Leibniz International Proceedings in Informatics
  (LIPIcs)}, Vol.\ \bibinfo {volume} {215},\ \bibinfo {editor} {edited by\
  \bibinfo {editor} {\bibfnamefont {M.}~\bibnamefont {Braverman}}}\ (\bibinfo
  {publisher} {Schloss Dagstuhl -- Leibniz-Zentrum f{\"u}r Informatik},\
  \bibinfo {address} {Dagstuhl, Germany},\ \bibinfo {year} {2022})\ pp.\
  \bibinfo {pages} {6:1--6:22}\BibitemShut {NoStop}%
\bibitem [{\citenamefont {Baiguera}\ \emph {et~al.}(2025)\citenamefont
  {Baiguera}, \citenamefont {Balasubramanian}, \citenamefont {Caputa},
  \citenamefont {Chapman}, \citenamefont {Haferkamp}, \citenamefont {Heller},\
  and\ \citenamefont {Halpern}}]{baiguera2025quantumcomplexitygravityquantum}%
  \BibitemOpen
  \bibfield  {author} {\bibinfo {author} {\bibfnamefont {S.}~\bibnamefont
  {Baiguera}}, \bibinfo {author} {\bibfnamefont {V.}~\bibnamefont
  {Balasubramanian}}, \bibinfo {author} {\bibfnamefont {P.}~\bibnamefont
  {Caputa}}, \bibinfo {author} {\bibfnamefont {S.}~\bibnamefont {Chapman}},
  \bibinfo {author} {\bibfnamefont {J.}~\bibnamefont {Haferkamp}}, \bibinfo
  {author} {\bibfnamefont {M.~P.}\ \bibnamefont {Heller}}, \ and\ \bibinfo
  {author} {\bibfnamefont {N.~Y.}\ \bibnamefont {Halpern}},\ }\href
  {https://arxiv.org/abs/2503.10753} {\enquote {\bibinfo {title} {Quantum
  complexity in gravity, quantum field theory, and quantum information
  science},}\ } (\bibinfo {year} {2025}),\ \Eprint
  {http://arxiv.org/abs/2503.10753} {arXiv:2503.10753 [hep-th]} \BibitemShut
  {NoStop}%
\bibitem [{\citenamefont {Wen}(2013)}]{Wen2013LRE}%
  \BibitemOpen
  \bibfield  {author} {\bibinfo {author} {\bibfnamefont {X.-G.}\ \bibnamefont
  {Wen}},\ }\href {\doibase https://doi.org/10.1155/2013/198710} {\bibfield
  {journal} {\bibinfo  {journal} {International Scholarly Research Notices}\
  }\textbf {\bibinfo {volume} {2013}},\ \bibinfo {pages} {198710} (\bibinfo
  {year} {2013})},\ \Eprint
  {http://arxiv.org/abs/https://onlinelibrary.wiley.com/doi/pdf/10.1155/2013/198710}
  {https://onlinelibrary.wiley.com/doi/pdf/10.1155/2013/198710} \BibitemShut
  {NoStop}%
\bibitem [{\citenamefont {Leone}\ \emph {et~al.}(2021)\citenamefont {Leone},
  \citenamefont {Oliviero}, \citenamefont {Zhou},\ and\ \citenamefont
  {Hamma}}]{Leone2021quantumchaos}%
  \BibitemOpen
  \bibfield  {author} {\bibinfo {author} {\bibfnamefont {L.}~\bibnamefont
  {Leone}}, \bibinfo {author} {\bibfnamefont {S.~F.~E.}\ \bibnamefont
  {Oliviero}}, \bibinfo {author} {\bibfnamefont {Y.}~\bibnamefont {Zhou}}, \
  and\ \bibinfo {author} {\bibfnamefont {A.}~\bibnamefont {Hamma}},\ }\href
  {\doibase 10.22331/q-2021-05-04-453} {\bibfield  {journal} {\bibinfo
  {journal} {{Quantum}}\ }\textbf {\bibinfo {volume} {5}},\ \bibinfo {pages}
  {453} (\bibinfo {year} {2021})}\BibitemShut {NoStop}%
\bibitem [{\citenamefont {Kitaev}(2003)}]{Kitaev2003anyons}%
  \BibitemOpen
  \bibfield  {author} {\bibinfo {author} {\bibfnamefont {A.}~\bibnamefont
  {Kitaev}},\ }\href {\doibase https://doi.org/10.1016/S0003-4916(02)00018-0}
  {\bibfield  {journal} {\bibinfo  {journal} {Annals of Physics}\ }\textbf
  {\bibinfo {volume} {303}},\ \bibinfo {pages} {2} (\bibinfo {year}
  {2003})}\BibitemShut {NoStop}%
\bibitem [{\citenamefont {Yi}\ \emph {et~al.}(2024)\citenamefont {Yi},
  \citenamefont {Ye}, \citenamefont {Gottesman},\ and\ \citenamefont
  {Liu}}]{Yi2024AQEC}%
  \BibitemOpen
  \bibfield  {author} {\bibinfo {author} {\bibfnamefont {J.}~\bibnamefont
  {Yi}}, \bibinfo {author} {\bibfnamefont {W.}~\bibnamefont {Ye}}, \bibinfo
  {author} {\bibfnamefont {D.}~\bibnamefont {Gottesman}}, \ and\ \bibinfo
  {author} {\bibfnamefont {Z.-W.}\ \bibnamefont {Liu}},\ }\href {\doibase
  10.1038/s41567-024-02621-x} {\bibfield  {journal} {\bibinfo  {journal}
  {Nature Physics}\ } (\bibinfo {year} {2024}),\
  10.1038/s41567-024-02621-x}\BibitemShut {NoStop}%
\bibitem [{\citenamefont {Sachdev}(1999)}]{Sachdev1999phase}%
  \BibitemOpen
  \bibfield  {author} {\bibinfo {author} {\bibfnamefont {S.}~\bibnamefont
  {Sachdev}},\ }\href {\doibase 10.1088/2058-7058/12/4/23} {\bibfield
  {journal} {\bibinfo  {journal} {Physics World}\ }\textbf {\bibinfo {volume}
  {12}},\ \bibinfo {pages} {33} (\bibinfo {year} {1999})}\BibitemShut {NoStop}%
\bibitem [{\citenamefont {Zhang}\ \emph {et~al.}(2021)\citenamefont {Zhang},
  \citenamefont {Tang}, \citenamefont {Zhou},\ and\ \citenamefont
  {Ma}}]{Zhang2021generalized}%
  \BibitemOpen
  \bibfield  {author} {\bibinfo {author} {\bibfnamefont {Y.}~\bibnamefont
  {Zhang}}, \bibinfo {author} {\bibfnamefont {Y.}~\bibnamefont {Tang}},
  \bibinfo {author} {\bibfnamefont {Y.}~\bibnamefont {Zhou}}, \ and\ \bibinfo
  {author} {\bibfnamefont {X.}~\bibnamefont {Ma}},\ }\href {\doibase
  10.1103/PhysRevA.103.052426} {\bibfield  {journal} {\bibinfo  {journal}
  {Phys. Rev. A}\ }\textbf {\bibinfo {volume} {103}},\ \bibinfo {pages}
  {052426} (\bibinfo {year} {2021})}\BibitemShut {NoStop}%
\bibitem [{\citenamefont {Plenio}(2007)}]{Martin2007generalized}%
  \BibitemOpen
  \bibfield  {author} {\bibinfo {author} {\bibfnamefont {M.~B.}\ \bibnamefont
  {Plenio}},\ }\href {\doibase 10.1080/09500340701275774} {\bibfield  {journal}
  {\bibinfo  {journal} {Journal of Modern Optics}\ }\textbf {\bibinfo {volume}
  {54}},\ \bibinfo {pages} {2193} (\bibinfo {year} {2007})},\ \Eprint
  {http://arxiv.org/abs/https://doi.org/10.1080/09500340701275774}
  {https://doi.org/10.1080/09500340701275774} \BibitemShut {NoStop}%
\bibitem [{\citenamefont {Rossi}\ \emph {et~al.}(2013)\citenamefont {Rossi},
  \citenamefont {Huber}, \citenamefont {Bruß},\ and\ \citenamefont
  {Macchiavello}}]{Rossi2013Hyper}%
  \BibitemOpen
  \bibfield  {author} {\bibinfo {author} {\bibfnamefont {M.}~\bibnamefont
  {Rossi}}, \bibinfo {author} {\bibfnamefont {M.}~\bibnamefont {Huber}},
  \bibinfo {author} {\bibfnamefont {D.}~\bibnamefont {Bruß}}, \ and\ \bibinfo
  {author} {\bibfnamefont {C.}~\bibnamefont {Macchiavello}},\ }\href {\doibase
  10.1088/1367-2630/15/11/113022} {\bibfield  {journal} {\bibinfo  {journal}
  {New Journal of Physics}\ }\textbf {\bibinfo {volume} {15}},\ \bibinfo
  {pages} {113022} (\bibinfo {year} {2013})}\BibitemShut {NoStop}%
\bibitem [{\citenamefont {Qu}\ \emph {et~al.}(2013)\citenamefont {Qu},
  \citenamefont {Wang}, \citenamefont {Li},\ and\ \citenamefont
  {Bao}}]{Qu2013hypergraph}%
  \BibitemOpen
  \bibfield  {author} {\bibinfo {author} {\bibfnamefont {R.}~\bibnamefont
  {Qu}}, \bibinfo {author} {\bibfnamefont {J.}~\bibnamefont {Wang}}, \bibinfo
  {author} {\bibfnamefont {Z.-s.}\ \bibnamefont {Li}}, \ and\ \bibinfo {author}
  {\bibfnamefont {Y.-r.}\ \bibnamefont {Bao}},\ }\href {\doibase
  10.1103/PhysRevA.87.022311} {\bibfield  {journal} {\bibinfo  {journal} {Phys.
  Rev. A}\ }\textbf {\bibinfo {volume} {87}},\ \bibinfo {pages} {022311}
  (\bibinfo {year} {2013})}\BibitemShut {NoStop}%
\bibitem [{\citenamefont {Chen}\ \emph {et~al.}(2024)\citenamefont {Chen},
  \citenamefont {Yan},\ and\ \citenamefont {Zhou}}]{Chen2024magicofquantum}%
  \BibitemOpen
  \bibfield  {author} {\bibinfo {author} {\bibfnamefont {J.}~\bibnamefont
  {Chen}}, \bibinfo {author} {\bibfnamefont {Y.}~\bibnamefont {Yan}}, \ and\
  \bibinfo {author} {\bibfnamefont {Y.}~\bibnamefont {Zhou}},\ }\href {\doibase
  10.22331/q-2024-05-21-1351} {\bibfield  {journal} {\bibinfo  {journal}
  {{Quantum}}\ }\textbf {\bibinfo {volume} {8}},\ \bibinfo {pages} {1351}
  (\bibinfo {year} {2024})}\BibitemShut {NoStop}%
\bibitem [{\citenamefont {Wei}\ and\ \citenamefont
  {Goldbart}(2003)}]{Wei2003Geometric}%
  \BibitemOpen
  \bibfield  {author} {\bibinfo {author} {\bibfnamefont {T.-C.}\ \bibnamefont
  {Wei}}\ and\ \bibinfo {author} {\bibfnamefont {P.~M.}\ \bibnamefont
  {Goldbart}},\ }\href {\doibase 10.1103/PhysRevA.68.042307} {\bibfield
  {journal} {\bibinfo  {journal} {Phys. Rev. A}\ }\textbf {\bibinfo {volume}
  {68}},\ \bibinfo {pages} {042307} (\bibinfo {year} {2003})}\BibitemShut
  {NoStop}%
\bibitem [{\citenamefont {Orús}\ \emph {et~al.}(2014)\citenamefont {Orús},
  \citenamefont {Wei}, \citenamefont {Buerschaper},\ and\ \citenamefont
  {Nest}}]{Orus2014Geometric}%
  \BibitemOpen
  \bibfield  {author} {\bibinfo {author} {\bibfnamefont {R.}~\bibnamefont
  {Orús}}, \bibinfo {author} {\bibfnamefont {T.-C.}\ \bibnamefont {Wei}},
  \bibinfo {author} {\bibfnamefont {O.}~\bibnamefont {Buerschaper}}, \ and\
  \bibinfo {author} {\bibfnamefont {M.~V.~d.}\ \bibnamefont {Nest}},\ }\href
  {\doibase 10.1088/1367-2630/16/1/013015} {\bibfield  {journal} {\bibinfo
  {journal} {New Journal of Physics}\ }\textbf {\bibinfo {volume} {16}},\
  \bibinfo {pages} {013015} (\bibinfo {year} {2014})}\BibitemShut {NoStop}%
\bibitem [{\citenamefont {Panteleev}\ and\ \citenamefont
  {Kalachev}(2022)}]{Panteleev2022goodQLDPC}%
  \BibitemOpen
  \bibfield  {author} {\bibinfo {author} {\bibfnamefont {P.}~\bibnamefont
  {Panteleev}}\ and\ \bibinfo {author} {\bibfnamefont {G.}~\bibnamefont
  {Kalachev}},\ }in\ \href {\doibase 10.1145/3519935.3520017} {\emph {\bibinfo
  {booktitle} {Proceedings of the 54th Annual ACM SIGACT Symposium on Theory of
  Computing}}},\ \bibinfo {series and number} {STOC 2022}\ (\bibinfo
  {publisher} {Association for Computing Machinery},\ \bibinfo {address} {New
  York, NY, USA},\ \bibinfo {year} {2022})\ p.\ \bibinfo {pages}
  {375–388}\BibitemShut {NoStop}%
\bibitem [{\citenamefont {Gottesman}(1998)}]{gottesman1998heisenberg}%
  \BibitemOpen
  \bibfield  {author} {\bibinfo {author} {\bibfnamefont {D.}~\bibnamefont
  {Gottesman}},\ }\href@noop {} {\enquote {\bibinfo {title} {The heisenberg
  representation of quantum computers},}\ } (\bibinfo {year} {1998}),\ \Eprint
  {http://arxiv.org/abs/quant-ph/9807006} {arXiv:quant-ph/9807006 [quant-ph]}
  \BibitemShut {NoStop}%
\bibitem [{\citenamefont {Aaronson}\ and\ \citenamefont
  {Gottesman}(2004)}]{Aaronson2004stabilizer}%
  \BibitemOpen
  \bibfield  {author} {\bibinfo {author} {\bibfnamefont {S.}~\bibnamefont
  {Aaronson}}\ and\ \bibinfo {author} {\bibfnamefont {D.}~\bibnamefont
  {Gottesman}},\ }\href {\doibase 10.1103/PhysRevA.70.052328} {\bibfield
  {journal} {\bibinfo  {journal} {Phys. Rev. A}\ }\textbf {\bibinfo {volume}
  {70}},\ \bibinfo {pages} {052328} (\bibinfo {year} {2004})}\BibitemShut
  {NoStop}%
\bibitem [{\citenamefont {Knill}\ and\ \citenamefont
  {Laflamme}(1997)}]{Knill1997KLcondition}%
  \BibitemOpen
  \bibfield  {author} {\bibinfo {author} {\bibfnamefont {E.}~\bibnamefont
  {Knill}}\ and\ \bibinfo {author} {\bibfnamefont {R.}~\bibnamefont
  {Laflamme}},\ }\href {\doibase 10.1103/PhysRevA.55.900} {\bibfield  {journal}
  {\bibinfo  {journal} {Phys. Rev. A}\ }\textbf {\bibinfo {volume} {55}},\
  \bibinfo {pages} {900} (\bibinfo {year} {1997})}\BibitemShut {NoStop}%
\bibitem [{\citenamefont {Leverrier}\ and\ \citenamefont
  {Zémor}(2022)}]{leverrier2022quantumtannercodes}%
  \BibitemOpen
  \bibfield  {author} {\bibinfo {author} {\bibfnamefont {A.}~\bibnamefont
  {Leverrier}}\ and\ \bibinfo {author} {\bibfnamefont {G.}~\bibnamefont
  {Zémor}},\ }\href {https://arxiv.org/abs/2202.13641} {\enquote {\bibinfo
  {title} {Quantum tanner codes},}\ } (\bibinfo {year} {2022}),\ \Eprint
  {http://arxiv.org/abs/2202.13641} {arXiv:2202.13641 [quant-ph]} \BibitemShut
  {NoStop}%
\bibitem [{\citenamefont {Dinur}\ \emph {et~al.}(2022)\citenamefont {Dinur},
  \citenamefont {Hsieh}, \citenamefont {Lin},\ and\ \citenamefont
  {Vidick}}]{dinur2022goodquantumldpccodes}%
  \BibitemOpen
  \bibfield  {author} {\bibinfo {author} {\bibfnamefont {I.}~\bibnamefont
  {Dinur}}, \bibinfo {author} {\bibfnamefont {M.-H.}\ \bibnamefont {Hsieh}},
  \bibinfo {author} {\bibfnamefont {T.-C.}\ \bibnamefont {Lin}}, \ and\
  \bibinfo {author} {\bibfnamefont {T.}~\bibnamefont {Vidick}},\ }\href
  {https://arxiv.org/abs/2206.07750} {\enquote {\bibinfo {title} {Good quantum
  ldpc codes with linear time decoders},}\ } (\bibinfo {year} {2022}),\ \Eprint
  {http://arxiv.org/abs/2206.07750} {arXiv:2206.07750 [quant-ph]} \BibitemShut
  {NoStop}%
\bibitem [{\citenamefont {Ji}\ \emph {et~al.}(2007)\citenamefont {Ji},
  \citenamefont {Chen}, \citenamefont {Wei},\ and\ \citenamefont
  {Ying}}]{ji2007LULC}%
  \BibitemOpen
  \bibfield  {author} {\bibinfo {author} {\bibfnamefont {Z.}~\bibnamefont
  {Ji}}, \bibinfo {author} {\bibfnamefont {J.}~\bibnamefont {Chen}}, \bibinfo
  {author} {\bibfnamefont {Z.}~\bibnamefont {Wei}}, \ and\ \bibinfo {author}
  {\bibfnamefont {M.}~\bibnamefont {Ying}},\ }\href {\doibase
  10.26421/QIC10.1-2-8} {\bibfield  {journal} {\bibinfo  {journal} {Quantum
  Information and Computation}\ }\textbf {\bibinfo {volume} {10}} (\bibinfo
  {year} {2007}),\ 10.26421/QIC10.1-2-8}\BibitemShut {NoStop}%
\bibitem [{\citenamefont {Leibfried}\ \emph {et~al.}(2004)\citenamefont
  {Leibfried}, \citenamefont {Barrett}, \citenamefont {Schaetz}, \citenamefont
  {Britton}, \citenamefont {Chiaverini}, \citenamefont {Itano}, \citenamefont
  {Jost}, \citenamefont {Langer},\ and\ \citenamefont
  {Wineland}}]{leibfried2004toward}%
  \BibitemOpen
  \bibfield  {author} {\bibinfo {author} {\bibfnamefont {D.}~\bibnamefont
  {Leibfried}}, \bibinfo {author} {\bibfnamefont {M.~D.}\ \bibnamefont
  {Barrett}}, \bibinfo {author} {\bibfnamefont {T.}~\bibnamefont {Schaetz}},
  \bibinfo {author} {\bibfnamefont {J.}~\bibnamefont {Britton}}, \bibinfo
  {author} {\bibfnamefont {J.}~\bibnamefont {Chiaverini}}, \bibinfo {author}
  {\bibfnamefont {W.~M.}\ \bibnamefont {Itano}}, \bibinfo {author}
  {\bibfnamefont {J.~D.}\ \bibnamefont {Jost}}, \bibinfo {author}
  {\bibfnamefont {C.}~\bibnamefont {Langer}}, \ and\ \bibinfo {author}
  {\bibfnamefont {D.~J.}\ \bibnamefont {Wineland}},\ }\href {\doibase
  10.1126/science.1097576} {\bibfield  {journal} {\bibinfo  {journal}
  {Science}\ }\textbf {\bibinfo {volume} {304}},\ \bibinfo {pages} {1476}
  (\bibinfo {year} {2004})},\ \Eprint
  {http://arxiv.org/abs/https://www.science.org/doi/pdf/10.1126/science.1097576}
  {https://www.science.org/doi/pdf/10.1126/science.1097576} \BibitemShut
  {NoStop}%
\bibitem [{\citenamefont {Lücke}\ \emph {et~al.}(2011)\citenamefont {Lücke},
  \citenamefont {Scherer}, \citenamefont {Kruse}, \citenamefont {Pezzé},
  \citenamefont {Deuretzbacher}, \citenamefont {Hyllus}, \citenamefont {Topic},
  \citenamefont {Peise}, \citenamefont {Ertmer}, \citenamefont {Arlt},
  \citenamefont {Santos}, \citenamefont {Smerzi},\ and\ \citenamefont
  {Klempt}}]{lucke2011twin}%
  \BibitemOpen
  \bibfield  {author} {\bibinfo {author} {\bibfnamefont {B.}~\bibnamefont
  {Lücke}}, \bibinfo {author} {\bibfnamefont {M.}~\bibnamefont {Scherer}},
  \bibinfo {author} {\bibfnamefont {J.}~\bibnamefont {Kruse}}, \bibinfo
  {author} {\bibfnamefont {L.}~\bibnamefont {Pezzé}}, \bibinfo {author}
  {\bibfnamefont {F.}~\bibnamefont {Deuretzbacher}}, \bibinfo {author}
  {\bibfnamefont {P.}~\bibnamefont {Hyllus}}, \bibinfo {author} {\bibfnamefont
  {O.}~\bibnamefont {Topic}}, \bibinfo {author} {\bibfnamefont
  {J.}~\bibnamefont {Peise}}, \bibinfo {author} {\bibfnamefont
  {W.}~\bibnamefont {Ertmer}}, \bibinfo {author} {\bibfnamefont
  {J.}~\bibnamefont {Arlt}}, \bibinfo {author} {\bibfnamefont {L.}~\bibnamefont
  {Santos}}, \bibinfo {author} {\bibfnamefont {A.}~\bibnamefont {Smerzi}}, \
  and\ \bibinfo {author} {\bibfnamefont {C.}~\bibnamefont {Klempt}},\ }\href
  {\doibase 10.1126/science.1208798} {\bibfield  {journal} {\bibinfo  {journal}
  {Science}\ }\textbf {\bibinfo {volume} {334}},\ \bibinfo {pages} {773}
  (\bibinfo {year} {2011})},\ \Eprint
  {http://arxiv.org/abs/https://www.science.org/doi/pdf/10.1126/science.1208798}
  {https://www.science.org/doi/pdf/10.1126/science.1208798} \BibitemShut
  {NoStop}%
\bibitem [{\citenamefont {Bose}\ \emph {et~al.}(1998)\citenamefont {Bose},
  \citenamefont {Vedral},\ and\ \citenamefont
  {Knight}}]{bose1998multiparticle}%
  \BibitemOpen
  \bibfield  {author} {\bibinfo {author} {\bibfnamefont {S.}~\bibnamefont
  {Bose}}, \bibinfo {author} {\bibfnamefont {V.}~\bibnamefont {Vedral}}, \ and\
  \bibinfo {author} {\bibfnamefont {P.~L.}\ \bibnamefont {Knight}},\ }\href
  {\doibase 10.1103/PhysRevA.57.822} {\bibfield  {journal} {\bibinfo  {journal}
  {Phys. Rev. A}\ }\textbf {\bibinfo {volume} {57}},\ \bibinfo {pages} {822}
  (\bibinfo {year} {1998})}\BibitemShut {NoStop}%
\bibitem [{\citenamefont {Cleve}\ \emph {et~al.}(1999)\citenamefont {Cleve},
  \citenamefont {Gottesman},\ and\ \citenamefont {Lo}}]{cleve1999how}%
  \BibitemOpen
  \bibfield  {author} {\bibinfo {author} {\bibfnamefont {R.}~\bibnamefont
  {Cleve}}, \bibinfo {author} {\bibfnamefont {D.}~\bibnamefont {Gottesman}}, \
  and\ \bibinfo {author} {\bibfnamefont {H.-K.}\ \bibnamefont {Lo}},\ }\href
  {\doibase 10.1103/PhysRevLett.83.648} {\bibfield  {journal} {\bibinfo
  {journal} {Phys. Rev. Lett.}\ }\textbf {\bibinfo {volume} {83}},\ \bibinfo
  {pages} {648} (\bibinfo {year} {1999})}\BibitemShut {NoStop}%
\bibitem [{\citenamefont {Ouyang}(2014)}]{ouyang2014permutation}%
  \BibitemOpen
  \bibfield  {author} {\bibinfo {author} {\bibfnamefont {Y.}~\bibnamefont
  {Ouyang}},\ }\href {\doibase 10.1103/PhysRevA.90.062317} {\bibfield
  {journal} {\bibinfo  {journal} {Phys. Rev. A}\ }\textbf {\bibinfo {volume}
  {90}},\ \bibinfo {pages} {062317} (\bibinfo {year} {2014})}\BibitemShut
  {NoStop}%
\bibitem [{\citenamefont {Cruz}\ \emph {et~al.}(2019)\citenamefont {Cruz},
  \citenamefont {Fournier}, \citenamefont {Gremion}, \citenamefont {Jeannerot},
  \citenamefont {Komagata}, \citenamefont {Tosic}, \citenamefont
  {Thiesbrummel}, \citenamefont {Chan}, \citenamefont {Macris}, \citenamefont
  {Dupertuis},\ and\ \citenamefont {Javerzac-Galy}}]{cruz2019efficient}%
  \BibitemOpen
  \bibfield  {author} {\bibinfo {author} {\bibfnamefont {D.}~\bibnamefont
  {Cruz}}, \bibinfo {author} {\bibfnamefont {R.}~\bibnamefont {Fournier}},
  \bibinfo {author} {\bibfnamefont {F.}~\bibnamefont {Gremion}}, \bibinfo
  {author} {\bibfnamefont {A.}~\bibnamefont {Jeannerot}}, \bibinfo {author}
  {\bibfnamefont {K.}~\bibnamefont {Komagata}}, \bibinfo {author}
  {\bibfnamefont {T.}~\bibnamefont {Tosic}}, \bibinfo {author} {\bibfnamefont
  {J.}~\bibnamefont {Thiesbrummel}}, \bibinfo {author} {\bibfnamefont {C.~L.}\
  \bibnamefont {Chan}}, \bibinfo {author} {\bibfnamefont {N.}~\bibnamefont
  {Macris}}, \bibinfo {author} {\bibfnamefont {M.-A.}\ \bibnamefont
  {Dupertuis}}, \ and\ \bibinfo {author} {\bibfnamefont {C.}~\bibnamefont
  {Javerzac-Galy}},\ }\href {\doibase https://doi.org/10.1002/qute.201900015}
  {\bibfield  {journal} {\bibinfo  {journal} {Advanced Quantum Technologies}\
  }\textbf {\bibinfo {volume} {2}},\ \bibinfo {pages} {1900015} (\bibinfo
  {year} {2019})},\ \Eprint
  {http://arxiv.org/abs/https://advanced.onlinelibrary.wiley.com/doi/pdf/10.1002/qute.201900015}
  {https://advanced.onlinelibrary.wiley.com/doi/pdf/10.1002/qute.201900015}
  \BibitemShut {NoStop}%
\bibitem [{\citenamefont {Claudon}\ \emph {et~al.}(2024)\citenamefont
  {Claudon}, \citenamefont {Zylberman}, \citenamefont {Feniou}, \citenamefont
  {Debbasch}, \citenamefont {Peruzzo},\ and\ \citenamefont
  {Piquemal}}]{Claudon2024polylogarithmicdepth}%
  \BibitemOpen
  \bibfield  {author} {\bibinfo {author} {\bibfnamefont {B.}~\bibnamefont
  {Claudon}}, \bibinfo {author} {\bibfnamefont {J.}~\bibnamefont {Zylberman}},
  \bibinfo {author} {\bibfnamefont {C.}~\bibnamefont {Feniou}}, \bibinfo
  {author} {\bibfnamefont {F.}~\bibnamefont {Debbasch}}, \bibinfo {author}
  {\bibfnamefont {A.}~\bibnamefont {Peruzzo}}, \ and\ \bibinfo {author}
  {\bibfnamefont {J.-P.}\ \bibnamefont {Piquemal}},\ }\href {\doibase
  10.1038/s41467-024-50065-x} {\bibfield  {journal} {\bibinfo  {journal}
  {Nature Communications}\ }\textbf {\bibinfo {volume} {15}},\ \bibinfo {pages}
  {5886} (\bibinfo {year} {2024})}\BibitemShut {NoStop}%
\bibitem [{\citenamefont {Buhrman}\ \emph {et~al.}(2024)\citenamefont
  {Buhrman}, \citenamefont {Folkertsma}, \citenamefont {Loff},\ and\
  \citenamefont {Neumann}}]{Buhrman2024statepreparation}%
  \BibitemOpen
  \bibfield  {author} {\bibinfo {author} {\bibfnamefont {H.}~\bibnamefont
  {Buhrman}}, \bibinfo {author} {\bibfnamefont {M.}~\bibnamefont {Folkertsma}},
  \bibinfo {author} {\bibfnamefont {B.}~\bibnamefont {Loff}}, \ and\ \bibinfo
  {author} {\bibfnamefont {N.~M.~P.}\ \bibnamefont {Neumann}},\ }\href
  {\doibase 10.22331/q-2024-12-09-1552} {\bibfield  {journal} {\bibinfo
  {journal} {{Quantum}}\ }\textbf {\bibinfo {volume} {8}},\ \bibinfo {pages}
  {1552} (\bibinfo {year} {2024})}\BibitemShut {NoStop}%
\bibitem [{\citenamefont {Bärtschi}\ and\ \citenamefont
  {Eidenbenz}(2022)}]{bartschi2022shortdepth}%
  \BibitemOpen
  \bibfield  {author} {\bibinfo {author} {\bibfnamefont {A.}~\bibnamefont
  {Bärtschi}}\ and\ \bibinfo {author} {\bibfnamefont {S.}~\bibnamefont
  {Eidenbenz}},\ }in\ \href {\doibase 10.1109/QCE53715.2022.00027} {\emph
  {\bibinfo {booktitle} {2022 IEEE International Conference on Quantum
  Computing and Engineering (QCE)}}}\ (\bibinfo {year} {2022})\ pp.\ \bibinfo
  {pages} {87--96}\BibitemShut {NoStop}%
\bibitem [{\citenamefont {Takahashi}\ and\ \citenamefont
  {Tani}(2016)}]{Takahashi2016collapse}%
  \BibitemOpen
  \bibfield  {author} {\bibinfo {author} {\bibfnamefont {Y.}~\bibnamefont
  {Takahashi}}\ and\ \bibinfo {author} {\bibfnamefont {S.}~\bibnamefont
  {Tani}},\ }\href {\doibase 10.1007/s00037-016-0140-0} {\bibfield  {journal}
  {\bibinfo  {journal} {computational complexity}\ }\textbf {\bibinfo {volume}
  {25}},\ \bibinfo {pages} {849} (\bibinfo {year} {2016})}\BibitemShut
  {NoStop}%
\bibitem [{\citenamefont {B{\"a}rtschi}\ and\ \citenamefont
  {Eidenbenz}(2019)}]{bartschi2019deterministic}%
  \BibitemOpen
  \bibfield  {author} {\bibinfo {author} {\bibfnamefont {A.}~\bibnamefont
  {B{\"a}rtschi}}\ and\ \bibinfo {author} {\bibfnamefont {S.}~\bibnamefont
  {Eidenbenz}},\ }in\ \href@noop {} {\emph {\bibinfo {booktitle} {Fundamentals
  of Computation Theory}}}\ (\bibinfo  {publisher} {Springer International
  Publishing},\ \bibinfo {address} {Cham},\ \bibinfo {year} {2019})\ pp.\
  \bibinfo {pages} {126--139}\BibitemShut {NoStop}%
\bibitem [{\citenamefont {Bergamaschi}\ and\ \citenamefont
  {Liu}(2025)}]{Bergamaschi2025Fault}%
  \BibitemOpen
  \bibfield  {author} {\bibinfo {author} {\bibfnamefont {T.}~\bibnamefont
  {Bergamaschi}}\ and\ \bibinfo {author} {\bibfnamefont {Y.}~\bibnamefont
  {Liu}},\ }in\ \href {\doibase 10.4230/LIPIcs.ITCS.2025.16} {\emph {\bibinfo
  {booktitle} {16th Innovations in Theoretical Computer Science Conference
  (ITCS 2025)}}},\ \bibinfo {series} {Leibniz International Proceedings in
  Informatics (LIPIcs)}, Vol.\ \bibinfo {volume} {325},\ \bibinfo {editor}
  {edited by\ \bibinfo {editor} {\bibfnamefont {R.}~\bibnamefont {Meka}}}\
  (\bibinfo  {publisher} {Schloss Dagstuhl -- Leibniz-Zentrum f{\"u}r
  Informatik},\ \bibinfo {address} {Dagstuhl, Germany},\ \bibinfo {year}
  {2025})\ pp.\ \bibinfo {pages} {16:1--16:9}\BibitemShut {NoStop}%
\bibitem [{\citenamefont
  {Aaronson}(2016)}]{aaronson2016complexityquantumstatestransformations}%
  \BibitemOpen
  \bibfield  {author} {\bibinfo {author} {\bibfnamefont {S.}~\bibnamefont
  {Aaronson}},\ }\href {https://arxiv.org/abs/1607.05256} {\enquote {\bibinfo
  {title} {The complexity of quantum states and transformations: From quantum
  money to black holes},}\ } (\bibinfo {year} {2016}),\ \Eprint
  {http://arxiv.org/abs/1607.05256} {arXiv:1607.05256 [quant-ph]} \BibitemShut
  {NoStop}%
\bibitem [{\citenamefont {Du}\ \emph {et~al.}(2024)\citenamefont {Du},
  \citenamefont {Liu},\ and\ \citenamefont
  {Ma}}]{du2024embeddedcomplexityquantumcircuit}%
  \BibitemOpen
  \bibfield  {author} {\bibinfo {author} {\bibfnamefont {Z.}~\bibnamefont
  {Du}}, \bibinfo {author} {\bibfnamefont {Z.-W.}\ \bibnamefont {Liu}}, \ and\
  \bibinfo {author} {\bibfnamefont {X.}~\bibnamefont {Ma}},\ }\href
  {https://arxiv.org/abs/2408.16602} {\enquote {\bibinfo {title} {Embedded
  complexity and quantum circuit volume},}\ } (\bibinfo {year} {2024}),\
  \Eprint {http://arxiv.org/abs/2408.16602} {arXiv:2408.16602 [quant-ph]}
  \BibitemShut {NoStop}%
\bibitem [{\citenamefont {Onorati}\ \emph {et~al.}(2023)\citenamefont
  {Onorati}, \citenamefont {Rouzé}, \citenamefont {França},\ and\
  \citenamefont {Watson}}]{onorati2023efficient}%
  \BibitemOpen
  \bibfield  {author} {\bibinfo {author} {\bibfnamefont {E.}~\bibnamefont
  {Onorati}}, \bibinfo {author} {\bibfnamefont {C.}~\bibnamefont {Rouzé}},
  \bibinfo {author} {\bibfnamefont {D.~S.}\ \bibnamefont {França}}, \ and\
  \bibinfo {author} {\bibfnamefont {J.~D.}\ \bibnamefont {Watson}},\
  }\href@noop {} {\enquote {\bibinfo {title} {Efficient learning of ground \&
  thermal states within phases of matter},}\ } (\bibinfo {year} {2023}),\
  \Eprint {http://arxiv.org/abs/2301.12946} {arXiv:2301.12946 [quant-ph]}
  \BibitemShut {NoStop}%
\bibitem [{\citenamefont {Yu}\ and\ \citenamefont
  {Wei}(2023)}]{yu2023learning}%
  \BibitemOpen
  \bibfield  {author} {\bibinfo {author} {\bibfnamefont {N.}~\bibnamefont
  {Yu}}\ and\ \bibinfo {author} {\bibfnamefont {T.-C.}\ \bibnamefont {Wei}},\
  }\href@noop {} {\enquote {\bibinfo {title} {Learning marginals suffices!}}\ }
  (\bibinfo {year} {2023}),\ \Eprint {http://arxiv.org/abs/2303.08938}
  {arXiv:2303.08938 [quant-ph]} \BibitemShut {NoStop}%
\bibitem [{\citenamefont {Geh\'er}\ \emph {et~al.}(2024)\citenamefont
  {Geh\'er}, \citenamefont {Crawford},\ and\ \citenamefont
  {Campbell}}]{Geher2024Tangling}%
  \BibitemOpen
  \bibfield  {author} {\bibinfo {author} {\bibfnamefont {G.~P.}\ \bibnamefont
  {Geh\'er}}, \bibinfo {author} {\bibfnamefont {O.}~\bibnamefont {Crawford}}, \
  and\ \bibinfo {author} {\bibfnamefont {E.~T.}\ \bibnamefont {Campbell}},\
  }\href {\doibase 10.1103/PRXQuantum.5.010348} {\bibfield  {journal} {\bibinfo
   {journal} {PRX Quantum}\ }\textbf {\bibinfo {volume} {5}},\ \bibinfo {pages}
  {010348} (\bibinfo {year} {2024})}\BibitemShut {NoStop}%
\bibitem [{\citenamefont {Çetin K.~Koç}\ and\ \citenamefont
  {Arachchige}(1991)}]{Cetin1991GF2}%
  \BibitemOpen
  \bibfield  {author} {\bibinfo {author} {\bibnamefont {Çetin K.~Koç}}\ and\
  \bibinfo {author} {\bibfnamefont {S.~N.}\ \bibnamefont {Arachchige}},\ }\href
  {\doibase https://doi.org/10.1016/0743-7315(91)90115-P} {\bibfield  {journal}
  {\bibinfo  {journal} {Journal of Parallel and Distributed Computing}\
  }\textbf {\bibinfo {volume} {13}},\ \bibinfo {pages} {118} (\bibinfo {year}
  {1991})}\BibitemShut {NoStop}%
\end{thebibliography}%

\end{document}